\documentclass[12pt]{amsart}
\usepackage{extsizes}
\usepackage[T1]{fontenc}
\usepackage[utf8]{inputenc}
\usepackage[titletoc]{appendix}
\usepackage{pxfonts}

\makeatletter
\def\blx@maxline{77}
\makeatother

\usepackage{blindtext}
\usepackage{amssymb}
\usepackage{amsfonts}
\usepackage{amsthm}
\usepackage{tikz}
\usepackage[margin=1.0in]{geometry}
\usepackage{caption}
\usepackage{color}
\usepackage{graphicx}
\usepackage{xcolor}
\usepackage{shadethm}
\usepackage{subfigure}
\usepackage[mathcal]{euscript}
\usepackage{epigraph}
\usepackage[T1]{fontenc}
\usepackage[utf8]{inputenc}
 
\usepackage{tgpagella}

\usepackage{enumerate}
\usepackage{verbatim}
\usetikzlibrary{trees}
\usetikzlibrary{decorations.markings}
\usetikzlibrary{arrows}
\usepackage{mathrsfs}
\usepackage{xparse}
\usepackage{diagrams}
\usepackage{tikz-cd}

\usetikzlibrary{positioning,arrows,patterns}
\usetikzlibrary{decorations.markings}
\usetikzlibrary{calc}
\tikzset{
  photon/.style={decorate, decoration={snake}, draw=black},
  fermion/.style={draw=black, postaction={decorate},decoration={markings,mark=at position .55 with {\arrow{>}}}},
  vertex/.style={draw,shape=circle,fill=black,minimum size=5pt,inner sep=0pt},
particle/.style={thick,draw=black},
particle2/.style={thick,draw=blue},
avector/.style={thick,draw=black, postaction={decorate},
    decoration={markings,mark=at position 1 with {\arrow[black]{triangle 45}}}},
gluon/.style={decorate, draw=black,
    decoration={coil,aspect=0}}
 }
\NewDocumentCommand\semiloop{O{black}mmmO{}O{above}}
{%
\draw[#1] let \p1 = ($(#3)-(#2)$) in (#3) arc (#4:({#4+180}):({0.5*veclen(\x1,\y1)})node[midway, #6] {#5};)
}

\newcommand{\calT}{\mathcal{T}}

\usepackage{url}
\usepackage{hyperref}
\hypersetup{colorlinks=true, linkcolor=blue, urlcolor=cyan} 

\newtheoremstyle{own}
  {3pt}
  {3pt}
  {\sffamily}
  {0pt}
  {\bfseries}
  {.}
  {5pt plus 1pt minus 1pt}
  {}

\theoremstyle{plain}
\newtheorem{thm}{Theorem}[subsection]

\theoremstyle{definition}
\newtheorem{defn}{Definition}[subsection]
\newtheorem{prop}{Proposition}[subsection]
\newtheorem{lem}{Lemma}[subsection]
\newtheorem{cor}{Corollary}[subsection]

\theoremstyle{definition}
\newtheorem{exe}{Exercise}[subsection]
\newtheorem{ex}{Example}[subsection]

\theoremstyle{remark}
\newtheorem{rem}{Remark}[subsection]

\newcommand{\R}{\mathbb{R}}
\newcommand{\A}{\mathcal{A}}
\newcommand{\N}{\mathbb{N}}

\newcommand{\Exp}{\textnormal{Exp}}
\newcommand{\sgn}{\textnormal{sgn}}

\newcommand{\dd}{{\mathrm{d}}}

\DeclareMathOperator{\Path}{\text{Path}}

\newcommand{\p}{\widehat{p}}
\newcommand{\x}{\widehat{x}}

\DeclareMathOperator{\ev}{ev}
\newcommand{\id}{\mathrm{id}}

\newcommand{\calA}{\mathcal{A}}
\newcommand{\calB}{\mathcal{B}}
\newcommand{\calH}{\mathcal{H}}
\newcommand{\calS}{\mathcal{S}}

\newcommand{\calL}{\mathcal{L}}

\newcommand{\calE}{\mathcal{E}}
\newcommand{\calP}{\mathcal{P}}

\def\gpd{\,\lower1pt\hbox{$\longrightarrow$}\hskip-.24in\raise2pt
               \hbox{$\longrightarrow$}\,}

\newcommand{\I}{\mathrm{i}}
\newcommand{\F}{\mathcal{F}}

\newcommand{\ee}{\textnormal{e}}

\newcommand{\calR}{\mathcal{R}}

\makeatletter

\pgfdeclareshape{genus}{
 \anchor{center}{\pgfpointorigin}
\backgroundpath{
    \begingroup
    \tikz@mode
    \iftikz@mode@fill

         \pgfpathmoveto{\pgfqpoint{-0.78cm}{-.17cm}}
     \pgfpathcurveto %
            {\pgfpoint{-0.35cm}{-.44cm}}
        {\pgfpoint{0.35cm}{-.44cm}}
        {\pgfpoint{.78cm}{-0.17cm}} 
     \pgfpathmoveto{\pgfqpoint{-0.78cm}{-0.17cm}}
     \pgfpathcurveto %
            {\pgfpoint{-0.25cm}{.25cm}}
        {\pgfpoint{.25cm}{.25cm}}
        {\pgfpoint{0.78cm}{-0.17cm}}
        \pgfusepath{fill}
        \fi

        \iftikz@mode@draw
     \pgfpathmoveto{\pgfqpoint{-1cm}{0cm}}
     \pgfpathcurveto %
            {\pgfpoint{-0.5cm}{-.5cm}}
        {\pgfpoint{0.5cm}{-.5cm}}
        {\pgfpoint{1cm}{0cm}}

     \pgfpathmoveto{\pgfqpoint{-0.75cm}{-0.15cm}}
     \pgfpathcurveto %
            {\pgfpoint{-0.25cm}{.25cm}}
        {\pgfpoint{.25cm}{.25cm}}
        {\pgfpoint{0.75cm}{-0.15cm}}
              \pgfusepath{stroke}
        \fi
        \endgroup
    }
    }

    \makeatother

\begin{document}

\title[Quantum Field Theory and Functional Integrals]{Quantum Field Theory and Functional integrals}
\author[N. Moshayedi]{Nima Moshayedi}
\address{Institut f\"ur Mathematik\\ Universit\"at Z\"uich\\ 
Winterthurerstrasse 190
CH-8057 Z\"urich}
\email[N.~Moshayedi]{nima.moshayedi@math.uzh.ch}

\maketitle

%
%

\begin{abstract}
These notes were inspired by the course ``Quantum Field Theory from a Functional Integral Point of View'' given at the University of Zurich in Spring 2017 by Santosh Kandel.
We describe Feynman's path integral approach to quantum mechanics and quantum field theory from a functional integral point of view, where the main focus lies in Euclidean field theory. The notion of Gaussian measure and the construction of the Wiener measure are covered. Moreover, we recall the notion of classical mechanics and the Schr\"odinger picture of quantum mechanics, where it shows the equivalence to the path integral formalism, by deriving the quantum mechanical propagator out of it. Additionally, we give an introduction to elements of constructive quantum field theory. 
\end{abstract}

\tableofcontents

\section{Introduction}
We want to give a review of quantum field theory, using perturbative methods with the notion of Feynman path integrals. In classical mechanics we consider an action functional\footnote{In the physics literature, it is common to denote the time-derivatives by ``dots'', i.e. $\frac{\dd}{\dd t}q(t)=\dot{q}(t)$.} $$S(q)=\int_{t_0}^{t_1} L(q(t),\dot{q}(t))\dd t,$$ 

\begin{center}
\begin{figure}[h!]

\begingroup%
  \makeatletter%
  \providecommand\color[2][]{%
    \errmessage{(Inkscape) Color is used for the text in Inkscape, but the package 'color.sty' is not loaded}%
    \renewcommand\color[2][]{}%
  }%
  \providecommand\transparent[1]{%
    \errmessage{(Inkscape) Transparency is used (non-zero) for the text in Inkscape, but the package 'transparent.sty' is not loaded}%
    \renewcommand\transparent[1]{}%
  }%
  \providecommand\rotatebox[2]{#2}%
  \ifx\svgwidth\undefined%
    \setlength{\unitlength}{114.30579976bp}%
    \ifx\svgscale\undefined%
      \relax%
    \else%
      \setlength{\unitlength}{\unitlength * \real{\svgscale}}%
    \fi%
  \else%
    \setlength{\unitlength}{\svgwidth}%
  \fi%
  \global\let\svgwidth\undefined%
  \global\let\svgscale\undefined%
  \makeatother%
  \begin{picture}(1,0.67732551)%
    \put(0,0){\includegraphics[width=\unitlength]{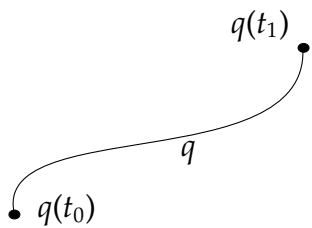}}%
    \put(0.56948513,0.21842524){\color[rgb]{0,0,0}\makebox(0,0)[lb]{\smash{$q$}}}%
    \put(0.09700054,0.01566911){\color[rgb]{0,0,0}\makebox(0,0)[lb]{\smash{$q(t_0)$}}}%
    \put(0.73540662,0.62451362){\color[rgb]{0,0,0}\makebox(0,0)[lb]{\smash{$q(t_1)$}}}%
  \end{picture}%
\endgroup%

\caption{The path of least action, i.e. the solution to $\delta S=0$, between two points $x=q(t_0)$ and $y=q(t_1)$ in space-time.}
\label{least_action_path}
\end{figure}
\end{center}

where $L(q,\dot{q})=\frac{1}{2}m\|\dot{q}\|^2-V(q)$ is called the \emph{Lagrangian} function of the paths $q\colon [t_0,t_1]\to\R^n$ with some function $V\in C^\infty(\R^n)$ depending on $q$, called the \emph{potential energy}. We denote by $\Path_{(x,y)}^{[t_0,t_1]}(\R^n)$ the space of all such paths with $q(t_0)=x$ and $q(t_1)=y$. By considering the methods of variational calculus, one can show that the solutions of the equation $\delta S=0$ for fixed endpoints (i.e. the extremal points of $S$) give us the classical trajectory of the particle with mass $m\in\R^+$. The equations following from $\delta S=0$ are called the Euler--Lagrange equations (EL), and they are exactly the equations of motion obtained from Newtonian mechanics. Netwon's equations of motion appear from the law $F=ma(t)=m\ddot{q}(t)$ (read it ``force equals mass times acceleration''). To see this, we recall that the \emph{momentum} in physics is given by $p=mv$, where $v$ denotes the velocity of the particle with mass $m$. Then, by the fact that $v=\dot{q}$, one considers the coordinates $\dot{q}=\frac{p}{m}$ and $\dot{p}=-\nabla V$, where $\nabla$ denotes the gradient operator. The \emph{Hamiltonian} approach considers the space with these coordinates to be the classical phase space (classical space of \emph{states}) given by $T^*\R^n\ni(q,p)$ endowed with a \emph{symplectic form}\footnote{we will not always write $\land$ between forms but secretly always mean the exterior product between them, i.e. for two differential forms $\alpha,\beta$, we have $\alpha\beta=\alpha\land \beta$.} given by $$\omega=\sum_{i=1}^n\dd q^{i}\dd p_i.$$ Moreover, one considers a \emph{total energy function} (or a \emph{Hamiltonian function}) $H(q,p)=\frac{\|p\|^2}{2m}+V$, where $V$ is again a potential energy function. In the physics literature, the first term of $H$ is called the \emph{kinetic energy}. This function is said to be \emph{Hamiltonian} if there is a vector field $X_H$ such that $$\iota_{X_H}\omega=-\dd H,$$ where $\iota$ denotes the \emph{contraction} map (also called \emph{interior derivative}). The vector field $X_H$ is called the \emph{Hamiltonian vector field} of $H$. In the case at hand, since $\omega$ is nondegenerate, every function is Hamiltonian and its Hamiltonian vector field is uniquely determined. For $H$ being the total energy function and the canonical symplectic form on the cotangent space, we get the following Hamiltonian vector field: A vector field on $T^*\R^n$ has the form general form $X=X^{i}\partial_{q^{i}}+X_i\partial_{p_i}$. Thus, applying the equation for being the Hamiltonian vector field of $H$ we get $-\dd H=X_i\dd q^{i}+X^{i}\dd p_i=\iota_X\omega$. Now since $\dd H=\partial_i V\dd q^{i}+\frac{p_i}{m}$, we get the coefficients of the vector field to be $X_i=-\partial_i V$ and $X^{i}=\frac{p_i}{m}$. Hence, we get the Hamiltonian vector field 
$$X_H=-\partial_i V\partial_{q_i}+\frac{p_i}{m}\partial_{p_i}.$$
Naturally, $X_H$ induces a \emph{Hamiltonian flow} $T^*\R^n\to T^*\R^n$.

\vspace{0.5cm}

An approach of quantization of the above is to associate to $T^*\R^n$ the space of square integrable functions $L^2(\R^n)$ on $\R^n$. The Hamiltonian flow can then be replaced by a linear map $$\ee^{\frac{\I}{\hbar}\widehat{H}}\colon L^2(\R^n)\to L^2(\R^n),$$ where $\widehat{H}:=-\frac{\hbar^2}{2m}\Delta+V$ denotes the \emph{Hamilton operator}, which is the canonical quantization of the classical Hamiltonian function, where $\Delta=\sum_{1\leq j\leq n}(\partial_{x^j})^2$ denotes the Laplacian. Note that the space of states is now given by a Hilbert space $\calH_0$ and the observables as operators on $\calH_0$. One can show that the action of this operator can be expressed as an integral of the form 
\[
\left(\ee^{\frac{\I}{\hbar}\widehat{H}}\psi\right)(x)=\int K(x,y)\psi(y)\dd y,
\]
for $\psi\in \calH_0$, where $K$ denotes the \emph{integral kernel} for the operator.  Feynman showed in \cite{Feynman1942} that this kernel (quantum mechanical \emph{propagator}) can be seen as a \emph{path integral}, which is given by 
\[
K(x,y)=\int_{\Path^{[t_0,t_1]}_{(x,y)}(\R^n)}\ee^{\frac{\I}{\hbar}S(q)}\mathscr{D}q.
\]
where $S$ denotes the action of the classical system and $\mathscr{D}$ a measure on the path space (see also figure \ref{QM_paths}). 

Since $\mathscr{D}$ is suppose to be a ``measure'' on an infinite-dimensional space, it is mathematically ill-defined. However, one can still make sense of such an integral in several ways; one of them is by considering its \emph{perturbative expansion} in formal power series with \emph{Feynman diagrams} as coefficients. This procedure is mathematically well-defined.
These notes are based on \cite{A,VB,NE,GJ,AG,BH,SJ,JL,HK,MP,RS,S,S2,T}.

\subsection*{Acknowledgements}
The author acknowledges partial support of the SNF grant No. 200020 172498/1 and by the Forschungskredit of the University of Zurich, grant no. FK-18-095. Moreover, the author wants to thank Santosh Kandel for sharing his lecture notes with him. 

\begin{center}
\begin{figure}[h!]

\begingroup%
  \makeatletter%
  \providecommand\color[2][]{%
    \errmessage{(Inkscape) Color is used for the text in Inkscape, but the package 'color.sty' is not loaded}%
    \renewcommand\color[2][]{}%
  }%
  \providecommand\transparent[1]{%
    \errmessage{(Inkscape) Transparency is used (non-zero) for the text in Inkscape, but the package 'transparent.sty' is not loaded}%
    \renewcommand\transparent[1]{}%
  }%
  \providecommand\rotatebox[2]{#2}%
  \ifx\svgwidth\undefined%
    \setlength{\unitlength}{131.29432958bp}%
    \ifx\svgscale\undefined%
      \relax%
    \else%
      \setlength{\unitlength}{\unitlength * \real{\svgscale}}%
    \fi%
  \else%
    \setlength{\unitlength}{\svgwidth}%
  \fi%
  \global\let\svgwidth\undefined%
  \global\let\svgscale\undefined%
  \makeatother%
  \begin{picture}(1,1.01600495)%
    \put(0,0){\includegraphics[width=\unitlength]{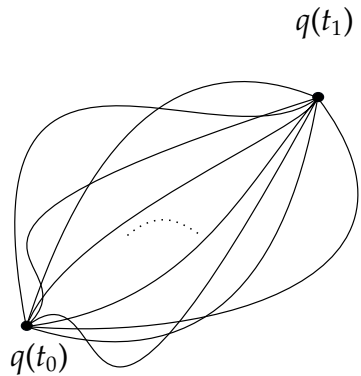}}%
    \put(-0.00369337,0.01329828){\color[rgb]{0,0,0}\makebox(0,0)[lb]{\smash{$q(t_0)$}}}%
    \put(0.81556436,0.97862504){\color[rgb]{0,0,0}\makebox(0,0)[lb]{\smash{$q(t_1)$}}}%
  \end{picture}%
\endgroup%

\caption{Illustration of the fact that all the paths between $x=q(t_0)$ and $y=q(t_1)$ are taken into account.}
\label{QM_paths}
\end{figure}
\end{center}

\part{A Brief Recap of Classical Mechanics}
\section{Newtonian Mechanics with examples}

Consider a particle of mass $m$ moving in $\R^n$. The position of a praticle $x=(x_1,...,x_n)$ is a vector in $\R^n$. More precisely $x(t)=(x_1(t),...,x_n(t))$ is the position of the particle at time $t$. Let $v(t)$ and $a(t)$ denote the velocity and the acceleration at time $t$ respectively. Then 
\begin{align}
v(t)&=\dot{x}(t)=(\dot{x}_1(t),...,\dot{x}_n(t)),\\
a(t)&=\ddot{x}(t)=(\ddot{x}_1(t),...,\ddot{x}_n(t)),
\end{align}
where $\dot{x}_i(t)=\frac{\dd}{\dd t}x_i(t)$ and $\ddot{x}_i(t)=\frac{\dd}{\dd t}\dot{x}_i(t)=\frac{\dd^2}{\dd t^2}x_i(t)$. We recall Newton's second law of motion:
\begin{equation}
\label{Newton}
m\ddot{x}(t)=F(x(t),\dot{x}(t)),
\end{equation}
where $F$ is a force acting on the particle with mass $m$. Hence, the trajectories of motion are given by solutions of \eqref{Newton}. We note that \eqref{Newton} is a system of second order ordinary differential equations and is nonlinear in general\footnote{Nonlinearity depends on the nature of $F$}.

\begin{ex}[The free particle on $\R^n$] 
The force $F=0$, which implies that \eqref{Newton} becomes $\ddot{x}=0$, hence the trajectories of motion are given by $x(t)=at+b$ with $a,b\in\R^n$.
\end{ex}

\begin{ex}[Harmonic oscillator in one dimension ($n=1$)]
The force is given by $F=-Kx$ (Hooke's law), where $K=\omega^2m$ is the so-called \textsf{spring constant}. Then the equation of motion becomes $m\ddot{x}+Kx=0$. Hence the trajectories of motion are given by 
\begin{equation}
x(t)=a\cos(\omega t)+b\sin(\omega t),
\end{equation}
with $a,b\in\R$.
\end{ex}

Thus, in Newtonian mechanics, we are interested in \textsf{solving} the \textsf{equation} \eqref{Newton}. One way to try to solve \eqref{Newton} would be to try to find conserved quantities which may help simplifying the problem.

\subsection{Conservation of Energy}
Assume that the force $F$ depends only on the position and it has the form $F=-\nabla V(x)$, where $V:\R^n\to\R$ is some function. Such a force $F$ is called a \textsf{conservative} force and $V$ is called the \textsf{potential energy} of $F$. Since \eqref{Newton} is a second order differential equation, the state space or \textsf{phase space} of \eqref{Newton} is $\R^{2n}=\{(x,v)\mid x,v\in\R^n\}$. Define the \textsf{total energy function} $E$ by 
\begin{equation}
\label{total_energy}
E(x,v)=\frac{1}{2}m\|v\|^2+V(x),
\end{equation}
where $\|v\|^2=\langle v,v\rangle$ with the standard inner product $\langle\enspace,\enspace\rangle$ on $\R^n$. The main significance of the total energy function is that it is conserved, meaning that its value along any trajectory of motion is constant. 
\begin{prop}
\label{conservation}
Suppose a particle moving on $\R^n$ satisfying Newton's law of the form \eqref{Newton}. Then 
\begin{equation}
\frac{\dd}{\dd t}E(x(t),\dot{x}(t))=0,
\end{equation}
along any trajectory $x(t)$ satisfying \eqref{Newton}.
\end{prop}
\begin{proof}
Along a solution $x(t)$ of \eqref{Newton} we have
\begin{align}
\begin{split}
\frac{\dd}{\dd t}E(x,v)&=\sum_{i=1}^n\frac{\partial E}{\partial x_i}\dot{x}_i+\sum_{i=1}^n\frac{\partial E}{\partial v_i}\dot{v}_i\\
&=\sum_{i=1}^n\frac{\partial V}{\partial x_i}v_i+m\sum_{i=1}^n v_i\dot{v_i}\\
&=(\nabla V+ma)v\\
&=(-F+ma)v\\
&=0,
\end{split}
\end{align}
\end{proof}
\begin{defn}[Constant of motion]
Let $f$ be a function on the phase space $\R^{2n}$. We say $f$ is a \textsf{constant of motion} if $\frac{\dd}{\dd t}f=0$ along $(x(t),\dot{x}(t))$, whenever $x(t)$ is a trajectory of motion.
\end{defn}
\begin{rem}
Constants of motion are conserved quantities.
\end{rem}
By proposition \ref{conservation}, the total energy is a constant of motion. Next, using an example, we investigate that the conservation of energy helps us to understand the solution of the equation of motion. Let us rewrite \eqref{Newton} in terms of first order equations
\begin{align}
\label{system_newton}
\begin{split}
\frac{\dd}{\dd t}x_i(t)&=v_i(t),\hspace{0.3cm}i=1,2,...,n\\
\frac{\dd}{\dd t}v_i(t)&=\frac{1}{m}F_i(x(t)),\hspace{0.3cm}i=1,2,...,n
\end{split}
\end{align}
For simplicity, assume $n=1$. Hence we have 
\begin{equation}
\label{first_ord}
\begin{split}
\frac{\dd}{\dd t}x(t)&=v(t),\\
\frac{\dd}{\dd t}v(t)&=\frac{1}{m}F(x(t))
\end{split}
\end{equation}
By conservation of energy, we know that $\frac{\dd}{\dd t}E(x,v)=0$ along $(x(t),v(t))$, whenever $(x(t),v(t))$ satisfy \eqref{first_ord}. Let $E(x(t),v(t))=E_0$. Then 
\begin{equation}
\frac{1}{2}m\dot{x}(t)^2+V(x(t))=E_0,
\end{equation}
and thus 
\begin{equation}
\label{sol}
\dot{x}(t)=\pm\sqrt{\frac{2(E_0-V(x(t))}{m}},
\end{equation}
which can be solved using separation of variables. From this example, we learned that the conservation of energy helps us simplify the given system of equation in the one dimensional case (previous example), we were able to reduce the second order equation into a first order equation and even solve the equation. A general ``mantra'' is: \textsf{the knowledge of conserved quantities helps to simplify the equation of motion}.

\section{Hamiltonian Mechanics}

\subsection{The general formulation}
Hamiltonian mechanics gives a systematic approach to understand conserved quantities. Consider a particle moving in $\R^n$. The idea is to think of the total energy as a function of position and momentum rather than a function of position and velocity:
\begin{equation}
\label{Hamilton}
H(x,p)=\frac{1}{2m}\sum_{j=1}^np_j^2+V(x),
\end{equation}
where $p_j=m\dot{x}_j$. Now the system of equations \eqref{Hamilton} can be written as 
\begin{equation}
\label{canonical_eq}
\begin{split}
\frac{\dd}{\dd t}x_i(t)&=x_i(t)=\frac{1}{m}p_i=\frac{\partial H}{\partial p_i}\\
\frac{\dd}{\dd t}p_i(t)&=m\frac{\dd}{\dd t}x_i(t)=-\frac{\partial V}{\partial x_i}=-\frac{\partial H}{\partial x_i}.
\end{split}
\end{equation}
The equations of \eqref{canonical_eq}, i.e. 
\begin{equation}
\dot{x}_i=\frac{\partial H}{\partial p_i},\hspace{1cm}\dot{p}_i=-\frac{\partial H}{\partial x_i}
\end{equation}
are called \textsf{Hamilton's equations}. 
\subsection{The Poisson bracket}
The previous observation implies that in Hamiltonian mechanics we consider the phase space to be 
\[
\R^{2n}:=\{(x,p)\mid x,p\in \R^n\}.
\]
It turns out that $\R^{2n}$ has more structures. If $f$ and $g$ are smooth functions on $\R^{2n}$, one can define the \textsf{Poisson bracket} 
\begin{equation}
\label{P_bracket}
\{f,g\}:=\sum_{j=1}^n\left(\frac{\partial f}{\partial x_j}\frac{\partial g}{\partial p_j}-\frac{\partial g}{\partial x_j}\frac{\partial f}{\partial p_j}\right).
\end{equation}
\begin{exe}
Verify that the Poisson bracket satisfies the following properties. Let $f,g$ and $h$ be smooth function on $\R^{2n}$. Then 
\begin{enumerate}
\item{$\{f,g\}=-\{g,f\}$}
\item{$\{f,g+ch\}=\{f,g\}+c\{f,h\},\hspace{0.5cm}c\in\R$}
\item{$\{f,gh\}=\{f,g\}h+\{f,h\}g$}
\item{$\{f,\{g,h\}\}=\{\{f,g\},h\}+\{g,\{f,h\}\}$ (\textsf{Jacobi identity})}
\end{enumerate}
\end{exe}
\begin{ex}
Let $p_j$ and $x_j$ be momentum and position observables as images of the following maps respectively.
\begin{align}
\begin{split}
(x,p)&\longmapsto p_j\\
(x,p)&\longmapsto x_j.
\end{split}
\end{align}
Then $\{x_i,x_j\}=0=\{p_i,p_j\}$ and $\{x_i,p_j\}=\delta_{ij}$, where $\delta_{ij}$ denotes the Kronecker delta.
\end{ex}
Next we will see that we can use the Poisson bracket to describe the conserved quantities. For that we need the following proposition.
\begin{prop}
\label{time_ev}
Let $f\in C^\infty(\R^{2n})$. Then 
\begin{equation}
\frac{\dd}{\dd t}f=\{f,H\}
\end{equation}
along a solution of Hamilton's equations $\{(x(t),p(t))\}\subset \R^{2n}$.
\end{prop}
\begin{proof}
Exercise.
\end{proof}
\begin{cor}
Let $f\in C^\infty(\R^{2n})$. Then $f$ is conserved along solutions of Hamilton's equations iff $$\{f,H\}=0.$$
\end{cor}
\begin{proof}
By proposition \ref{time_ev} $\frac{\dd}{\dd t}f=\{f,H\}$ along solutions $(x(t),p(t))$ of Hamilton's equations. By definition, $f$ is conserved if $\frac{\dd}{\dd t}f=0$ iff $\{f,H\}=0$.
\end{proof}
\begin{rem}
Given any $f\in C^\infty(\R^{2n})$, we can define Hamilton's equations by 
\begin{equation}
\label{Hamiltonian_function}
\begin{split}
\dot{x}_i&=\frac{\partial f}{\partial p_i}\\
\dot{p}_i&=-\frac{\partial f}{\partial x_i}\\
i&=1,2,...,n
\end{split}
\end{equation}
\end{rem}
For the next remarks we assume familliarity with basic differential geometry notions such as vector fiedls, differential forms etc.
\begin{rem}
$\R^{2n}$ has a canonical symplectic structure $\omega=\sum_{i=1}^ndp_i\land dx_i$. Given $f\in C^\infty(\R^{2n})$ there exists a vector field $X_f$, called the \textsf{Hamiltonian vector field of $f$}, defined by 
\begin{equation}
\label{Hamiltonian_vector_field}
\omega(X_f,\enspace)=-\dd f
\end{equation}
The flow of $X_f$ is given by solutions of \eqref{Hamiltonian_function}. In this case, one can check that 
\begin{equation}
\label{symp_poisson}
\{f,g\}=\omega(X_f,X_g).
\end{equation}
This means that if $(N,\omega)$ is a symplectic manifold, then we can define the Poisson bracket of $f,g\in C^\infty(N)$ using \eqref{symp_poisson}.
\end{rem}
\begin{rem}
Let $f\in C^\infty(\R^{2n})$ and $X_f$ be the corresponding Hamiltonian vector field. The flow of $X_f$ (or in other words the solutions of \eqref{Hamiltonian_function}) defines one-parameter diffeomorphisms
\begin{align}
\begin{split}
\Phi^t_{X_f}:\R^{2n}&\longrightarrow \R^{2n}\\
(x,p)&\longmapsto \Phi^t_{X_f}(x,p)=(x(t),p(t)),
\end{split}
\end{align}
where $(x(t),p(t))$ satisfy Hamilton's equations with $x(0)=x$ and $p(0)=p$. Then, assuming that the flow is complete, we get 
\begin{enumerate}
\item{$\Phi^t_{X_f}$ preserves $\omega$ (i.e. $(\Phi^t_{X_f})^*\omega=\omega$). Such a map is called a \textsf{symplectomorphism}.}
\item{$\Phi^t_{X_f}$ preserves the volume form $v=\dd x_1\dd x_2\dotsm \dd x_n\dd p_1\dd p_2\dotsm \dd p_n$ (i.e. $(\Phi^t_{X_f})^*v=v$). This is known as \textsf{Liouville's theorem}.}
\end{enumerate}
\end{rem}
\begin{rem}
Let $f,g\in C^\infty(\R^{2n})$. Then $f$ is conserved along the solutions of Hamilton's equations of $g$ iff $\{f,g\}=0$ (This is an instance of \textsf{Noether's theorem}).
\end{rem}

\section{Lagrangian Mechanics}
There are two important points in this formalism:
\begin{itemize}
\item{Mechanics on a configuration space.}
\item{Basic theorems are invariant under actions of diffeomorphisms of the configuration space. It is useful to compute conserved quantities.}
\end{itemize}
\subsection{Lagrangian system}
Let $M$ be a smooth manifold (we will usually consider $M=\R^n$). A \textsf{Lagrangian system} with configuration space $M$ consists of a smooth real valued function $L:TM\times \R\to \R$, where $TM$ denotes the tangent bundle of $M$ (e.g. if $M=\R^n$, then $T\R^n=\R^n\times\R^n=\{(x,v)\}$). $L$ is called the \textsf{Lagrangian function} or simply \textsf{Lagrangian}. Lagrangian mechanics uses special ideas such as the \textsf{least action principle} from calculus of variation. Let $x_0,x_1\in M$ and $P(M,x_0,x_1):=\{\gamma:[t_0,t_1]\subset\R\to M\mid \gamma(t_0)=x_0,\gamma(t_1)=x_1\}$, which is the space of paramterized paths joining $x_0$ to $x_1$.
\begin{defn}[Action functional]
The \textsf{action functional} $S:P(M,x_0,x_1)\to\R$ of the Lagrangian system $(M,L)$ is defined by 
\begin{equation}
\label{action}
S(\gamma)=\int_{t_0}^{t_1}L(\gamma(t),\dot{\gamma}(t),t)\dd t.
\end{equation}
\end{defn}
From now on we take $M=\R^n$. We are interested in understanding ``critical points'' of $S$. Let $h:[t_0,t_1]\to \R^n$ be such that $\gamma+h\in P(\R^n,x_0,x_1)$ and $h(t_0)=h(t_1)=0$. We think of $h$ as a small variation of $\gamma\in P(\R^n,x_0,x_1)$. Then, if we change $\gamma(t)$ by $h$, we get 
\begin{equation}
S(\gamma+\varepsilon h)=\int_{t_0}^{t_1}L(\gamma(t)+\varepsilon h(t),\dot{\gamma}(t)+\varepsilon h(t),t)\dd t,
\end{equation}
which needs to be extremal with respect to the parameter $\varepsilon$. Hence 
\begin{equation}
\frac{\dd}{\dd\varepsilon}S(\gamma+\varepsilon h)=\int_{t_0}^{t_1}\left(\frac{\partial L}{\partial \gamma}h+\frac{\partial L}{\partial \dot{\gamma}}\dot{h}\right)\dd t=0.
\end{equation}
For the second part, we use integration by parts, which gives
\begin{equation}
\int_{t_0}^{t_1}\frac{\partial L}{\partial \dot{\gamma}}\dot{h}(t)\dd t=\underbrace{\frac{\partial L}{\partial \dot{\gamma}}h\Big|_{t_0}^{t_1}}_{=0}-\int_{t_0}^{t_1}\frac{\dd}{\dd t}\frac{\partial L}{\partial \dot{\gamma}}h(t)\dd t.
\end{equation}
The last term remains and by the product rule we get 
\begin{equation}
\int_{t_0}^{t_1}\left(\frac{\partial L}{\partial \gamma}-\frac{\dd}{\dd t}\frac{\partial L}{\partial \dot{\gamma}}\right)h(t)\dd t=0.
\end{equation}
\begin{defn}[Extremal point/Critical point]
An \textsf{extremal (or critical) point} of $S$ is some $x\in P(\R^n,x_0,x_1)$ such that 
\begin{equation}
\int_{t_0}^{t_1}\left(\frac{\partial L}{\partial x}-\frac{\dd}{\dd t}\frac{\partial L}{\partial \dot{x}}\right)h\dd t=0
\end{equation}
along $x$ for all paths $h$ such that $h(t_0)=h(t_1)=0$.
\end{defn}
\begin{thm}
A path $x \in P(\R^n,t_0,t_1)$ is an extremal of $S$ iff along $x$ we have
\begin{equation}
\label{EL1}
\frac{\partial L}{\partial x}-\frac{\dd}{\dd t}\frac{\partial L}{\partial \dot{x}}=0
\end{equation}
\end{thm}
The proof for this theorem follows from the following lemma. 
\begin{lem}
Let $f:[t_0,t_1]\to\R^n$ be a continuous path and 
\begin{equation}
\int_{t_0}^{t_1}fh\dd t=0
\end{equation}
for all continuous $h:[t_0,t_1]\to\R^n$ such that $h(t_0)=h(t_1)=0$. Then $f\equiv0$ on $[t_0,t_1]$.
\end{lem}
\begin{proof}
For simplicity assume $n=1$, i.e. $f:[t_0,t_1]\to \R$ and $h:[t_0,t_1]\to\R$. By contradiction assume there is some $t\in[t_0,t_1]$ such that $f(t)>0$. Then by continuity there is some $\delta >0$ such that $f>0$ on $(t-\delta,t+\delta)$. Let $h$ be a continuous function on $[t_0,t_1]$ such that $h$ vanishes outside $(t-\delta,t+\delta)$ but $h>0$ on $(t-\delta/2,t+\delta/2)$. Then 
\[
\int_{t_0}^{t_1}fh\dd t>0,
\]
which is a contradiction.
\end{proof}
\begin{defn}[Euler-Lagrange equations]
The equations 
\begin{equation}
\label{EL2}
\frac{\partial L}{\partial x}-\frac{\dd}{\dd t}\frac{\partial L}{\partial \dot{x}}=0
\end{equation}
are called the \textsf{Euler-Lagrange (EL) equations} of $S(x)$.
\end{defn}
\begin{cor}
A path $x \in P(\R^n,x_0,x_1)$ is an extremal of $S$ iff it satisfies the Euler-Lagrange equations.
\end{cor}
\subsection{Hamilton's least action principle}
Recall that we defined the total energy function by 
\[
E(x,v)=\frac{1}{2}m\|v\|^2+V(x),
\]
where the first term is the kinetic energy and the second the potential energy. 
\begin{thm}
Define $L(\gamma(t),\dot{\gamma}(t),t)=\frac{1}{2}m\|\dot{\gamma}(t)\|^2-V(\gamma(t))$. Then an extremal path $\gamma(t)$ of $S$ solves the system \eqref{system_newton}.
\end{thm}
\begin{proof}
Exercise.
\end{proof}
\begin{rem}
Even though only an extremal path of $S$ is involved here, it is called \textsf{Hamilton's least action principle}.
\end{rem}
Next we briefly investigate how Hamilton's equations and the EL equations are related. 

\section{The Legendre Transform}
Let $f$ be a convex function, i.e. $f''(x)>0$. Let $p\in \R$ and define $g(x)=px-f(x)$. Then $g'(x)=p-f'(x)$. Since $f$ is convex (i.e. $f'$ is increasing), there is a unique $x_0$ such that $g(x_0)=0$. We denote this $x_0$ by $x(p)$. Moreover, $f''(x)>0$ implies$g''(x)<0$, and hence $g$ has a maximum at $x(p)$. In this case the \textsf{Legendre transform} of $f$ is defined by 
\begin{equation}
\label{legendre}
\mathcal{L}f(p)=\max_xg(x)=\max_x(px-f(x)).
\end{equation}
\begin{ex}
Let $f(x)=x^2$, then $\mathcal{L}f(p)=\frac{1}{4}p^2$.
\end{ex}
\begin{ex}
Let $f(x)=\frac{1}{2}x^2$, then $\mathcal{L}f(p)=\frac{1}{2}p^2$.
\end{ex}
More generally, let $V$ be a finite dimensional vector space and $V^*$ be its dual and $f:V\to \R$ be a function. Then $\mathcal{L}f:V^*\to\R$ is defined by 
\begin{equation}
\label{legendre2}
\mathcal{L}f(p)=\max_{x\in V}(p(x)-f(x)),
\end{equation}
where $p(x)$ is the pairing between $x\in V$ and $p\in V^*$. If $f$ is convex, then $\mathcal{L}f$ exists.
\begin{exe}
Show that if $f$ is convex, then so is $\mathcal{L}f$. Moreover, show that $\mathcal{L}(\mathcal{L}f)=f$.
\end{exe}
\begin{ex}
Let $A$ be an $n\times n$ positiv definite matrix and $f:\R^n\to \R$, $f(x)=\frac{1}{2}\langle Ax,x\rangle$, where $\langle\enspace,\enspace\rangle$ is the standard inner product on $\R^n$. Then 
\[
\mathcal{L}f(\omega)=\frac{1}{2}\langle A^{-1}\omega,\omega\rangle.
\]
\end{ex}
Let us now consider a Lagrangian system $(\R^n,L)$, i.e. $L:\underbrace{\R^n\times\R^n\times\R}_{\ni(x,v,t)}\to\R$. Let $H(x,p,t)$ be the Legendre transform of $L$ in $v$-direction.
\begin{thm}
The system of EL equations are equivalent to Hamilton's equation with $H$ defined as above.
\end{thm}
\begin{proof}
Exercise.
\end{proof}

\part{The Schr\"odinger Picture of Quantum Mechanics}
Classical physics is inconsistent at the level of atoms and molecules. For example, the hydrogen atom which is composed of two particles a proton of charge $+e$ and an electron of charge $-e$. If we follow classical mechanics, then the charged electron would radiate energy continuously causing the atom to collapse. But this is not true. We need \textsf{quantum mechanics} to explain the stability of molecules and atoms.\footnote{We refer to a standard physics book on quantum mechanics for the motivation leading to postulates of quantum mechanics.} 
\section{Postulates of Quantum Mechanics}
\subsection{First Postulate}
The pure states of a quantum mechanical system are rays in a Hilbet space $\mathcal{H}$, i.e. one dimensional subspaces of $\mathcal{H}$. The Hilbert space $\mathcal{H}$ is called the \textsf{space of states}. Define 
\[
P\mathcal{H}:=(\mathcal{H}\setminus\{0\})/(\mathbb{C}\setminus\{0\}).
\]
Let $\phi,\psi\in\mathcal{H}\setminus\{0\}$. We say $\phi\sim\psi$ iff there is an $\alpha\in\mathbb{C}\setminus\{0\}$ such that $\phi=\alpha\psi$. Then $P\mathcal{H}$ is the set of equivalence classes with respect to this equivalence relation.
\begin{lem}
There is a canonical bijection 
\[
\Big\{\text{1-dimensional subspaces of $\mathcal{H}$}\Big\}\longleftrightarrow P\mathcal{H}.
\]
\end{lem}
\begin{proof}
Let $L$ be a one dimensional subspace of $\mathcal{H}$, and $\phi\in L$ such that $\phi\not=0$. Define 
\[
\beta(L)=[\phi],\hspace{1cm}[\phi]\in P\mathcal{H}.
\]
Let us check that $\beta$ is well defined. Let $\psi\in L\setminus\{0\}$. Then there is an $\alpha\in \mathbb{C}\setminus\{0\}$ such that $\psi=\alpha\phi$ (since $L$ is a one dimensional subspace). Thus $[\psi]=[\phi]$. This shows that $\beta$ is well defined. One can easily check that $\beta$ is a bijection.
\end{proof}
\begin{rem}
More precisely, the \textsf{space of pure states} is $P\mathcal{H}$.
\end{rem}
From now on when we say a state we mean $\psi\in\mathcal{H}$ such that $\|\psi\|=1$ (these are called \textsf{normalized states}). The concept of a state as a ray in $\mathcal{H}$ leads to the probability interpretation in quantum mechanics. This means that if a physical system is in the state $\psi$, then the probability that it is in the state $\phi$ is $\vert\langle \psi,\phi\rangle\vert^2$. Since we assume $\|\phi\|=1$, $\|\psi\|=1$, clearly $0\leq \vert\langle\psi,\phi\rangle\vert^2\leq 1$.
\subsection{Second Postulate}
Quantum mechanical observables are \textsf{self adjoint} operators on $\mathcal{H}$. Let $A$ be an observable. Then the \textsf{expectation} of $A$ in the state $\psi$ is defined as 
\begin{equation}
\label{expectation}
\langle A\rangle_\psi=\frac{\langle A\psi,\psi\rangle}{\langle\psi,\psi\rangle}
\end{equation}

\subsection{Third Postulate}
The Hamiltonian $\widehat{H}$ is the infinitesimal generator of the \textsf{unitary} group $U(t)=\ee^{-\frac{\I}{\hbar}t\widehat{H}}$. It describes the dynamics of the system. Let $\psi$ be a state. Then time evolution is described by the \textsf{Schr\"odinger equation}
\begin{equation}
\label{Schroedinger}
\I\hbar\frac{\dd}{\dd t}\psi(t)=\widehat{H}\psi(t).
\end{equation}
Using an Ansatz for equation \eqref{Schroedinger}, we get a solution of the form $\psi(t)=\ee^{-\frac{\I}{\hbar}t\widehat{H}}\psi(0)$. In the so-called \textsf{Heisenberg picture}, the Schr\"odinger equation takes the form 
\begin{equation}
\label{Heisenberg_picture}
\hbar\frac{\dd}{\dd t}A(t)=[\I\widehat{H},A(t)],
\end{equation}
where $A$ is an observable and $[\enspace,\enspace]$ is the commutator of operators, defined by $[A,B]=AB-BA
$.
\begin{lem}
Let $\phi(t)$ and $\psi(t)$ be solutions of \eqref{Schroedinger}, such that $\phi(0)=\phi$ and $\psi(0)=\psi$. Then 
\[
\langle \phi(t),\psi(t)\rangle=\langle\phi,\psi\rangle,\hspace{1cm}\forall t
\]
\end{lem}
\begin{proof}
We have 
\begin{align}
\phi(t)&=\ee^{-\frac{\I}{\hbar}t\widehat{H}}\phi(0),\\
\psi(t)&=\ee^{-\frac{\I}{\hbar}t\widehat{H}}\psi(0),
\end{align}
and since $\ee^{-\frac{\I}{\hbar}t\widehat{H}}$ is a unitary operator, we get the result.
\end{proof}
\subsection{Summary of CM and QM}
The following should summarize the differences of classical and quantum mechanics. 
\begin{center}
    \begin{tabular}{ | l | p{5cm} | p{5cm} |}
    \hline
     & {\bf Classical Mechanics} & {\bf Quantum Mechanics}  \\ \hline
    State space & $T^*M$ (cotangent bundle) & $P\mathcal{H}$, where $\mathcal{H}$ is a Hilbert space  \\ \hline
    Observables & $C^\infty(T^*M)$ (smooth functions on the cotagent bundle) & Self adjoint operators on $\mathcal{H}$  \\ \hline
    Dynamics & Described by Hamilton's equation associated to a Hamiltonian function $H\in C^\infty(T^*M)$ & Described by the Schr\"odinger equation associated to a quantum Hamiltonian operator $\widehat{H}$: \newline $\I\hbar\frac{\dd}{\dd t}\psi(t)=\widehat{H}\psi(t)$ \\
    \hline
    \end{tabular}
\end{center}

Next, we will define basic notations and concepts used to define quantum mechanical systems.\footnote{This is fairly standard. One can find them in any text book about functional analysis.}

\section{Elements of Functional Analysis}
Let $\mathcal{H}$ be a Hilbert space (we always assume it is seperable, i.e. there exists a basis). An operator in $\mathcal{H}$ is a pair $(A,D(A))$ where $D(A)$ is a subspace of $\mathcal{H}$, called the \textsf{domain} of $A$, and $A:D(A)\to \mathcal{H}$ is a linear map. We can always assume that $D(A)$ is dense in $\mathcal{H}$.
\begin{defn}[Bounded operator]
A linear map $A:D(A)\to\mathcal{H}$ is called \textsf{bounded} if there exists some $\varepsilon>0$ such that for all $\psi \in D(A)$
\[
\|A\psi\|\leq \varepsilon\|\psi\|.
\]
Otherwise, we say $A$ is \textsf{unbounded}.
\end{defn}
\begin{rem}
If $A$ is bounded, $A$ can be always extended to a bounded operator $\widetilde A:\mathcal{H}\to \mathcal{H}$. Hence, when we talk about bounded operator, we always consider $A:\mathcal{H}\to\mathcal{H}$. 
\end{rem}
Let $A:\mathcal{H}\to\mathcal{H}$ be a bounded operator. Then there is a unique operator $A^*:\mathcal{H}\to\mathcal{H}$ such that 
\[
\langle \phi,A\psi\rangle=\langle A^*\phi,\psi\rangle,\hspace{0.5cm}\forall \phi,\psi\in \mathcal{H}.
\]
\begin{defn}[Adjoint/Self adjoint operator]
We call $A^*$ the \textsf{adjoint} of $A$. Moreover, a bounded operator $A:\mathcal{H}\to\mathcal{H}$ is called \textsf{self adjoint} if $A^*=A$.
\end{defn}
\begin{ex}
Let $\mathcal{H}=L^2([0,1])$ and $X:\mathcal{H}\to\mathcal{H}$, $(Xf)(x)=xf(x)$. Then 
\[
\|Xf\|^2=\int_0^1x^2\vert f(x)\vert^2\dd x\leq \int_0^1\vert f(x)\vert^2\dd x=\|f\|^2,
\]
which implies that $\|Xf\|\leq \|f\|$ and thus $X$ is bounded. Let now $f,g\in L^2([0,1])$. Then 
\begin{equation}
\langle f,Xg\rangle=\int_0^1\overline{f(x)}xg(x)\dd x
=\int_0^1\overline{xf(x)}g(x)\dd x
=\langle Xf,g\rangle,
\end{equation}
and thus $A^*=A$. Hence $A$ is self adjoint. 
\end{ex}
\subsection{Unbounded operators}
\begin{ex}
Let $\mathcal{H}=L^2(\R)$. Let $X$ be the multiplication operator like before and define its domain $D(X)=\{\phi\in L^2(\R)\mid x\phi(x)\in L^2(\R)\}$. We claim that 
\begin{enumerate}
\item{$D(X)$ is dense in $L^2(\R)$.}
\item{$X$ is unbounded.}
\end{enumerate}
Let $\phi\in L^2(\R)$. Define $\phi_n=\phi\chi_{[-n,n]}$, where $\chi$ denotes the characteristic function. Then it is clear that $x\phi_n\in L^2(\R)$ and by the dominated convergence theorem $\phi_n\xrightarrow{n\to\infty} \phi$, in $L^2(\R)$. This proves $(1)$. To see that $X$ is unbounded, consider $\phi_n=\frac{1}{\sqrt{n}}\chi_{[0,n]}$, then $\|\phi_n\|=1$ for all $n$, but 
\[
\|X\phi_n\|^2=\frac{1}{n}\int_0^1x^2\dd x=\frac{n^2}{3}\xrightarrow{n\to\infty}\infty.
\]
Thus $X$ is unbounded, proving $(2)$.
\end{ex}

\subsection{Adjoint of an unbounded operator}
Let $A$ be an unbounded operator in $\mathcal{H}$ with domain $D(A)$. Define $D(A^*)=\{\phi\in\mathcal{H}\mid\langle \phi,A \enspace\rangle \text{ is a bounded linear functional on $D(A^*)$}\}$. Using \textsf{Riesz's theorem}, one can show that if $\phi\in D(A^*)$, then there is a unique $\psi\in\mathcal{H}$ such that 
\[
\langle \psi,\chi\rangle=\langle \phi,A\chi\rangle,\hspace{0.5cm}\forall\chi\in D(A).
\]
We define $A^*\phi=\psi$.
\begin{defn}[Symmetric operator]
Let $A$ be an unbounded operator with $D(A)$. We say $A$ is \textsf{symmetric} if 
\[
\langle \phi,A\psi\rangle=\langle A\phi,\psi\rangle,\hspace{0.5cm}\forall \phi,\psi\in D(A).
\]
Moreover, $A$ is self adjoint if $D(A)=D(A^*)$ and $A^*\phi=A\phi$.
\end{defn}
\begin{exe}
Show that if $A$ is symmetric, then $D(A)\subseteq D(A^*)$. Hence, $A$ is self adjoint iff $A$ is symmetric and $D(A)=D(A^*)$.
\end{exe}
\begin{exe}
Let $\mathcal{H}=L^2(\R)$ and $V:\R\to \R$ be a measurable map. Define the domain $$D(V(x))=\{\phi\in L^2(\R)\mid V(x)\phi(x)\in L^2(\R)\}$$ for the operator 
\begin{align}
V(X):D(V(X))&\longrightarrow L^2(\R)\\
\phi&\longmapsto V(x)\phi
\end{align}
\begin{prop}
$V(X)$ is self adjoint.
\end{prop}
\begin{proof}
We need to check that $D(V(X))$ is dense in $L^2(\R)$. $V(X)$ is symmetric and $D(V(X))=D(V(X)^*)$. It is easy to check that $D(V(X))$ is dense in $L^2(\R)$. Since $V$ is a real valued function, $V(X)$ is symmetric as well. We only need to show $D(V(X)^*)\subseteq D(V(X))$. For this let $\phi\in D(V(X)^*)$. We want to show that $V(x)\phi(x)\in L^2(\R)$. Since $\phi\in D(V(X)^*)$, we get that $\psi\mapsto \langle \phi,V(X)\psi\rangle$ is a bounded linear functional on $D(V(X))$. In fact, it can be extended to a bounded linear functional on $L^2(\R)$ (since $D(V(X))$ is dense). Hence by Riesz's theorem there is a unique $\chi\in L^2(\R)$ such that 
\begin{equation}
\langle\chi,\psi\rangle=\langle\phi,V(X)\psi\rangle,\hspace{0.5cm}\forall\psi\in L^2(\R),
\end{equation}
thus 
\begin{equation}
\int_\R\overline{\chi(x)}\psi(x)\dd x=\int_\R\overline{\phi(x)}V(x)\psi(x)\dd x,\hspace{0.5cm}\forall\psi\in L^2(\R),
\end{equation}
and hence
\begin{equation}
\int_\R\overline{\chi(x)}\psi(x)\dd x=\int_\R\overline{\phi(x)V(x)}\psi(x)\dd x,\hspace{0.5cm}\forall\psi\in L^2(\R).
\end{equation}
which shows that $\chi=V(x)\phi$ a.e., and therefore $V(x)\phi\in L^2(\R)$. Hence $\phi\in D(V(X))$.
\end{proof}
\end{exe}
Similarly one can show that the operator $P$, defined by $P\psi(x)=-\I\hbar\frac{\dd}{\dd x}\psi(x)$, is a self adjoint operator with domain
\[
D(P)=\{\psi\in L^2(\R)\mid k\widehat{\psi}(k)\in L^2(\R)\},
\]
where 
\[
\widehat{\psi}(k)=\frac{1}{2\pi}\int_\R\ee^{-\I kx}\psi(x)\dd x
\]
is the \textsf{Fourier transform} of $\psi$. Next we mention two techincal results without proof.

\begin{thm}[Spectral theorem/ Functional calculus]
Let $A$ be a self adjoint operator on $\mathcal{H}$. Let $L(\mathcal{H})$ denote the space of bounded linear operators in $\mathcal{H}$. Then, there is a unique map
\[
\Big\{\text{Bounded measurable functions on $\R$}\Big\}\xrightarrow{\widehat{\phi}}L(\mathcal{H}),
\]
such that 
\begin{enumerate}
\item{$\widehat{\phi}$ is linear and $\widehat{\phi}(fg)=\widehat{\phi}(f)\widehat{\phi}(g)$ for all bounded measureable functions $f,g$ on $\R$.}
\item{$\widehat{\phi}(f)=(\widehat{\phi}(f))^*$}
\item{$\|\widehat{\phi}(h)\|\leq \|h\|_\infty$}
\item{If $h_n\xrightarrow{n\to\infty} x$ and for all $n$ we have $\vert h_n(x)\vert\leq \vert x\vert$, then for all $\psi\in D(A)$ $$\widehat{\phi}(h_n)\psi\xrightarrow{n\to\infty}A\psi.$$}
\item{$A\psi=\lambda\psi$ for $\lambda\in\mathbb{C}$.}
\end{enumerate}
\end{thm}
We can use this theorem to produce bounded operators from a self adjoint operator, e.g. let $f(x)=\ee^{\I tx}$. We can see that $f$ is bounded and measurable. Hence we can talk about $f(A)=\ee^{\I tA}$ as a bounded linear operator on $\mathcal{H}$.
\begin{thm}[Stone's theorem]
\label{Stone}
Let $A$ be a self adjoint operator on $\mathcal{H}$. Define $U(t)=\ee^{\I tA}$. Then 
\begin{enumerate}
\item{$U(t)$ is a unitary operator: $$\langle U(t),\phi,U(t)\psi\rangle=\langle \phi,\psi\rangle$$ for all $\phi,\psi\in\mathcal{H}$. Moreover, $U(t+s)=U(t)\circ U(s)$.}
\item{For $\phi\in\mathcal{H}$ and $t\to t_0$ we have that $U(t)\phi\to U(t_0)\phi$ in $\mathcal{H}$ (strong convergence)}
\item{The limit $\lim_{t\to 0}\frac{U(t)\psi-\psi}{t}$ exists in $\mathcal{H}$ for all $\psi\in D(A)$ and $$\lim_{t\to 0}\frac{U(t)\psi-\psi}{t}=\I A\psi.$$
(formally, this means $\frac{\dd}{\dd t}U(t)=\I A$)}
\item{Let $\psi\in\mathcal{H}$ such that the limit $\lim_{t\to 0}\frac{U(t)\psi-\psi}{t}$ exists. Then $\psi\in D(A)$.}
\end{enumerate}
\end{thm}
Moreover, if $U(t)$, for $t\in\R$, is a family of unitary operators such that $(1)$ and $(2)$ hold, then $U(t)=\ee^{\I t A}$ for some self adjoint operator $A$.
\begin{defn}[Strongly continuous one parameter unitary group]
A familiy $U(t)$ satisfying (1) and (2) of theorem \ref{Stone} is called \textsf{strongly continuous one parameter unitary group} and $A$ is called the infinitesimal generator.
\end{defn}

\begin{defn}[Resolvent]
Let $A$ be an operator with domain $D(A)$ and let $\lambda\in\mathbb{C}$. We say that $A$ is in the \textsf{resolvent} set $\rho(A)$ of $A$ if 
\begin{enumerate}
\item{$\lambda I-A\colon D(A)\to \mathcal{H}$ is bijective,}
\item{$(\lambda I-A)^{-1}$ is a bounded operator.}
\end{enumerate}
\end{defn}

\begin{defn}[Spectrum]
The spectrum $\sigma(A)$ of $A$ is defined by $\sigma(A):=\mathbb{C}\setminus \rho(A)$.
\end{defn}

One can actually check that if $\lambda$ is an eigenvalue of $A$, then $\lambda\in\sigma(A)$. We call the set of eigenvalues of $A$ the \textsf{point spectrum} of $A$.

Let $\A [0,1]$ denote the set of absolutely continuous $L^2$-functions on $[0,1]$. 

\begin{ex}
Consider the operator $T=\I\frac{\dd}{\dd x}$ on $L^2([0,1])$ with domain $D(T)=\A[0,1]$. Then $\sigma(T)=\mathbb{C}$ (just a differential equation).
\end{ex}

\begin{ex}
Consider the operator $T=\I\frac{\dd}{\dd x}$ with domain $D(T)=\{f\in \A[0,1]\mid f(0)=0\}$. We claim that $\rho(T)=\mathbb{C}$. 
\end{ex}
\begin{proof} Let $\lambda\in\mathbb{C}$ and define 
$$S_\lambda g(x):=\I\int_0^x \ee^{-\I\lambda(x-s)}g(s)\dd s.$$
One can show that $(T-\lambda I)S_\lambda g=g$ for all $g\in L^2([0,1])$. Moreover, $S_\lambda(T-\lambda I)g=g$ for all $g\in D(T)$. We need to show that $S_\lambda$ is bounded. Indeed, we have 
\begin{align*}
\| S_\lambda g\|_2^2&=\int_0^1\vert S_\lambda g(x)\vert^2\dd x\leq \sup_{x\in [0,1]}\vert S_\lambda g(x)\vert^2=\sup_{x\in [0,1]}\left\vert\int_0^x\ee^{-\I\lambda (x-s)}g(s)\dd s\right\vert^2\\
&\leq \sup_{x\in [0,1]}\left(\int_0^x\vert \ee^{-\I\lambda(x-s)}g(s)\vert\dd s\right)^2\leq \left(\sup_{x\in[0,1]}\left\vert \int_0^x\ee^{-\I\lambda(x-s)}\dd x\right\vert^2\right)\left(\sup_{x\in[0,1]}\left\vert\int_0^x g(s)\dd s\right\vert^2\right)\\
&\leq C(\lambda)\| g\|^2_2.
\end{align*}
\end{proof}

\subsection{Quantization of a classical system}
We want to talk about quantization of a classical system by considering a ``quantization map'' between classical and quantum data. Consider a map $\mathscr{Q}$, which maps a classical system to a quantum system. The classical (path space) space of states $(T^*M,\omega)$, which is a symplectic manifold coming from a cotangent space, is mapped to a Hilbert space $\mathcal{H}$. Moreover, the space of observables $C^\infty(T^*M)$ is mapped to the space of self adjoint operators. We know that $C^\infty(T^*M)$ is endowed with a Poisson bracket $\{\enspace,\enspace\}$, but the question is what its image is under $\mathscr{Q}$.

\begin{ex}
Let $T^*M=\R^{2n}=\{(x,p)\mid x,p\in\R^n\}$. Then $x_i,p^{i}$ represent position and momentum observables and $\{x_i,p^{i}\}=\delta_{ij}$. Denote by $\widehat{x}_i$ the operator given by multiplication with $x_i$ and by $\widehat{p}^{i}:=-\I\hbar\frac{\partial}{\partial x_i}$. Then their commutator bracket is given by $[\widehat{x}_i,\widehat{p}^{j}]=\I\hbar\delta_{ij}$.
\end{ex}

The previous example can be generalized such that given $\{f,g\}$ for $f,g\in C^\infty(T^*M)$ it will be mapped by $\mathscr{Q}$ to $$\frac{1}{\I\hbar}[\mathscr{Q}(f),\mathscr{Q}(g)],$$
or, by considering $\frac{\dd f}{\dd t}=\{f,H\}$, we get the quantum image $$\I\hbar \frac{\dd}{\dd t}A(t)=[A(t),\widehat{H}],$$
which is basically the Schr\"odinger equation for the Heisenberg picture.

\begin{defn}[Quantization]
A \textsf{quantization} of a classical system $(\R^{2n},\omega)$ is an argument of a quantum Hilber space $\mathcal{H}$ together with a linear map 
$$\mathscr{Q}\colon C^\infty(\R^{2n})\to \{\text{self adjoint operators on $\mathcal{H}$}\}$$
such that the following hold:
\begin{enumerate}[(q1)]
\item{$\mathscr{Q}$ is linear,}
\item{$\mathscr{Q}(1)=\id_\calH$,}
\item{$\mathscr{Q}(x_i)=\widehat{x}_i$, $\mathscr{Q}(p_i)=\widehat{p}_i$,}
\item{$[\mathscr{Q}(f),\mathscr{Q}(g)]=\I\hbar\mathscr{Q}(\{f,g\})$,}
\item{$\mathscr{Q}(\phi\circ f)=\phi(\mathscr{Q}(f))$ for any map $\phi\colon \R\to\R$.}
\end{enumerate}
\end{defn}

\begin{rem}
The problem is that $(q1)-(q5)$ are inconsistent. Even $(q1),(q3)$, and $(q5)$ are inconsistent.
\end{rem}

\begin{ex}
Consider $n=1$. We want to know what the image of the classical observable $x^2p^2$ is under $\mathscr{Q}$, i.e. $\mathscr{Q}(x^2p^2)$. We write 
$$x^2p^2=\frac{(x^2+p^2)^2-x^4-p^4}{2}.$$
Then we use $(q3)$ and $(q5)$ to get the quantum observables 
$$\frac{(\widehat{x}^2+\widehat{p}^2)^2-\widehat{x}^4-\widehat{p}^4}{2}=\frac{\widehat{p}^2\widehat{x}^2+\widehat{x}^2\widehat{p}^2}{2}.$$
On the other hand we have 
$$xp=\frac{(x+p)^2-x^2-p^2}{2}\Longrightarrow \mathscr{Q}(x^2p^2)=\mathscr{Q}((xp)^2)=\left(\frac{(\widehat{x}^2+\widehat{p}^2)^2-\widehat{x}^4-\widehat{p}^4}{2}\right)^2,$$
which implies 
$$\mathscr{Q}(x^2p^2)=\left(\frac{\widehat{p}^2\widehat{x}^2+\widehat{x}^2\widehat{p}^2}{2}\right)^2,$$
which is in general not what we get before.
\end{ex}

The question here is: what are general approaches to a solution? Even $(q1),(q2),(q4)$ and $(q5)$ are not consistent. We can have two different solutions:
\begin{itemize}
\item{Keep $(q1),(q2),(q3),(q4)$ and choose an appropriate domain for $\mathscr{Q}$,}
\item{Keep $(q1),(q2),(q3)$ and demand $(q4)$ holds asymptotically, i.e.
$$[\mathscr{Q}(f),\mathscr{Q}(g)]=\I\hbar \mathscr{Q}(\{f,g\})+O(\hbar^2).$$
}
\end{itemize}

We have two different approaches:
\begin{enumerate}
\item{(Canonical quantization)
Here we quantize the observables $x_i,p_{i}$ as the image of $\mathscr{Q}$, i.e. $x_i\mapsto \widehat{x}_i$ and $p_i\mapsto \widehat{p}_{i}$. Moreovr, $f(x,p)\mapsto f(\widehat{x},\widehat{p})$ and the question will be what to do for $x_ip_j$? More precisely, there is an ordering problem. We need to know how to define $\mathscr{Q}(x_i^2p_j^2)$.
}
\item{(Wick ordering quantization)
Consider $z=x+i\alpha p$ and $\bar z=x-\I\alpha p$. Then write $f(x,p)$ as $f(z,\bar z)$, e.g. $$f(z,\bar z)=\sum_{ij}a_{ij}z_i^{r_i}\bar z_j^{r_j},$$ and with $\widehat{z}=\widehat{x}+\I\alpha\widehat{p}$, $\widehat{z}^*=\widehat{x}-\I\alpha\widehat{p}$ we get 
$$\mathscr{Q}_{Wick}(f)=f(\widehat{z},\overline{\widehat{z}})=\sum_{ij}(\widehat{z}_j^{r_j})^*\widehat{z}_i^{r_i}.$$

\begin{ex}
Consider $n=1$. Then, by writing $x=\frac{1}{2}(z+\bar z)$, we get
\begin{align*}
\mathscr{Q}_{Wick}(x^2)&=\mathscr{Q}_{Wick}\left(\frac{1}{4}z^2+2z\bar z+\bar z^2\right)\\
&=\frac{1}{4}\left((\widehat{x}+\I\alpha\widehat{p})^2+2(\widehat{x}+\I\alpha\widehat{p})(\widehat{x}-\I\alpha\widehat{p})+(\widehat{x}+\I\alpha\widehat{p})^2\right)\\
&=\frac{1}{4}\left(\widehat{x}^2-\alpha^2\widehat{p}^2+\I\alpha(\widehat{x}\widehat{p}+\widehat{p}\widehat{x})+2(\widehat{x}^2+\alpha^2\widehat{p}^2+\I\alpha[\widehat{x},\widehat{p}])+\widehat{x}^2-\alpha^2\widehat{p}^2-\I\alpha(\widehat{x}\widehat{p}+\widehat{p}\widehat{x})\right)\\
&=\frac{1}{4}\left( 4\widehat{x}'2+2\I\alpha[\widehat{x},\widehat{p}]\right)=\widehat{x}^2-\frac{1}{2}\hbar \alpha I,
\end{align*}
where $I$ is the identity operator. 
\end{ex}

}
\item{(Weyl Quantization)
Consider $n=1$. We define $\mathscr{Q}_{Weyl}(x,p):=\frac{\widehat{x}\widehat{p}+\widehat{p}\widehat{x}}{2}$. E.g. $\mathscr{Q}_{Weyl}(x^2p)=\mathscr{Q}_{Weyl}(xxp)=\frac{\x^2\p+\x\p\x+\p\x^2}{3!}$. More generally, 
$$\mathscr{Q}_{Weyl}(x_ip_j)=\frac{1}{(n+m)!}\sum_{\sigma\in S_{n+m}}\x_{\sigma(1)}\dotsm\x_{\sigma(1)}\p_{\sigma(1)}\dotsm \p_{\sigma(m)}.$$

\begin{exe}
\label{ex_weyl}
Let $g$ be any polynomial  in $x$ and $p$. Then 
$$\mathscr{Q}_{Weyl}(x\cdot g)=\mathscr{Q}_{Weyl}(x)\mathscr{Q}_{Weyl}(g)-\frac{\I\hbar}{2}\mathscr{Q}_{Weyl}\left(\frac{\partial g}{\partial p}\right)=\mathscr{Q}_{Weyl}(g)\mathscr{Q}_{Weyl}(x)-\frac{\I\hbar}{2}\mathscr{Q}_{Weyl}\left(\frac{\partial g}{\partial p}\right)$$
$$\mathscr{Q}_{Weyl}(p\cdot g)=\mathscr{Q}_{Weyl}(p)\mathscr{Q}_{Weyl}(g)+\frac{\I\hbar}{2}\mathscr{Q}_{Weyl}\left(\frac{\partial g}{\partial x}\right)=\mathscr{Q}_{Weyl}(g)\mathscr{Q}_{Weyl}(p)-\frac{\I\hbar}{2}\mathscr{Q}_{Weyl}\left(\frac{\partial g}{\partial x}\right)$$
\end{exe}

\begin{prop}
Let $f$ be a polynomial in $x$ and $p$ of degree at most $2$ and $g$ be any polynomial. Then 
$$[\mathscr{Q}_{Weyl}(f),\mathscr{Q}_{Weyl}(g)]=\I\hbar \mathscr{Q}_{Weyl}(\{f,g\}).$$
\end{prop}

\begin{proof}
Let $f=f$. Then $\{x,g\}=\frac{\partial g}{\partial p}$. Using exercise \ref{ex_weyl} we get 
$$[\mathscr{Q}_{Weyl}(x),\mathscr{Q}_{Weyl}(g)]=\frac{\I\hbar}{2}\mathscr{Q}_{Weyl}\left(\frac{\partial g}{\partial p}\right)+\frac{\I\hbar}{2}\mathscr{Q}_{Weyl}\left(\frac{\partial g}{\partial p}\right)=\I\hbar\mathscr{Q}\left(\frac{\partial g}{\partial p}\right).$$
\end{proof}
}
\end{enumerate}

\begin{rem}
This is not possible for arbitrary polynomials $f$ and $g$, because of the NO-GO theorem of Gronewald.
\end{rem}

\subsection{More on self adjoint operators}
\begin{thm}
Let $A$ be a self adjoint operator. Then $\sigma(A)\subseteq \R$.
\end{thm}
\begin{proof}
Assume $A$ is bounded. Let $\lambda=a+\I b$ with $b\not=0$. We calim that $\lambda\in\rho(A)$. Let $\psi\in\mathcal{H}$. Moreover, define $T:=(A-aI)$. Then 
\begin{align*}
\langle (A-\lambda I)\psi,(A-\lambda I)\psi\rangle&=\langle (A-aI)\psi-\I b\psi,(A-aI)\psi-\I b\psi\rangle\\
&=\| T\psi\|^2-\langle \I b\psi,T\psi\rangle-\langle T\psi,\I b\psi\rangle+b^2\|\psi\|^2\\
&=\| T\psi\|^2+b^2\|\psi\|^2\\
&>b^2\|\psi\|^2.
\end{align*}
Hence $\langle (A-\lambda I)^*(A-\lambda I)\psi,\psi\rangle>b^2\|\psi\|^2$ and thus $(A-\lambda I)^*(A-\lambda I)$ is a positive operator. Moreover, we can show that $(A-\lambda I)^{-1}$ is bounded.
\end{proof}
\begin{rem}
There are also plenty examples for unbounded operators.
\end{rem}

\subsection{Eigenvalues of single Harmonic Oscillator}

Let $\mathcal{H}=L^2(\R)$ and consider the Hamiltonian $H(x,p)=\frac{1}{2m}p^2+\frac{kx^2}{2}$ with $k=m\omega^2$. Then going to the corresponding operator formulation, we have $\widehat{p}=\I\hbar\frac{\dd}{\dd x}$ and $\widehat{x}$ is just multiplication by $x$. Then the Hamilton operator is given by $$\widehat{H}=\frac{1}{2m}\widehat{p}^2+\frac{k\widehat{x}^2}{2}=\frac{1}{2m}\left(\widehat{p}^2+(m\omega\widehat{x})^2\right).$$
We will only do formal computations (i.e. we forget about the domains). Define $a=\frac{m\omega \widehat{x}+\I \widehat{p}}{\sqrt{2\hbar m\omega}}$ and $a^*=\frac{m\omega\widehat{x}-\I\widehat{p}}{\sqrt{2\hbar m\omega}}$.

\begin{lem}
We have 
$$\widehat{H}=\hbar \omega\left(a^*a+\frac{1}{2}I\right).$$
\end{lem}

\begin{lem}
The following hold:
\begin{enumerate}
\item{$[a,a^*]=I,$}
\item{$[a,a^*a]=a,$}
\item{$[a^*,a^*a]=-a^*$}
\end{enumerate}
\end{lem}

\begin{proof}
Exercise
\end{proof}

\begin{prop}
Assume that $\psi$ is an eigenvector of $a^*a$ with eigenvalue $\lambda$. Then 
\begin{align}
a^*a(a\psi)&=(\lambda-1)a\psi,\\
a^*a(a^*\psi)&=(\lambda+1)a^*\psi.
\end{align}
\end{prop}

\begin{rem}
The consequence of this proposition is that either $a\psi$ is an eigenvector or $a\psi=0$. We know that $a^*a\geq 0$ so all eigenvalues are non-negative. Hence, if $\psi$ is an eigenvector with eigenvalue $\lambda$, then there is some number $N$ sucht that $a^N\psi\not=0$ but $a^{N+1}\psi=0$. 
\end{rem}

Define $\psi_0=a^N\psi$. Then $a^*a\psi_0=0$ and thus $\psi_0$ is an eigenvector of zero eigenvalue.

\begin{prop}
Let $\psi_0$ be such that $\|\psi_0\|=1$ and $a\psi_0=0$. Then, $\psi_n:=(a^*)^n\psi_0$, for $n\geq 0$, satisfies the following:
\begin{enumerate}[$(i)$]
\item{$a^*\psi_n=\psi_{n+1},$}
\item{$(a^*a)\psi_n=n\psi_n$,}
\item{$\langle\psi_n,\psi_m\rangle=n!\delta_{mn}$,}
\item{$a\psi_{n+1}=(n+1)\psi_n$}
\end{enumerate}
\end{prop}

\begin{rem}
Our goal is to find some $\psi_0\in L^2(\R)$ such that $a\psi_0=0$ and $\|\psi_0\|=1$.
\end{rem}

Define $\widetilde{x}=\frac{x}{\sqrt{\frac{\hbar}{m\omega}}}$, then $\frac{\dd}{\dd \widetilde{x}}=\sqrt{\frac{\hbar}{m\omega}}\frac{\dd}{\dd x}$. Thus 
$$a=\frac{1}{\sqrt{2}}\left(\widetilde{x}+\frac{\dd}{\dd \widetilde{x}}\right),\hspace{0.3cm}a^*=\frac{1}{\sqrt{2}}\left(\widetilde{x}-\frac{\dd}{\dd\widetilde{x}}\right).$$

We want to solve the equation $a\psi_0=0$. This is equivalent to $\frac{\dd\psi_0}{\dd\widetilde{x}}+\widetilde{x}\psi_0=0$, which implies that 
$$\psi_0(x)=\sqrt{\frac{2m\omega}{\hbar}}\ee^{-\frac{m\theta}{2\hbar}x^2}\in \calS(\R).$$
Here $\calS(\R)$ represents the space of \textsf{Schwartz functions} on $\R$ (see Subsection \ref{free particle})
\begin{prop}
For $H_n(\widetilde{x})$ satisfying $H_0(\widetilde{x})=1$ and $H_{n+1}(\widetilde{x})=\frac{1}{\sqrt{2}}\left(2\widetilde{x}H_n(\widetilde{x})-\frac{\dd H_n(\widetilde{x})}{\dd \widetilde{x}}\right)$ we have 
$$\psi_n(\widetilde{x})=H_n(\widetilde{x})\psi_0(\widetilde{x}).$$
\end{prop}

\begin{rem}
One can check that the family $\{\psi_n\}$ forms an orthogonal basis of $L^2(\R)$. 
\end{rem}

We want to ask the following question: Is $\{\hbar\omega(n+\frac{1}{2})$ for $n=0,1,2,...\}$ the full spectrum of $\widehat{H}$? The answer is yes, but the proof is not straight forward.

\subsection{Weyl Quantization on $\R^{2n}$}
Let $f$ be a sufficiently nice function, e.g. $f\in\calS(\R^{2n})$. We define $\mathscr{Q}_{Weyl}(f)$ as an operator on $L^2(\R^n)$ by 
\[
\mathscr{Q}_{Weyl}(f):=\frac{1}{(2\pi)^n}\int_{\R^{2n}}\widehat{f}(a,b)\underbrace{\ee^{\I(a\widehat{x}+b\widehat{p})}}_{U(a,b)}\dd a\dd b,
\]
where $\widehat{f}$ denotes the \textsf{Fourier transform} of $f$. We can compute $U(a,b)$ by using the BCH formula: $\ee^{A+B}=\ee^{[A,B]/2}\ee^{A}\ee^B$ if $[[A,B],B]=[A,[A,B]]$. Formally, we get
\[
U(a,b)=\ee^{-\frac{1}{2}[\I a\widehat{x},\I b\widehat{p}]}\ee^{\I a\widehat{x}}\ee^{\I b\widehat{p}}=\ee^{\frac{\I\hbar}{2}ab}\ee^{\I a\widehat{x}}\ee^{\I b\widehat{p}}.
\]
\begin{exe}
Show $\left(\ee^{\I b\widehat{p}}\psi\right)(x)=\psi(x+\hbar b)$.
\end{exe}
Using the exercise, we get $U(a,b)\psi(x)=\ee^{\I\hbar ab}\ee^{\I a\widehat{x}}\psi(x+\hbar b)$. There are some nice properties for the Weyl quantization:
\begin{itemize}
\item{If $f\in\calS(\R^{2n})$, then $\mathscr{Q}_{Weyl}(f)$ is a bounded operator on $L^2(\R^n)$. In fact, it is a Hilbert-Schmidt operator.
}
\item{The map $\mathscr{Q}_{Weyl}\colon \calS(\R^{2n})\to L^2(\R^n)$ is one-to-one.
}
\item{Let $f,g\in\calS(\R^{2n})$. Then $[\mathscr{Q}_{Weyl}(f),\mathscr{Q}_{Weyl}(g)]=\I\hbar\mathscr{Q}_{Weyl}(\{f,g\})+O(\hbar^2)$.
}
\end{itemize}

\section{Solving Schr\"odinger equations, Fourier Transform and Propagator}

Recall that in the Hamiltonian formalism of classical mechanics the dynamics (time evolution) was generated by Hamilton's equations associated to a Hamiltonian function $H\in C^\infty(T^*M)$. In quantum mechanics, it is postulated that time evolution is described by the Schr\"odinger equation associated to the quantum Hamiltonian $\widehat{H}$: Given $\psi\in \mathcal{H}$ we consider 
\begin{equation}
\label{Schroedinger2}
\begin{cases}
\I\hbar\frac{\dd}{\dd t}\psi(t)&=\widehat{H}\psi(t)\\
\psi(0)&=\psi
\end{cases}
\end{equation}
Before we discuss how to solve the Schr\"odinger equation (SE), let us briefly mention some features of the equation.
\begin{enumerate}
\item{The SE is a \textsf{linear} equation: If $\psi_1(t)$ and $\psi_2(t)$ solve the SE with $\psi_1(0)=\psi_1$ and $\psi_2(0)=\psi_2$, then $\alpha\psi_1(t)+\beta\psi_2(t)$ solve the SE with 
\[
\alpha\psi_1(0)+\beta\psi_2(0)=\alpha\psi_1+\beta\psi_2.
\]
\begin{rem}
The linear SE can easily be generalized to a nonlinear equation but we do not discuss that here.
\end{rem}
}
\item{The SE is \textsf{deterministic} in the sense that given $\psi\in\mathcal{H}$, there is a canonical way to produce $\psi(t)$ (we will make this precise later).}
\item{\textsf{Unitarity}: $\|\psi(t)\|^2=\|\psi\|^2$ for all $t$ (compare this with conservation of energy in classical mechanics).}
\end{enumerate}

\subsection{Solving the Schr\"odinger equation}

We start with a simple situation, namely we assume that $\{\lambda_j\}_{j\in I}$ are eigenvalues of $\widehat{H}$ and $\{\phi_{\lambda_j}\}$ form an orthonormal basis of $\mathcal{H}$, where $\phi_{\lambda_j}$ is an eigenvector associated to the eigenvalue $\lambda_j$, i.e. the equation $\widehat{H}\phi_{\lambda_j}=\lambda_j\phi_{\lambda_j}$ holds. We want to solve 
\begin{equation}
\label{eig_schroedinger}
\begin{cases}
\I\hbar \phi_{\lambda_j}(t)&=\widehat{H}\phi_{\lambda_j}(t)\\
\phi_{\lambda_j}(0)&=\phi_{\lambda_j}
\end{cases}
\end{equation}
We want to formulate the idea for solving this equation. Look for solutions of the form 
\[
\phi_{\lambda_j}(t)=f(t)\phi_{\lambda_j}.
\]
From \eqref{eig_schroedinger} it follows that 
\begin{equation}
\label{}
\begin{cases}
\I\hbar f'(t)\phi_{\lambda_j}&=\lambda_jf(t)\phi_{\lambda_j}\\
f(0)&=1
\end{cases}
\end{equation}
Clearly we can take $f(t)=\ee^{-\frac{\I}{\hbar}\lambda_j t}$, and we see that $\phi_{\lambda_j}(t)=\ee^{-\frac{\I}{\hbar}t\lambda_j}\phi_{\lambda_j}$ solves \eqref{eig_schroedinger}. Note that we can write 
\begin{equation}
\label{eig_eq}
\phi_{\lambda_j}(t)=\ee^{-\frac{\I}{\hbar}t\widehat{H}}\phi_{\lambda_j}. 
\end{equation}
Now equation \eqref{eig_eq} together with the linearity of the SE suggests that ``formally'' for all $\psi\in \mathcal{H}$, 
\begin{equation}
\label{sol}
\psi(t)=\ee^{-\frac{\I}{\hbar}t\widehat{H}}\psi
\end{equation}
solves the SE \eqref{Schroedinger2}. In fact, if $\psi\in D(\widehat{H})$, then using Stone's theorem it can be deduced that $\psi(t)\in D(\widehat{H})$ for all $t$, and in this case $\psi(t)$ defined as in \eqref{sol} indeed solves the SE \eqref{Schroedinger2}. Hence \eqref{sol} can be interpreted as a canonical time evolution of $\psi\in\mathcal{H}$. This is what is usually referred as the deterministic feature of the SE. 
\begin{rem}
To define $\psi(t)=\ee^{-\frac{\I}{\hbar}t\widehat{H}}\psi$ we do not need the assumption that it has an eigenbasis. We only need $\widehat{H}$ to be self adjoint. 
\end{rem}
\begin{defn}[Propagator]
The operator $U(t)=\ee^{-\frac{\I}{\hbar}t\widehat{H}}$ is called the (quantum mechanical) \textsf{propagator}.
\end{defn}
\begin{lem}
If $\{\phi_{\lambda_j}\}$ is an eigenbasis with $\phi_{\lambda_j}$ being eigenvectors associated to the eigenvalues $\lambda_j$ then 
\begin{equation}
U(t)=\sum_{j=1}^n\ee^{-\frac{\I}{\hbar}t\lambda_j}\phi_{\lambda_j}^*\otimes\phi_{\lambda_j},
\end{equation}
where $\phi_{\lambda_j}^*\in\mathcal{H}^*$ is the dual of $\phi_{\lambda_j}$.
\end{lem}
\begin{proof}
Let $\psi\in\mathcal{H}$. Then we can write it as a linear combination $\psi=\sum_{k=1}^nc_k\phi_{\lambda_{k}}$. We know that 
\begin{equation}
U(t)\psi=\sum_{k=1}^nc_kU(t)\phi_{\lambda_k}=\sum_{k=1}^nc_k\ee^{-\frac{\I}{\hbar}t\lambda_k}\phi_{\lambda_k}.
\end{equation}
On the other hand
\begin{equation}
\left(\sum_{j=1}^n\ee^{-\frac{\I}{\hbar}t\lambda_j}\phi_{\lambda_j}^*\otimes\phi_{\lambda_j}\right)\psi=\sum_{k,j=1}^nc_k\ee^{-\frac{\I}{\hbar}t\lambda_k}\phi_{\lambda_j}\underbrace{\phi_{\lambda_j}^*(\phi_{\lambda_k})}_{=\delta_{jk}}\phi_{\lambda_j}
=\sum_{k=1}^nc_k\ee^{-\frac{\I}{\hbar}t\lambda_k}\phi_{\lambda_k}.
\end{equation}
Thus for all $\psi$ we get 
\[
U(t)\psi=\left(\sum_{j=1}^n\ee^{-\frac{\I}{\hbar}t\lambda_j}\phi_{\lambda_j}^*\otimes\phi_{\lambda_j}\right)\psi.
\]
\end{proof}
Let us give a short summary of the discussion so far. 
\begin{itemize}
\item{The operator $U(t)=\ee^{-\frac{\I}{\hbar}t\widehat{H}}$ can be used to describe time evolution of states in a canonical way.
}
\item{If $\widehat{H}$ has a eigenbasis $\{\phi_{\lambda_j}\}$, correpsondig to the eigenvalues $\lambda_j$, then $U(t)$ can be described explicitely as $$U(t)=\sum_{j=1}^n\ee^{-\frac{\I}{\hbar}t\widehat{H}}\phi_{\lambda_j}^*\otimes\phi_{\lambda_j}.$$
}
\end{itemize}
\subsection{The Schr\"odinger equation for the free particle moving on $\R$}
\label{free particle}
Recall that we have $\mathcal{H}=L^2(\R)$ and $\widehat{H}=\frac{1}{2m}\p^2=-\frac{\hbar^2}{2m}\frac{\dd^2}{\dd x^2}$. Hence the SE \eqref{Schroedinger2} becomes
\begin{equation}
\label{time_dep}
\begin{cases}
\I\hbar\frac{\partial}{\partial t}\psi(x,t)&=-\frac{\hbar^2}{2m}\frac{\partial^2}{\partial x^2}\psi(x,t)\\
\psi(x,0)&=\psi(x)
\end{cases}
\end{equation}
Here we will discuss how to solve \eqref{time_dep} with Fourier transform. We will also try to find an explicit representation of $U(t)$.

\subsubsection{Digression on Fourier Transform} 
We will briefly recall the definition and properties of the Fourier transform. Let $\mathcal{S}(\R^n)$ be the space of \textsf{Schwartz functions} on $\R^n$. Recall that $f\in \mathcal{S}(\R^n)$ roughly means that $f\in C^\infty(\R^n)$ and $f$ and all its derivatives approach to zero as $\vert x\vert\to\infty$ faster than any polynomial function approaches to infinity. Now let $f\in\mathcal{S}(\R^n)$. The \textsf{Fourier transform} $\mathcal{F}(f)$, or simply $\widehat{f}$, of $f$ is defined by 
\begin{equation}
\label{Fourier}
\widehat{f}(k)=\frac{1}{(2\pi)^{\frac{n}{2}}}\int_{\R^n}\ee^{-\I \langle k,x\rangle}f(x)\dd x,
\end{equation}
where $\langle\enspace,\enspace\rangle:\R^n\times\R^n\to\R$ again denotes the standard inner product on $\R^n$. 
We want to list some properties of the Fourier transform without proofs:
\begin{enumerate}[$(i)$]
\item{If $f\in \mathcal{S}(\R^n)$, then $\widehat{f}\in \mathcal{S}(\R^n)$.
}
\item{Let $f\in\mathcal{S}(\R^n)$. Then 
\begin{align}
\widehat{\frac{\partial f}{\partial x_j}}&=\I k_j\widehat{f},\\
\widehat{x_jf}&=\I\frac{\partial \widehat{f}}{\partial k_j}
\end{align}
}
\item{Let $f\in\mathcal{S}(\R^n)$, then 
\begin{equation}
\label{Fourier_inv}
\mathcal{F}^{-1}(\widehat{f})(x)=f(x)=\frac{1}{(2\pi)^{\frac{n}{2}}}\int_{\R^n}\ee^{\I\langle k,x\rangle}\widehat{f}(k)\dd k.
\end{equation}
This is called the \textsf{inverse Fourier transform}.
}
\item{Let $f\in\mathcal{S}(\R^n)$, then 
\begin{equation}
\label{Plancherel}
\int_{\R^n}\vert f(x)\vert^2\dd x=\int_{\R^n}\vert \widehat{f}(k)\vert^2\dd x.
\end{equation}
This is called \textsf{Plancherel's formula}.
}
\item{
\begin{thm}[Combined inversion and Plancherel formula]
The Fourier transform $\mathcal{F}:\mathcal{S}(\R^n)\to\mathcal{S}(\R^n)$ can be extended to a unique bounded map $\mathcal{F}:L^2(\R^n)\to L^2(\R^n)$. This map can be computed as 
\begin{equation}
\mathcal{F}(f)(k)=\frac{1}{(2\pi)^{\frac{n}{2}}}\lim_{A\to\infty}\int_{\vert x\vert\leq A}\ee^{-\I \langle k,x\rangle}f(x)\dd x.
\end{equation}
Moreover, the inverse Fourier transform $\mathcal{F}^{-1}:L^2(\R^n)\to L^2(\R^n)$ is unitary and 
\begin{equation}
\mathcal{F}^{-1}(f)(k)=\frac{1}{(2\pi)^{\frac{n}{2}}}\lim_{A\to\infty}\int_{\vert x\vert\leq A}\ee^{\I \langle k,x\rangle}\widehat{f}(k)\dd k.
\end{equation}
\end{thm}
\begin{rem}
If $f\in L^1(\R^n)\cap L^2(\R^n)$, then 
\[
\mathcal{F}(f)(k)=\frac{1}{(2\pi)^{\frac{n}{2}}}\int_{\R^n}\ee^{-\I \langle k,x\rangle}f(x)\dd x,
\]
because in this case 
\[
\lim_{A\to\infty}\int_{\vert x\vert \leq A}\ee^{-\I\langle k,x\rangle}f(x)\dd x=\int_{\R^n}\ee^{-\I\langle k,x\rangle}f(x)\dd x
\]
by dominated convergence.
\end{rem}
}
\item{Let $f$ and $g$ be two measurable functions. Then the \textsf{convolution} $f*g$ of $f$ and $g$ is defined as 
\[
(f*g)(x)=\int_{\R^n}f(x-y)g(y)\dd y,
\]
where we assume that the right hand side exists. Suppose $f,g\in L^1(\R^n)\cap L^2(\R^n)$. Then 
\[
\frac{1}{(2\pi)^{\frac{n}{2}}}\mathcal{F}(f* g)=\mathcal{F}(f)\mathcal{F}(g).
\]
}
\end{enumerate}

\subsection{Solving the Schr\"odinger equation with Fourier Transform}

First, we look for solutions of the form $\psi(x,t)=\ee^{\I(kx-\omega(k)t)}$. From \eqref{time_dep}, it is clear that $\psi(x,t)$ is a solution iff $\omega(k)=\frac{\hbar k^2}{2m}$. Hence, 
\begin{equation}
\label{sol2}
\psi(x,t)=\ee^{\I kx-\I\frac{\hbar k^2}{2m}t} 
\end{equation}
is a solution. However, note that, such $\psi(x,t)\not\in L^2(\R^n)$. Therefore, $\psi(x,t)$ is not the solution we are looking for. Here, the idea is to use $\psi(x,t)$ to produce a senseble solution of \eqref{time_dep}

\begin{prop}
\label{prop2}
Let $\psi_0\in \mathcal{S}(\R)$ and let $\widehat{\psi}_0$ be its Fourier transform. Define 
\begin{equation}
\label{sol3}
\psi(x,t)=\frac{1}{(2\pi)^{\frac{1}{2}}}\int_\R\widehat{\psi}_0(k)\ee^{\I(kx-\omega(k)t)}\dd k.
\end{equation}
Then $\psi(x,t)$ is a solution of \eqref{time_dep} with $\psi(x,0)=\psi_0(x)$.
\end{prop}

\begin{proof}
Since $\widehat{\psi}_0(k)\in\mathcal{S}(\R)$, we can check that the derivatives with respect to $x$ and $t$ can be interchanged with the integral sign in the definition of $\psi(x,t)$. Since $\ee^{\I(kx-\omega(k)t)}$ solves the SE, we can easily check that $\psi(x,t)$ solves \eqref{time_dep}. Moreover, 
\[
\psi(x,0)=\frac{1}{(2\pi)^{\frac{1}{2}}}\int_\R \ee^{\I kx}\widehat{\psi}_0(k)\dd t=\psi_0(x),
\]
where the last equatlity holds because of the inverse Fourier transform.
\end{proof}

\begin{cor}
Let $\psi_0$ be as in proposition \ref{prop2}. Let $\widehat{\psi}(k,t)$ be the Fourier transform of $\psi(x,t)$ with respect to $t$. Then 
\[
\widehat{\psi}(x,t)=\widehat{\psi}_0(k)\ee^{-\I\omega(k)t}.
\]
\end{cor}
\begin{proof}
From proposition \ref{prop2} we know 
\[
\psi(x,t)=\frac{1}{(2\pi)^{\frac{1}{2}}}\int_\R\ee^{\I kx}\left(\ee^{\I\omega(k)t}\widehat{\psi}_0(k)\right)\dd k.
\]
Thus, the claim follows.
\end{proof}
Form property $(vi)$ of Fourier transforms, formally we get 
\[
\ee^{-\I\omega(k)t}\widehat{\psi}_0(k)=\frac{1}{(2\pi)^{\frac{1}{2}}}\mathcal{F}(K_t*\psi_0),
\]
where $\mathcal{F}(K_t)=\ee^{-\I\omega(k)t}$, i.e. $K_t=\mathcal{F}^{-1}\left(\ee^{-\I\omega(k)t}\right)=\frac{1}{(2\pi)^{\frac{1}{2}}}\int_\R\ee^{\I kx}\ee^{-\I\omega(k)t}\dd k$. Again, a ``formal computation'' shows that 
\[
K_t(x)=\sqrt{\frac{m}{\I2\pi\hbar t}}\ee^{\frac{\I mx^2}{2t\hbar}}.
\]
The computation of $K_t(x)$ is ``formal'' because $\ee^{-\I\omega(k)t}\not\in L^1(\R)\cap L^2(\R)$ and thus we do not know how to take the inverse Fourier transform of it. Hence, we need a way to make sense of an integral of the form 
\begin{equation}
\label{Fourier_sense}
\int_\R\ee^{\I kx}\ee^{-\I\omega(k)t}\dd k.
\end{equation}
Integrals of the form \eqref{Fourier_sense} are called \textsf{Fresnel Integrals}. 
\subsubsection{Digression on Fresnel Integrals} Let $Q$ be a real, symmetric $n\times n$-matrix with $\det(Q)\not=0$. An integral of the form 
\[
\int_{\R^n}\ee^{\frac{\I}{2}\langle Qx,x\rangle}\dd x
\]
is called a \textsf{Fresnel integral}, and is defined as 
\[
\int_{\R^n}\ee^{\frac{\I}{2}\langle Qx,x\rangle}\dd x:=\lim_{\varepsilon\to 0}\int_{\R^n}\ee^{-\frac{1}{2}\varepsilon\langle x,x\rangle}\ee^{\frac{\I}{2}\langle Qx,x\rangle}\dd x.
\]
As a matter of fact we have 
\[
\int_{\R^n}\ee^{\frac{\I}{2}\langle Qx,x\rangle}\dd x=\ee^{\frac{\pi \I}{4}sign(Q)}\frac{1}{\left\vert\det\left(\frac{Q}{2\pi}\right)\right\vert^{\frac{1}{2}}},
\]
where $sign(Q)=\#\text{positive eigenvalues}-\#\text{negative eigenvalues}$. More generally, for $\omega\in \R^n$, we have 
\begin{align}
\begin{split}
\label{eleven}
\int_{\R^n}\ee^{\frac{\I}{2}\langle Qx,x\rangle}\ee^{\langle\omega,x\rangle}\dd x&=\lim_{\varepsilon\to 0}\int_{\R^n}\ee^{\frac{\I}{2}\langle Qx,x\rangle-\frac{1}{2}\varepsilon\langle x,x\rangle}\ee^{\langle\omega,x\rangle}\dd x\\
&=\frac{\ee^{\frac{\pi \I}{4}sign(Q)}}{\left\vert\det\left(\frac{Q}{2\pi}\right)\right\vert^{\frac{1}{2}}}\ee^{\frac{\I}{2}\langle Q^{-1}\omega,\omega\rangle}.
\end{split}
\end{align}
We use this general result, to compute
\begin{equation}
\label{Fresnel}
\frac{1}{(2\pi)^{\frac{1}{2}}}\int_\R\ee^{-\I\omega(k)t}\ee^{\I kx}\dd k.
\end{equation}
Now, using \eqref{eleven}, it can be easily checked that 
\[
K_t(x)=\sqrt{\frac{m}{2\pi\I kt}}\ee^{\frac{\I mx^2}{2t\hbar}}.
\]
\begin{rem}
There is also another way to define \eqref{Fresnel} (see \cite{BH}). 
\end{rem}
Now we make our previous formal discussion mathematically.
\begin{prop}
Suppose $\psi_0\in L^1(\R)\cap L^2(\R)$ and define 
\[
\psi(x,t)=\mathcal{F}^{-1}\left(\widehat{\psi}_0(k)\ee^{-\frac{\hbar k^2 t}{2m}}\right).
\]
Then $\psi(x,t)=K_t*\psi_0$, where $K_t(x)=\sqrt{\frac{m}{2\pi\I kt}}\ee^{\frac{\I mx^2}{2t\hbar}}$.
\end{prop}
\begin{proof}
We will only briefly sketch the proof. The idea here is to show that 
\begin{equation}
\label{conv}
\mathcal{F}(K_t*\psi_0)=\widehat{\psi}_0(k)\ee^{-\frac{\I \hbar k^2 t}{2m}}.
\end{equation}
We can not talk about $\mathcal{F}(K_t)$ as $K_t\not\in L^2(\R)$. However, we can consider $K_t\chi_{[-n,n]}$ and its Fourier transform. Observe that 
\[
\frac{1}{(2\pi)^{\frac{1}{2}}}\mathcal{F}(K_t\chi_{[-n,n]}*\psi_0)=\mathcal{F}(K_t\chi_{[-n,n]})\mathcal{F}(\psi_0). 
\]
It can be shown that $K_t\chi_{[-n,n]}*\psi_0\xrightarrow{n\to\infty}K_t*\psi$ in $L^2(\R)$ and 
\[
\mathcal{F}(K_t\chi_{[-n,n]})\mathcal{F}(\psi_0)\xrightarrow{n\to\infty}\frac{1}{(2\pi)^{\frac{1}{2}}}\ee^{-\frac{\I\hbar k^2t}{2m}}\widehat{\psi}_0
\]
in $L^2(\R)$. These two observations imply that \eqref{conv} holds and hence 
\[
K_t*\psi_0=\mathcal{F}^{-1}\left(\widehat{\psi}_0(k)\ee^{-\frac{\hbar k^2 t}{2m}}\right).
\]
\end{proof}

\subsubsection{Summary of the discussion} We have shown that if $\psi_0\in L^1(\R)\cap L^2(\R)$, then 
\[
\ee^{-\frac{\I}{\hbar}t\widehat{H}}\psi_0=\left(\mathcal{F}^{-1}\circ \mathsf{m}\circ \mathcal{F}\right)\psi_0,
\]
where $(\mathsf{m}f)(k)=\ee^{-\frac{\I \hbar k^2 t}{2m}}f(k)$. This means we have shown that the following diagram is commutative.

\diagram
\centering
L^1(\R)\cap L^2(\R)&\rTo^{\mathcal{F}}&L^1(\R)\cap L^2(\R)\\
\dTo^{\ee^{-\frac{\I}{\hbar}t\widehat{H}}} & & \dTo_{\mathsf{m}}\\
L^2(\R)&\lTo_{\mathcal{F}^{-1}}&L^2(\R)
\enddiagram

Moreover, we have shown that $(\mathcal{F}^{-1}\circ \mathsf{m}\circ \mathcal{F})\psi_0=K_t*\psi_0$. Finally, combining these results, we conclude that 
\[
\left(\ee^{-\frac{\I}{\hbar}t\widehat{H}}\psi_0\right)(x)=\sqrt{\frac{m}{2\pi\I\hbar t}}\int_\R\ee^{\frac{\I m(x-y)^2}{2t\hbar}}\psi_0(y)\dd y,
\]
i.e. the integral kernel of $\ee^{-\frac{\I}{\hbar}t\widehat{H}}$ is $K_t(x-y)=\sqrt{\frac{m}{2\pi\I\hbar t}}\ee^{\frac{ \I m(x-y)^2}{2t\hbar }}$.
\begin{rem}
One can check that $K_t(x)$ satisfies the SE and $\lim_{t\to 0}K_t(x)=\delta(x)$ in distributional sense.
\end{rem}
\begin{defn}[Fundamental solution]
$K_t(x)$ is called the \textsf{fundamental solution} of the SE.
\end{defn}
\begin{rem}
One can easily extend the discussion above for the free particle in $\R^n$.
\end{rem}

\part{The Path Integral Approach to Quantum Mechanics}

We saw that Hamilton's approach to classical mechanics inspired an axiomatic approach to quantum mechanics. Hence, it is natural to ask whether there is a ``Lagrangian formulation" of quantum mechanics. Dirac, who viewed Lagrangian mechanics more fundamental, took first steps towards a Lagrangian formulation of quantum mechanics. Feynman advanced it further, which gave rise to the path integral formulation of quantum theory. Dirac suggested that the quantum mechanical propagator $K(t,x,y)$ may be represented by 
\begin{equation}
\label{pathint}
\int_{\gamma \in P(t,x,y)}\ee^{\frac{\I}{\hbar}S(\gamma)}\mathscr{D}\gamma,
\end{equation}
where $P(t,x,y)$ is the space of paths $\gamma:[0,t]\to\R$ joining $x$ to $y$. Since $P(t,x,y)$ is an infinite dimensional manifold, it is not clear what the integral \eqref{pathint} means. 

\section{Feynman's Formulation of the Path Integral}
Feynman's idea was to define \eqref{pathint} as a limit of integrals over finite dimensional manifolds, which roughly goes as follows. Let $P_n(t,x,y)$ be the space of piecewise linear paths joining $x$ to $y$, which consists of $n$ line segments $\ell_{x,x_1},\ell_{x_1,x_2},...,\ell_{x_{n-1},y}$. Clearly, to define $\gamma\in P_n(t,x,y)$, we need to specify $(x_1,...,x_{n-1})$. This means that we can identify $P_n(t,x,y)$ with $\R^{n-1}$. Hence, we can define 
\begin{equation}
\label{feyn_int}
\int_{\gamma\in P(t,x,y)}\ee^{\frac{\I}{\hbar}S(\gamma)}\mathscr{D}\gamma:=
\lim_{n\to\infty}A(n,t)
\int_{\gamma \in P_n(t,x,y)}\ee^{\frac{\I}{\hbar}S(\gamma)}\dd x_1\dotsm \dd x_{n-1},
\end{equation}
where $A(n,t)$ is some constant depending on $n$ and $t$.
\subsection{Free Propagator for the free particle on $\R$}
We have already shown that 
\begin{equation}
\label{Prop}
K(t,x,y)=\sqrt{\frac{m}{2\pi \I\hbar t}}\ee^{\frac{\I}{\hbar}\frac{m}{2t}(x-y)^2}.
\end{equation}
Let us now give a path integral derivation of $K(t,x,y)$. Let $0=t_0<\dotsm<t_n=t$ with $t_i-t_{i-1}=\frac{t}{n}=:\Delta t$. Moreover, let $(x_1,...,x_{n-1})\in\R^{n-1}$ and let $\gamma$ be the piecewise linear path joining $x$ to $y$ such that $\gamma(t_i)=x_i$ and the line segment joining $x_{i-1}$ to $x_i$ is given by 
\[
\gamma(s)=\frac{1}{\Delta t}\left((t_i-s)x_{i-1}+(s-t_{i-1})x_i\right),\hspace{0.5cm}s\in[t_{i-1},t_i],\hspace{0.5cm}i=1,2,...,n
\]
Then 
\begin{equation}
S(\gamma)=\frac{1}{2}m\sum_{i=1}^n\int_{t_{i-1}}^{t_i}
\frac{(x_i-x_{i-1})^2}{(\Delta t)^2}\dd s=\frac{1}{2}m\sum_{i=1}^n\frac{(x_i-x_{i-1})^2}{\Delta t}
\end{equation}
and thus 
\begin{equation}
A(n,t)\int_{\R^n}\ee^{\frac{\I}{\hbar}S(\gamma)}\dd x_1\dotsm \dd x_{n-1}=A(n,t)\int_{\R^n}\ee^{\frac{\I}{\hbar}\frac{m}{2}\sum_{i=1}^n\frac{(x_i-x_{i-1})^2}{\Delta t}}\dd x_1\dotsm \dd x_{n-1}
\end{equation}
Define $f_i=\sqrt{\frac{m}{2\hbar\Delta t}}x_i$. Then by change of variables, this integral will be 
\begin{align}
\begin{split}
A(n,t)\left(\frac{2\hbar \Delta t}{m}\right)^{\frac{n-1}{2}}\int_{\R^{n-1}}
\ee^{\I\sum_{i=1}^n(f_i-f_{i-1})^2}\dd f_1\dotsm \dd f_{n-1}
&=A(n,t)\left(\frac{2\hbar \Delta t}{m}\right)^{\frac{n-1}{2}}\frac{(\pi\I)^{\frac{n-1}{2}}}{\sqrt{n}}\ee^{\frac{\I}{n}(f_n-f_1)^2}\\
&=A(n,t)\left(\frac{2\hbar \Delta t}{m}\right)^{\frac{n-1}{2}}\frac{(\pi\I)^{\frac{n-1}{2}}}{\sqrt{n}}\ee^{\frac{\I}{\hbar}\frac{m}{2n\Delta t}(x_n-x_1)^2}\\
&=A(n,t)\left(\frac{2\pi\I\hbar \Delta t}{m}\right)^{\frac{n-1}{2}}\left(\frac{m}{2n\pi\I\hbar\Delta t}\right)^{\frac{1}{2}}\ee^{\frac{\I}{\hbar}\frac{m}{2n\Delta t}(x_n-x_1)^2}\\
&=A(n,t)\left(\frac{2\pi\I\hbar \Delta t}{m}\right)^{\frac{n-1}{2}}\left(\frac{m}{2n\pi\I\hbar\Delta t}\right)^{\frac{1}{2}}\ee^{\frac{\I}{\hbar}\frac{m}{2t}(y-x)^2}
\end{split}
\end{align}
Define $A(n,t):=\left(\frac{m}{2\pi\I\hbar t}\right)^{\frac{n}{2}}$, then 
\begin{align}
\begin{split}
\int_{\gamma\in P(t,x,y)}\ee^{\frac{\I}{\hbar}S(\gamma)}\mathscr{D}\gamma&=
\lim_{n\to\infty} A(n,t)\int_{\R^{n-1}}\ee^{\frac{\I}{\hbar}S(\widetilde{\gamma})}\dd x_1\dotsm \dd x_{n-1}\\
&=\left(\frac{m}{2\pi\I\hbar t}\right)^{\frac{1}{2}}\ee^{\frac{\I}{\hbar}\frac{m}{2t}(x-y)^2}\\
&=K(t,x,y).
\end{split}
\end{align}
Next we show how to derive the path integral representation of the propagator associated with a Hamiltonian of the form $\widehat{H}_0+V(\x)$, where $\widehat{H}_0=\frac{1}{2m}\p^2$ is the \textsf{free Hamiltonian}. Let us recall the \textsf{Kato-Lie-Trotter product formula}. Let $A$ and $B$ be self adjoint operators on a Hilbert space $\mathcal{H}$ with domains $D(A)$ and $D(B)$ respectively. Assume that $A+B$ is densely defined and essentially self adjoint on $D(A)\cap D(B)$. Then 
\begin{equation}
\lim_{n\to\infty}\left(\ee^{\frac{\I}{n}tA}\ee^{\frac{\I}{n}tB}\right)^n=\ee^{\I t(A+B)}
\end{equation}
in the strong operator topology (i.e. $A_n\to A$ iff $\|A_n\psi-A\psi\|\xrightarrow{n\to\infty}0$ for all $\psi\in\mathcal{H}$). We assume that $V(\x)$ is sufficently nice so that the assumption of the Kato-Lie-Trotter product formula is satisfied. Then for all $\psi\in L^2(\R)$, we have 
\begin{equation}
\label{KLT}
\ee^{-\frac{\I}{\hbar}t(\widehat{H}_0+V(\x))}\psi=
\lim_{n\to\infty}\left(\ee^{-\frac{\I}{\hbar}\frac{t}{n}\widehat{H}_0}\ee^{-\frac{\I}{\hbar}\frac{t}{n}V(\x)}\right)^n\psi
\end{equation}
Let us compute the right hand side of \eqref{KLT}. Recall that 
\begin{equation}
\left(\ee^{-\frac{\I}{\hbar}\frac{t}{n}\widehat{H}_0}\psi\right)(x_1)=\sqrt{\frac{m}{2\pi\I\hbar \frac{t}{n}}}\int_\R\ee^{\frac{\I}{\hbar}\frac{m}{2\frac{t}{n}}(x_1-x_0)^2}\dd x_0
\end{equation}
and 
\begin{equation}
\left(\ee^{-\frac{\I}{\hbar}\frac{t}{n}V(\x)}\psi\right)(x)=\ee^{-\frac{\I}{\hbar}\frac{t}{n}V(x)}\psi(x).
\end{equation}
Using these two relations, we compute
\begin{equation}
\left(\left(\ee^{-\frac{\I}{\hbar}\frac{t}{n}\widehat{H}_0}\ee^{-\frac{\I}{\hbar}\frac{t}{n}V(\x)}\right)\psi\right)(x_1)
=\sqrt{\frac{m}{2\pi\I\hbar\frac{t}{n}}}\int_\R\ee^{\frac{\I}{\hbar}\frac{m}{2\frac{t}{n}}(x_1-x_0)^2}\ee^{-\frac{\I}{\hbar}\frac{t}{n}V(x_0)}\psi(x_0)\dd x_0.
\end{equation}
Repeatedly applying the process we get 
\begin{align}
\begin{split}
\left(\left(\ee^{-\frac{\I}{\hbar}\frac{t}{n}\widehat{H}_0}\ee^{-\frac{\I}{\hbar}\frac{t}{n}V(\x)}\right)^n\psi\right)(x_n)
&=\left(\frac{m}{2\pi\I\hbar\frac{t}{n}}\right)^{\frac{n}{2}}\int_{\R^n}\ee^{\frac{\I}{\hbar}\frac{m}{2\frac{t}{n}}\sum_{k=1}^n(x_k-x_{k-1})^2-\frac{\I}{\hbar}\frac{t}{n}\sum_{k+1}^nV(x_{k-1})}\psi(x_0)\dd x_0\dd x_1\dotsm \dd x_{n-1}\\
&=\left(\frac{m}{2\pi\I\hbar\frac{t}{n}}\right)^{\frac{n}{2}}\int_{\R^n}\ee^{\frac{\I}{\hbar}\sum_{k=1}^n\frac{t}{n}\left\{\frac{m}{2}\left(\frac{x_k-x_{k-1}}{\frac{t}{n}}\right)^2-V(x_{k-1})\right\}}\psi(x)\dd x\dd x_1\dotsm \dd x_{n-1}\\
\end{split}
\end{align}
\begin{equation}
=\int_{\R}\left\{\left(\frac{m}{2\pi\I\hbar\frac{t}{n}}\right)^{\frac{n}{2}}\int_{\R^{n-1}}\ee^{\frac{\I}{\hbar}\sum_{k=1}^n\frac{t}{n}\left\{\frac{m}{2}\left(\frac{x_k-x_{k-1}}{\frac{t}{n}}\right)^2-V(x_{k-1})\right\}}\dd x_1\dotsm \dd x_{n-1}\right\}\psi(x_0)\dd x_0
\end{equation}
Then we get 
\begin{multline}
\lim_{n\to\infty}\left(\left(\ee^{-\frac{\I}{\hbar}\frac{t}{n}\widehat{H}_0}\ee^{-\frac{\I}{\hbar}\frac{t}{n}V(\x)}\right)\psi\right)(x)=\\
=\int_{\R}\left\{\lim_{n\to\infty}\left(\frac{m}{2\pi\I\hbar\frac{t}{n}}\right)^{\frac{n}{2}}\int_{\R^{n-1}}\ee^{\frac{\I}{\hbar}\sum_{k=1}^n\frac{t}{n}\left\{\frac{m}{2}\left(\frac{x_k-x_{k-1}}{\frac{t}{n}}\right)^2-V(x_{k-1})\right\}}\dd x_1\dotsm \dd x_{n-1}\right\}\psi(x_0)\dd x_0=\\
=\int_\R K(t,x,x_0)\psi(x_0)\dd x_0,
\end{multline}
where 
\[
K(t,x,x_0)=\lim_{n\to\infty}\left(\frac{m}{2\pi\I\hbar\frac{t}{n}}\right)^{\frac{n}{2}}\int_{\R^{n-1}}\ee^{\frac{\I}{\hbar}\sum_{k=1}^n\frac{t}{n}\left\{\frac{m}{2}\left(\frac{x_k-x_{k-1}}{\frac{t}{n}}\right)^2-V(x_{k-1})\right\}}\dd x_1\dotsm \dd x_{n-1}.
\]
Moreover, observe that 
\[
\lim_{n\to\infty}\sum_{k=1}^n\frac{t}{n}\left\{\frac{m}{2}\left(\frac{x_k-x_{k-1}}{\frac{t}{n}}\right)^2-V(x_{k-1})\right\}
\]
can be interpreted as 
\[
\int_0^t\left(\frac{m}{2}\|\dot{\gamma}(s)\|^2-V(\gamma(s))\right)\dd s.
\]
This means that 
\begin{equation}
\int_{\gamma\in P(t,x,y)}\ee^{\frac{\I}{\hbar}S(\gamma)}\mathscr{D}\gamma:=\lim_{n\to\infty}\left(\frac{m}{2\pi\I\hbar\frac{t}{n}}\right)^{\frac{1}{2}}\int_{\R^{n-1}}\ee^{\frac{\I}{\hbar}S(\gamma)}\dd x_1\dotsm \dd x_{n-1}=K(t,x,y).
\end{equation}

\section{Construction of the Wiener measure}

We saw that Feynman defined the path integral $\int_{\gamma\in P(t,x,y)}\ee^{\frac{\I}{\hbar}S(\gamma)}\mathscr{D}\gamma$ as a limit of integrals over finite dimensional manifolds. Now we plan to investigate whether or not it is possible to define a probability measure on $P(t,x,y)$, which is of the form 
\[
\frac{\ee^{\frac{\I}{\hbar}S(\gamma)}\mathscr{D}\gamma}{Z},
\]
where $Z$ is some quantity for normalization of the measure. A short answer to this question is \textsf{no}. However, if we replace $\I$ by $-1$ (i.e. \textsf{Wick rotate}) then it is possible to construct a measure of the desired form on a suitable $P(t,x,y)$. This was done by \textsf{Wiener} in 1923 for the case $V(x)=0$ and it is known as \textsf{Wiener measure}. From now on we assume $V(x)=0$ and $S(\gamma)=\frac{1}{2}\int_0^t\|\dot{\gamma}(s)\|^2\dd s$. The basic ideas are the following:
\begin{itemize}
\item{
Interpret 
\begin{equation}
\label{measure}
A(n,t)\int_{E\subseteq \R^{n-1} \atop\text{measurable}}\ee^{-S(\gamma)}\dd x_1\dotsm \dd x_{n-1}
\end{equation}
as a measure of a certain ``measurable" subset of $P(t,x,y)$. 
}
\item{
Instead of taking the limit $n\to\infty$, try to extend this ``measure" defined by \eqref{measure} to a measure on $P(t,x,y)$.
}
\end{itemize}
Essentially, the idea comes from Molecular-kinetic theory. Einstein showed that, if $\rho(x,t)$ is the probability density for finding the Brownian particle at location $x$ and at time $t$, then it satisfies the diffusion equation 
\begin{equation}
\label{diffusion}
\frac{\partial}{\partial t}\rho(x,t)=D\frac{\partial^2}{\partial x^2}\rho(x,t),
\end{equation}
where $D$ is the diffusion constant. This immediatly implies 
\[
\rho(x,t)=\frac{1}{\sqrt{4\pi Dt}}\ee^{\frac{x^2}{4Dt}},
\]
if we insist that $\lim_{t\to 0}\rho(x,t)=\delta_0(x)$, where $\delta_0$ is the \textsf{Dirac delta function}. This implies that for any measurable set $E\subseteq\R$, the probability of finding the Brownian particle in $E$ at time $t$ is given by 
\begin{equation}
\label{prob}
\frac{1}{\sqrt{4\pi Dt}}\int_E\ee^{-\frac{x^2}{4Dt}}\dd x.
\end{equation}
From now on we take $2D=1$. more generally, $\frac{1}{\sqrt{2\pi(t_2-t_1)}}\ee^{\frac{(x-y)^2}{2(t_2-t_1)}}$ is the probability density of finding the particle at $y$ at time $t=t_2$ if it was at $x$ at time $t=t_1$. This means, given $0=t_0<t_1<\dotsm <t_n\leq t$ and $E=\prod_{i=1}^n(\alpha_i,\beta_i]$, we can observe that 
\begin{equation}
\label{Gaussian}
A(n,t)\int_E\ee^{-\frac{1}{2}\sum_{i=1}^n\frac{(x_i-x_{i-1})^2}{t_i-t_{i-1}}}\dd x_1\dotsm \dd x_{n}
\end{equation}
can be interpreted as the probability of finding the Brownian particle in $(\alpha_i,\beta_i]$ at time $t=t_i$. Hence it should not be suprising to interpret \eqref{measure} as a ``measure" of a suitable subset of $P(t,x,y)$. Let us try to make this precise and construct the \textsf{Wiener measure}. First, we need some notations and definitions.
\begin{itemize}
\item{We write 
\[
C_0([0,1])=\{x:[0,1]\to\R\mid \text{$x$ is continuous at $x(0)$}\},
\]
which are paths starting at $0$. Recall that $C_0([0,1])$ is a Banach space with the norm 
\[
\|x\|=\sup_{t\in [0,1]}\vert x(t)\vert.
\]
Hence, it is a topological space. Let $\calB(C_0([0,1]))$ denote the Borel $\sigma$-algebra of $C_0([0,1])$ with respect to the topology induced by the norm $\|\cdot\|$.
}
\item{
Fix $t\in [0,1]$, define $\ev_t:C_0([0,1])\to\R$, $\ev_t(x)=x(t)$. It is easy to check that $\ev_t$ is continuous and hence it is Borel measurable. More generally, given $t_1,...,t_n\in [0,1]$, define 
\begin{align*}
P(t_1,...,t_n):C_0([0,1])&\longrightarrow \R^n\\
x&\longmapsto P(t_1,...,t_n)(x)=(x(t_1),...,x(t_n)),
\end{align*}
i.e. $P(t_1,...,t_n)=(\ev_{t_1},...,\ev_{t_n})$, thus $P(t_1,...,t_n)$ is continuous and hence Borel measurable.
}
\item{
Given $t_1,...,t_n\in [0,1]$ and $(\alpha_1,\beta_1]\times\dotsm \times (\alpha_n,\beta_n]=\prod_{i=1}^n(\alpha_i,\beta_i]\subseteq\R^n$, define 
\[
I\left(t_1,...,t_n,\prod_{i=1}^n(\alpha_i,\beta_i]\right)=
P(t_1,...,t_n)^{-1}
\left(\prod_{i=1}^n(\alpha_i,\beta_i]\right)=\left\{x\in C_0([0,1])\Big| (x(t_1),...,x(t_n))\in \prod_{i=1}^n(\alpha_i,\beta_i]\right\}.
\]
Observe that $I\left( t_1,...,t_n,\prod_{i=1}^n(\alpha_i,\beta_i]\right)$ is Borel measurable. Also note that
\begin{equation}
\label{rel}
I\left( t_1,...,t_n,\prod_{i=1}^n(\alpha_i,\beta_i]\right)=\bigcap_{i=1}^n\ev_{t_i}^{-1}\left((\alpha_i,\beta_i]\right).
\end{equation} 
From \eqref{rel} it is clear that we can always assume $t_1\leq t_2<\dotsm <t_{n-1}\leq t_n$.
}
\end{itemize}

\begin{exe}
Let $t_1,...,t_n\in [0,1]$ and $t_1<t_2,\dotsm <t_n$. Moreover, let $t_{k-1}<s<t_k$. Check that 
\[
I\left(t_1,...,t_n,\prod_{i=1}^n(\alpha_i,\beta_i]\right)=I\left( t_1,...,t_{k-1},s,t_k,...,t_n,\prod_{i=1}^{k-1}(\alpha_i.\beta_i]\times\R\times
\prod_{i=k}^n(\alpha_i,\beta_i]\right).
\]
Hint: use that $I=\bigcap_{i=1}^n\ev_{t_i}^{-1}((\alpha_i,\beta_i])$.
\end{exe}
Let $\mathcal{I}$ be the collection of all $I\left(t_1,...,t_n,\prod_{i=1}^n(\alpha_i,\beta_i]\right)$, where $n\in\N$ (note that we always include zero in $\N$) and $\alpha_i\leq\beta_i$ with $\alpha_i,\beta_i\in\R\cup\{\infty\}$ or all $i$.
\begin{exe}
Check that $\mathcal{I}$ is a semialgebra, i.e.
\begin{enumerate}
\item{$\varnothing, C_0([0,1])\in \mathcal{I}$}
\item{If $I,J\in\mathcal{I}$, then $I\cap J\in\mathcal{I}$.
}
\item{If $I\in \mathcal{I}$, then $C_0([0,1])\setminus I$ is a finite disjoint union of elements in $\mathcal{I}$.
}
\end{enumerate}
\end{exe}
\begin{proof}[Solution]
We have:
\begin{enumerate}
\item{$\varnothing=\ev_1^{-1}((1,1])$, and thus $\varnothing\in\mathcal{I}$.}
\item{Let $I=\bigcap_{i=1}^n\ev_{t_i}^{-1}((\alpha_i,\beta_i])$ and $J=\bigcap_{j=1}^m\ev_{s_j}^{-1}((\gamma_j,\delta_j])$. Then 
\[
I\cap J=\bigcap_{1\leq i\leq n\atop 1\leq j\leq m}\left(\ev_{t_i}^{-1}((\alpha_i,\beta_i])\cap\ev_{s_j}^{-1}((\gamma_j,\delta_j])\right).
\]
Note that $\ev_{t_i}^{-1}((\alpha_i,\beta_i])\cap \ev_{s_j}^{-1}((\gamma_j,\delta_j])$ is of the form $\ev_{t}^{-1}((a,b])$.
}
\item{Let $I\in\mathcal{I}$ with $I=\ev_t^{-1}((\alpha,\beta])$. Then 
$$C_0([0,1])\setminus I=\ev_{t}^{-1}((-\alpha,\alpha])\cup\ev_t^{-1}((\beta,\alpha])\in\mathcal{I}.$$
We leave the general case as an exercise.
}
\end{enumerate}
\end{proof}
\begin{thm}[Wiener]
\label{Wiener_thm}
There is unique probability measure $\mu$ on $\calB(C_0([0,1]))$, such that 
\begin{equation}
\label{Wiener}
\mu\left(I\left(t_1,...,t_n,\prod_{i=1}^n(\alpha_i,\beta_i]\right)\right)=\frac{\int_{\prod_{i=1}^n(\alpha_i,\beta_i]}\ee^{-\frac{1}{2}\sum_{i=1}^n\frac{(x_i-x_{i-1})^2}{t_i-t_{i-1}}}\dd x_1\dotsm \dd x_n
}{\sqrt{(2\pi)^nt_1(t_2-t_1)\dotsm (t_n-t_{n-1})}}\end{equation}
\end{thm}
Let us give a small overview of the proof strategy:
\begin{itemize}
\item{First, we will define $\mu(I)$ for $I\in\mathcal{I}$ by \eqref{Wiener}.}
\item{Then we will use the \textsf{Caratheodory extension} construction.}
\end{itemize}
Given $I\left( t_i,,,.,t_n,\prod_{i=1}^n(\alpha_i,\beta_i]\right)$, define $\mu(I)$ as in \eqref{Wiener}. First we show that $\mu$ is well defined i.e. if $t_{k-1}<s<t_k$, then 
\begin{equation}
\label{measure_prop}
\mu\left( I\left(t_1,...,t_n,\prod_{i=1}^n(\alpha_i,\beta_i]\right)\right)=\mu\left( I\left(t_1,...,t_{k-1},s,t_k,...,t_n,\prod_{i=1}^{k-1}(\alpha_i,\beta_i]\times\R\times\prod_{i=k}^{n}(\alpha_i,\beta_i]\right)\right)
\end{equation}
To verify \eqref{measure_prop}, we need the following lemma:
\begin{lem}[Kolmogorov-Chapman equation]
\label{KC_eq}
Define $K(t,x,y)=\frac{1}{\sqrt{2\pi t}}\ee^{-\frac{(x-y)^2}{2t}}$. Then 
\begin{equation}
\label{KC1}
\int_\R K(t_1,x,y)K(t_2,y,z)\dd y=K(t_1+t_2,x,z).
\end{equation}
In other words
\begin{equation}
\label{KC2}
\frac{1}{\sqrt{(2\pi)^2t_1t_2}}\int_\R\ee^{-\frac{(x-y)^2}{2t_1}}\ee^{-\frac{(y-z)^2}{2t_2}}\dd y=\frac{1}{\sqrt{2\pi(t_1+t_2)}}\ee^{-\frac{(x-z)^2}{2(t_1+t_2)}}.
\end{equation}
\end{lem}
\begin{proof}[Proof of Theorem \ref{Wiener_thm}]
Note that 
\begin{multline}
 \mu\left(I\left(t_1,...,t_{k-1},s,t_k,...,t_n,\prod_{i=1}^{k-1}(\alpha_i,\beta_i]\times\R\times\prod_{i=k}^{n}(\alpha_i,\beta_i]\right)\right)=\\
 =\frac{\int_{\prod_{i=1}^{k-1}(\alpha_i,\beta_i]\times\R\times\prod_{i=k}^{n}(\alpha_i,\beta_i]}\ee^{-\frac{1}{2}\left\{\sum_{i=1}^{k-2}\frac{(x_i-x_{i-1})^2}{t_i-t_{i-1}}+\sum_{i=k+1}^n\frac{(x_i-x_{i-1})^2}{t_i-t_{i-1}}+\frac{(y-x_{k-1})^2}{s-t_{k-1}}+\frac{(x_k-y)^2}{t_k-s} \right\}}\dd x_1\dotsm \dd x_{k-1}\dd y\dd x_{k}\dotsm\dd x_n}{\sqrt{(2\pi)^{n+1}t_1(t_2-t_1)\dotsm (t_{k-1}-t_{k-2})(s-t_{k-1})(t_k-s)\dotsm (t_n-t_{n-1})}}.
\end{multline}
Using Lemma \ref{KC_eq}, we see that 
\begin{equation}
\frac{\int_{\prod_{i=1}^{n}(\alpha_i,\beta_i]}\ee^{-\frac{1}{2}\sum_{i=1}^{n}\frac{(x_i-x_{i-1})^2}{t_i-t_{i-1}}}\dd x_1\dotsm\dd x_n}{\sqrt{(2\pi)^{n+1}t_1(t_2-t_1)\dotsm (t_{k}-t_{k-1})\dotsm (t_n-t_{n-1})}}= \mu \left(I\left(t_1,...,,t_n,\prod_{i=1}^{n}(\alpha_i,\beta_i]\right)\right).
\end{equation}
\end{proof}
\begin{exe}
Check that if $I,J\in\mathcal{I}$ and $I\cap J=\varnothing$, $I\cup J\in \mathcal{I}$, then 
$$\mu(I\cup J)=\mu(I)+\mu(J),$$ i.e. $\mu$ is finitely additive. Hint: Use 
\begin{align*}
I&=I\left(t_1,...,,t_n,\prod_{i=1}^{n}(\alpha_i,\beta_i]\right)\\
J&=J\left(s_1,...,,s_n,\prod_{j=1}^{m}(\gamma_j,\delta_j]\right).
\end{align*}
\end{exe}
A fact of the construction is that $\mu$ is countably additive on $\mathcal{I}$. Now, by the Caratheodory extension construction, $\mu$ induces a unique measure on $\sigma(\mathcal{I})$, the $\sigma$-algebra generated by $\mathcal{I}$. We will denote this ``measure'' again by $\mu$. To prove theorem \ref{Wiener_thm}, we will show that $\sigma(\mathcal{I})=\calB(C_0([0,1]))$, which is the content of the following proposition.
\begin{prop}
$$\sigma(\mathcal{I})=\calB(C_0([0,1])).$$
\end{prop}
\begin{proof}
We already know that $\mathcal{I}\subset \calB(C_0([0,1]))$. Hence $\sigma(\mathcal{I})\subset \calB(C_0([0,1]))$. To show the converse, it suffices to show that for any $\delta>0$, 
\[
\overline{B_\delta(x_0)}:=\{x\in C_0([0,1])\mid \| x-x_0\|\leq \delta\}\subset\sigma(\mathcal{I}).
\]
Fix $\delta>0$ and $x_0\in C_0([0,1])$. Our goal will be to show that 
$$\overline{B_\delta(x_0)}=\bigcap_{N=1}^\infty K_N,$$
where $K_N\in\sigma(\mathcal{I})$. Note that for fixed $t\in[0,1]$ we have 
\begin{equation}
\label{ball}
\overline{B_{\delta}(x_0)}\subset\{x\in C_0([0,1])\mid \vert x(t)-x_0(t)\vert\leq\delta\}.
\end{equation}
Let $\{t_k\}_{k=1}^\infty$ be a dense subset of $[0,1]$ and define 
\[
K_N=\{x\in C_0([0,1])\mid \vert x(t_j)-x_0(t_j)\vert\leq\delta\text{ for $j=1,2,...,N$}\}.
\]
Then by \eqref{ball}, $\overline{B_\delta(x_0)}\subset \bigcap_{N=1}^\infty K_N$. To show the reverse inclusion, we will show that 
$$x\not\in \overline{B_\delta(x_0)}\Longrightarrow x\not\in \bigcap_{N=1}^\infty K_N.$$
Assume that $x\not\in\overline{B_\delta(x_0)}$. Then there is an $s\in[0,1]$ such that 
\[
\vert x(s)-x_0(s)\vert\geq\delta+\delta_1
\]
for some $\delta_1>0$. Now, choose a subsequence $\{t_{k_j}\}$ of $\{t_k\}$ such that $t_{k_j}\to s$ (this can be done since $\{t_k\}$ is dense). Since $x$ and $x_0$ are both continuous, we get 
\begin{align}
x(t_{k_j})&\longrightarrow x(s),\\
x_0(t_{k_j})&\longrightarrow x_0(s).
\end{align}
Thus for large $j$ we get 
\[
\vert x(t_{k_j})-x_0(t_{k_j})\vert\geq\delta+\frac{\delta_1}{2},
\]
and thus $x\not\in \bigcap_{N=1}^\infty K_N$.
Hence, we were able to construct a measure on $\mathcal{B}(C_0([0,1]))$. To complete the proof of theorem \ref{Wiener_thm}, we check that $\mu$ is a probability measure. Indeed, we have 
\[
\mu(C_0([0,1]))=\mu(\ev_1^{-1}(\R))=\frac{1}{\sqrt{2\pi}}\int_\R\ee^{-\frac{x^2}{2}}\dd x=1.
\]
This completes the proof of theorem \ref{Wiener_thm}.
\end{proof}
Next we will compute the Wiener measure of the set
\[
A_{s,t}^{a,b}=\{x\in C_0([0,1])\mid a\leq x(t)-x(s)\leq b\},
\]
where $a,b\in\R$ with $a\leq b$, and $s,t\in[0,1]$ with $0\leq s<t$. Note that $A_{s,t}^{a,b}=P(s,t)^{-1}(E)$, where $E=\{(x,y)\in\R^2\mid a\leq x-y\leq b\}$. Hence 
\begin{align}
\begin{split}
\mu(A_{s,t}^{a,b})&=\frac{1}{\sqrt{(2\pi)^2s(t-s)}}\iint_E\ee^{-\frac{-y^2}{2s}-\frac{(x-y)^2}{2(t-s)}}\dd x\dd y\\
&=\frac{1}{\sqrt{(2\pi)^2s(t-s)}}\int_\R\left(\int_a^b\ee^{-\frac{u^2}{2(t-s)}}\dd u\right)\ee^{-\frac{y^2}{2s}}\dd y\\
\label{integral}
&=\frac{1}{\sqrt{2\pi(t-s)}}\int_a^b\ee^{-\frac{u^2}{2(t-s)}}\dd u
\end{split}
\end{align}
Let us make a short input on \textsf{pushforward of a measure}. Let $(X,\sigma(X),\mu)$ be a measure space, $(Y,\sigma(Y))$ a measurable space and $f:X\to Y$ a measurable map. Then we can define a measure $f_*\mu$ on $(Y,\sigma(Y))$, which is defined as 
\[
f_*\mu(P)=\mu(f^{-1}(P)),\hspace{1cm}P\in\sigma(Y).
\]
This measure $f_*\mu$ is called the pushfoward measure of $\mu$ along $f$. It is easy to check that for any integrable function $\alpha:Y\to\R$ we have 
\[
\int_Y\alpha(y)\dd(f_*\mu(y))=\int_X(f^*\alpha)(x)\dd \mu(x),
\]
where $(f^*\alpha)(x)=\alpha(f(x))$. Define the map $\alpha_{s,t}:C_0([0,1])\to\R$ by $\alpha_{s,t}(x)=x(t)-x(s)$. Then \eqref{integral} implies that $(\alpha_{s,t})_*\mu$, where $\mu$ is the Wiener measure on $C_0([0,1])$, is given by 
\[
(\alpha_{s,t})_*\mu([a,b])=\frac{1}{\sqrt{2\pi (t-s)}}\int_{a}^b\ee^{-\frac{x^2}{2(t-s)}}\dd x.
\]
Thus $(\alpha_{s,t})_*\mu$ is the Gaussian measure on $\R$, which is centered and it has variance $(t-s)$. As a corollary of this discussion we get
\begin{cor}
The following hold.
\begin{enumerate}
\item{$$\int_{C_0([0,1])}(x(t)-x(s))\dd\mu(x)=0,$$
}
\item{$$\int_{C_0([0,1])}(x(t)-x(s))^2\dd\mu(x)=t-s.$$}
\end{enumerate}
\end{cor}
\begin{exe}
Show that 
$$\int_{C_0([0,1])}x(s)x(t)\dd\mu(x)=\min_{s,t\in[0,1]}\{s,t\}.$$
Hint: Assume that $s<t$ and show that 
$$\frac{1}{\sqrt{2\pi s(t-s)}}\iint_{\R^2}xy\ee^{-\frac{x^2}{2s}}\ee^{-\frac{(x-y)^2}{2(t-s)}}\dd x\dd y=s.$$
\end{exe}
\begin{exe}
Compute 
\begin{enumerate}
\item{$$\int_{C_0([0,1])}\left(\int_0^1x(t)\dd t\right)\dd\mu(x),$$
}
\item{$$\int_{C_0([0,1])}\left(\int_0^1x(t)^2\dd t\right)\dd\mu(x).$$
}
\end{enumerate}
Hint: Use Fubini.
\end{exe}

\subsection{Towards nowhere differentiability of Brownian Paths}
Let $h>0$ and $0<\alpha\leq 1$. Define
\begin{itemize}
\item{$$C_h^\alpha(s,t)=\{x\in C_0([0,1])\mid \vert x(t)-x(s)\vert\leq h\vert t-s\vert^\alpha\},$$
}
\item{$$C_h^\alpha(t)=\bigcap_{s\in[0,1]}C_h^\alpha(s,t),$$
}
\item{$$C_h^\alpha=\bigcap_{t\in[0,1]}C_h^\alpha(t).$$
}
\end{itemize}
One can check that $C_h^\alpha(s,t)$ is closed in $C_0([0,1])$ and thus $C_h^\alpha(s,t),C_h^\alpha(t)$ and $C_h^\alpha$ are Borel measurable.
\begin{lem}
\begin{equation}
\label{measure_est}
\mu(C_h^\alpha(s,t))\leq\sqrt{\frac{2}{\pi}}h\vert t-s\vert.
\end{equation}
\end{lem}
\begin{proof}
We can write 
\[
C_h^\alpha(s,t)=\{x\in C_0([0,1])\mid -h\vert t-s\vert^\alpha\leq x(t)-x(s)\leq h\vert t-s\vert^\alpha\}=:A^{-h\vert t-s\vert^\alpha,h\vert t-s\vert^\alpha}_{s,t}.
\]
Assume that $s<t$. Then by \eqref{integral} we have 
\begin{align}
\begin{split}
\mu(C_h^\alpha(s,t))&=\frac{1}{\sqrt{2\pi}}\int_{-h\vert t-s\vert^\alpha}^{h\vert t-s\vert^\alpha}\ee^{-\frac{u^2}{2(t-s)}}\dd u\\
&=\frac{1}{\sqrt{2\pi}}\int_{-h\vert t-s\vert^{\alpha-\frac{1}{2}}}^{h\vert t-s\vert^{\alpha-\frac{1}{2}}}\ee^{-\frac{u^2}{2}}\dd u\\
&\leq \sqrt{\frac{2}{\pi}}h\vert t-s\vert^{\alpha-\frac{1}{2}},
\end{split}
\end{align}
where we used that $\ee^{-\frac{u^2}{2}}\leq 1$.
\end{proof}
\begin{cor}
If $\frac{1}{2}<\alpha\leq 1$, then $\mu(C_h^\alpha(t))=0$ and hence $\mu(C_h^\alpha)=0$.
\end{cor}
\begin{proof}
Let $\{t_k\}\subseteq[0,1]$ such that $t_k\to t$ (for $k\to\infty$). Now $C_h^\alpha(t)\subseteq C_h^\alpha(t,t_k)$. Thus 
$$\mu(C_h^\alpha(t))\leq \mu(C_h^\alpha(t,t_k))\leq \sqrt{\frac{2}{\pi}}h\vert t-t_k\vert^{\alpha-\frac{1}{2}}\xrightarrow{k\to\infty}0.$$
\end{proof}
\begin{prop}
Let $\frac{1}{2}<\alpha\leq 1$. Then 
$$\mu\left(\{x\in C_0([0,1])\mid \text{ $x$ is H\"older continuous of exponent $\alpha$}\}\right)=0.$$
\end{prop}
\begin{proof}
It is clear since $\{x\in C_0([0,1])\mid \text{ $x$ is H\"older continuous of exponent $\alpha$}\}\subseteq \bigcup_{h=1}^\infty C_h^\alpha$.
\end{proof}
\begin{cor} 
$$\mu\left(\{x\in C_0([0,1])\mid \text{ $x$ is differentiable}\}\right)=0.$$
\end{cor}
\begin{proof}
This follows since
$$\{x\in C_0([0,1])\mid \text{ $x$ is differentiable}\}\subseteq\{x\in C_0([0,1])\mid \text{ $x$ is H\"older continuous of exponent $1$}\}.$$
\end{proof}
The following lemma will play an important role when we discuss nowhere differentiability of Brownian paths.
\begin{lem}
\label{measure1}
Let $t\in [0,1]$. Then $\mu(D_t)=0$, where $D_t=\{x\in C_0([0,1])\mid \dot{x}(t)\text{ exists}\}$.
\end{lem}
\begin{proof}
We can easily check that $D_t\subseteq \bigcup_{h=1}^\infty C^\bullet_h(t)$. Moreover, we already know that $\mu(C_h^\bullet(t))=0$ and thus $\mu(D_t)=0$.
\end{proof}
\begin{lem}
\label{Brownian_path}
Define $F$ on $C_0([0,1])\times [0,1]$ by 
\[
F(x,t)=\begin{cases}1,&\text{if $\dot{x}(t)$ exists}\\ 0,& \text{otherwise}\end{cases}
\]
Then $F$ is measurable on $C_0([0,1])\times[0,1]$.
\end{lem}
\begin{proof}
We will show that the set $G=\{(x,t)\mid F(x,t)=1\}$ has measure $0$ with respect to $\mu\times m$, where $m$ is the Lebesgue measure on $[0,1]$. First we observe that $G\subseteq G^*$, where 
\[
G^*=\{(x,t)\in C_0([0,1])\times[0,1]\mid \lim_{n\to\infty}f_n(x,t)\text{ exists}\},
\]
with $f_n(x,t)=\lim_{n\to\infty}\frac{x\left(t-\frac{1}{n}\right)-x(t)}{\frac{1}{n}}$. Moreover, $G^*$ is measurable (since it is the set where a sequence of measurable functions have a imit). Note that 
\[
(\mu^*\times m)(G)\leq (\mu^*\times m)(G^*)=(\mu\times m)(G^*),
\]
where $\mu^*$ is the outer measure associated with the premeasure $\mu$ in the construction of the Wiener measure. Now 
\[
(\mu\times m)(G^*)=\int_0^1\mu(G_t^*)\dd t,
\]
where $G^*_t=\{x\in C_0([0,1])\mid \lim_{n\to\infty}f_n(x,t)\text{ exists}\}$. If we can show that $\mu(G_t^*)=0$, then we see that $(\mu\times m)(G^*)=0\Rightarrow (\mu^*\times m)(G)=0\Rightarrow G$ is measurable. To see that $\mu(G_t^*)=0$, one can show that
\[
G_t^*\subseteq \bigcup_{h=1}^\infty\bigcap_{n=1}^\infty C_h^\bullet\left(t,t+\frac{1}{n}\right),
\]
and use the fact that $\lim_{n\to\infty}\mu(C_h^\bullet\left(t,t+\frac{1}{n}\right))=0$. 
\end{proof}
\begin{thm}[Nowehere differentiable Brownian paths]
With probability $1$, paths $x\in C_0([0,1])$ are differentiable at most on a subset of Lebesgue measure $0$ of $[0,1]$. (In other words: with probability $1$, paths $x\in C_0([0,1]))$ are ``nowhere'' differentiable.)
\end{thm}
\begin{proof}
Let $F$ be defined as in lemma \ref{Brownian_path}. Note that 
\begin{equation}
\int_{C_0([0,1])\times[0,1]}F(x,t)\dd\mu(x)\dd t=\int_0^1\left(\int_{C_0([0,1])}F(x,t)\dd\mu(x)\right)\dd t
=\int_0^1\mu(D_t)\dd t=0,
\end{equation}
by lemma \ref{measure1}, since for fixed $t$, $F(x,t)=\chi_{D_t}$. Thus we get 
\[
\int_{C_0([0,1])}\left(\int_0^1F(x,t)\dd t\right)\dd\mu(x)=0,
\]
whence $\int_0^1F(x,t)\dd t=0$ for almost all $x\in C_0([0,1])$. For such $x$, $F(x,t)=0$ for almost all $t\in[0,1]$ and thus $\dot{x}(t)$ does not exists for almost all $t\in[0,1]$.
\end{proof}
We will state now the following facts without proof.
\begin{itemize}
\item{(Fact 1) $\mu(\{x\in C_0([0,1])\mid x \text{ is H\"older continuous of exponent $\alpha$}\})=1$ for $0\leq \alpha\leq \frac{1}{2}$ (see \cite{HK}).
}
\item{(Fact 2) $\mu(\{x\in C_0([0,1])\mid x \text{ is H\"older continuous of exponent $\frac{1}{2}$}\})=0$ (see \cite{S,MP}).
}
\end{itemize}
\begin{rem}
More generally, we can talk about the Wiener measure $\mu_x$ on $C_x([a,b])=\{\omega:[a,b]\to\R\mid \text{$\omega$ is continuous and $\omega(a)=x$}\}$.
\end{rem}
\begin{rem}
Moreover, there is the Wiener measure $\mu_x^y$ on $$C_x^y([a,b])=\{\omega:[a,b]\to\R\mid \text{$\omega$ is continuous and $\omega(a)=x$, $\omega(b)=y$}\}.$$ This measure is the unique measure on $\calB(C_x^y([0,1]))$ such that for all $t_1,...,t_n\in (a,b)$
\[
\mu_x^y(I(t_1,...,t_n),E)=\int_EK_{b-t_n}(y,x_n)K_{t_n-t_{n-1}}(x_n,x_{n-1})\dotsm K_{t_1-a}(x_1,x)\dd x_1\dotsm \dd x_n,
\]
where $E\subseteq\R^n$ is a measurable set and $K_t(u,v)=\frac{1}{\sqrt{2\pi t}}\ee^{-\frac{1}{2}\frac{(u-v)^2}{t}}$.
\end{rem}
\begin{defn}[Conditional Wiener measure]
The Wiener measure $\mu_x^y$ is called a \textsf{conditional Wiener measure}.
\end{defn}
\begin{rem}
$\mu_x^y$ is not a probability measure. In fact, $$\mu_x^y(C_x^y([a,b]))=\frac{1}{\sqrt{2\pi (b-a)}}\ee^{-\frac{(x-y)^2}{t(b-a)}}.$$
\end{rem}
\begin{rem}
$\mu_x^y$ is called \textsf{conditional Wiener measure} because $\mu_x$ and $\mu_x^y$ fit in the general framework (see \cite{HK,MP}) of a conditional measure 
\[
\mu_x=\int_\R\mu_x^y\dd y.
\]
\end{rem}
\subsection{The Feynman-Kac Formula}
The goal of this subsection is to prove the following theorem.
\begin{thm}[Feynman-Kac]
\label{FK}
Let $V$ be a continuous function on $\R$, which is bounded from below. Let $\widehat{H}_0=-\frac{1}{2}\frac{\dd^2}{\dd x^2}=\frac{1}{2}\Delta$ and $\widehat{H}=\widehat{H}_0+V$. Moreover, assume that $\widehat{H}$ is essentially self adjoint. Then for all $\psi\in L^2(\R)$
\begin{equation}
\label{FK_eq}
\left(\ee^{-t\widehat{H}}\psi\right)(x_0)=\int_{C_{x_0}([0,t])}\psi(x(t))\ee^{-\int_{0}^tV(x(s))\dd s}\dd\mu_{x_0}(x).
\end{equation}
\end{thm}
The main techincal tool, which we are going to use here, is the \textsf{Trotter product formula}, given in the following way. Let $A$ and $B$ be self adjoint operators bounded from below on $\mathcal{H}$. Assume that $\widehat{H}=A+B$ is essentially self adjoint in $D(A)\cap D(B)$. Denote the unique self adjoint extension of $\widehat{H}$ by $\widehat{H}$ again. Then for all $\phi\in\mathcal{H}$ and for all $t\geq 0$
\[
\ee^{-t\widehat{H}}\phi=\lim_{n\to\infty}\left(\left(\ee^{-\frac{t}{n}A}\ee^{-\frac{t}{n}B}\right)^n\phi\right).
\]
\begin{proof}[Proof of Theorem \ref{FK}]
We have 
\begin{equation}
\left(\left(\ee^{-\frac{t}{n}\widehat{H}_0}\ee^{-\frac{t}{n}V}\right)\psi\right)(x_0)=\int_\R K_{\frac{t}{n}}(x_n,x_0)\ee^{-\frac{t}{n}V(x_1)}\psi(x_1)\dd x_1.
\end{equation}
Taking the square of the operator we get 
\begin{equation}
\left(\left(\ee^{-\frac{t}{n}\widehat{H}_0}\ee^{-\frac{t}{n}V}\right)^2\psi\right)(x_0)=\iint_{\R^2} K_{\frac{t}{n}}(x_2,x_1)K_{\frac{t}{n}}(x_1,x_n)\ee^{-\frac{t}{n}(V(x_2)+V(x_1))}\psi(x_1)\dd x_1\dd x_2.
\end{equation}
Taking the $n$-th power of the operator, we get 
\begin{equation}
\label{FK2}
\left(\left(\ee^{-\frac{t}{n}\widehat{H}_0}\ee^{-\frac{t}{n}V}\right)^n\psi\right)(x_0)=\int_{\R^n} K_{\frac{t}{n}}(x_n,x_{n-1})\dotsm K_{\frac{t}{n}}(x_1,x_0)\ee^{-\frac{t}{n}\sum_{j=1}^nV(x_j)}\psi(x_n)\dd x_1\dotsm \dd x_n,
\end{equation}
where $x_j=x\left(\frac{jt}{n}\right)$ and thus $x_n=x(t)$. Then \eqref{FK2} is equal to 
\[
\int_{C_{x_0}([0,t])}\psi(x(t))\ee^{-\frac{t}{n}\sum_{j=1}^nV\left(x\left(\frac{jt}{n}\right)\right)}\dd \mu_{x_0},
\]
and thus 
\begin{equation}
\label{FK3}
\left(\ee^{-t\widehat{H}}\psi\right)(x_0)=\lim_{n\to\infty}\int_{C_{x_0}([0,t])}\psi(x(t))\ee^{-\frac{t}{n}\sum_{j=1}^nV\left(x\left(\frac{jt}{n}\right)\right)}\dd \mu_{x_0}.
\end{equation}
Since 
\[
\lim_{n\to\infty}\ee^{-\frac{t}{n}\sum_{j=1}^nV\left(x\left(\frac{jt}{n}\right)\right)}=\ee^{-\int_0^1V(x(s))\dd s},
\]
it is enough to justify that limit and integral are interchangable in \eqref{FK3}. This can be justified by using the assumption that $V$ is bounded from below and by Lebesgue's dominated convergence theorem. Details are left to the reader.
\end{proof}

\begin{rem}
The Feynman-Kac formula holds (see \cite{S}) for some general $V$.
\end{rem}
\begin{rem}
There is a Feynman-Kac formula with respect to $\mu_x^y$ on $C_x^y([0,t])$ as well (see \cite{GJ}). It simply sais that the integral kernel of $\ee^{-t\widehat{H}}$ is given by 
\[
K_{t}(x,y,\widehat{H})=\int_{C_x^y([0,t])}\ee^{-\int_0^1V(x(s))\dd s}\dd\mu_x^y.
\]
\end{rem}

\section{Gaussian Measures}

\subsection{Gaussian measures on $\R$}
\begin{defn}[Gaussian measure I]
A Borel probability measure $\mu$ on $\R$ is called \textsf{Gaussian} if it is either the Dirac measure $\delta_a$ at $a\in \R$, or it is of the form 
\begin{equation}
\label{Gaussian}
\dd\mu(x)=\frac{1}{\sqrt{2\pi \sigma}}\ee^{-\frac{(x-a)^2}{2\sigma}},
\end{equation}
 where $a\in\R$, and $\sigma>0$. The parameters $a$ and $\sigma$ are called \textsf{mean} and \textsf{variance} of $\mu$ respectively.
\end{defn}
\begin{rem}
\label{centeredGaussian}
If $\mu$ is given by \eqref{Gaussian}, we say $\mu$ is \textsf{nondegenerate Gaussian}. Moreover, if $a=0$, then $\mu$ is called a \textsf{centered Gaussian}.
\end{rem}
\begin{exe}
\label{exGauss}
Check that 
\begin{align}
a&=\int_\R x\dd\mu(x),\\
\sigma&=\int_\R (x-a)^2\dd\mu(x).
\end{align}
\end{exe}
Exercise \ref{exGauss} justifies the names ``mean'' and ``variance'' of the Gaussian measure $\mu$ given by \eqref{Gaussian}.
\begin{exe}
\label{exTrans}
Given a Broel measure $\mu$, define $\widetilde{\mu}:\R\to\mathbb{C}$ by 
\[
\widetilde{\mu}(y)=\int_\R\ee^{\I yx}\dd\mu(x).
\]
Check that $\widetilde{\mu}(y)=\ee^{\I ay-\frac{1}{2}\sigma y^2}$, if $\mu $ is given by \eqref{Gaussian}.
\end{exe}

\begin{defn}[Characteristic Functional I]
The map $\widetilde{\mu}$ defined as in exercise \ref{exTrans} is called the \textsf{charactersitic functional} (or \textsf{Fourier transform}) of $\mu$.
\end{defn}

\begin{exe}
Let $\mu$ be a Borel measure on $\R$. Show that $\mu$ is Gaussian iff 
\begin{equation}
\label{Ftrans}
\widetilde{\mu}(y)=\ee^{\I ay-\frac{1}{2}\sigma y^2}
\end{equation}
for some $a\in\R$ and $\sigma>0$.
\end{exe}

\subsection{Gaussian measures on finite dimensional vector spaces}

\begin{defn}[Gaussian measure II]
A Borel probability measure $\mu$ on $\R^n$ is called \textsf{Gaussian}, if for all linear maps $\alpha:\R^n\to\R$, the pushforward measure $\alpha_*\mu$ is Gaussian on $\R$. 
\end{defn}
This definition is abstract and we will later give a more ``working'' definition of a Gaussian measure.
\begin{rem}
From now on we will identify $(\R^n)^*$ with $\R^n$, using the standard metric on $\R^n$, i.e. a linear map $\alpha:\R^n\to\R$ will be considered a vector $\alpha\in\R^n$.
\end{rem}

\begin{defn}[Characteristic Functional II]
Given  a finite Borel measure $\mu$ on $\R^n$, define $\widehat{\mu}:\R\to\mathbb{C}$ by 
\[
\widehat{\mu}(y)=\int_\R\ee^{\I\langle y,x\rangle}\dd\mu(x).
\]
$\widehat{\mu}$ is called the \textsf{Characteristic functional} (or \textsf{Fourier transform}) of $\mu$.
\end{defn}
\begin{prop}
A Borel measure $\mu$ on $\R^n$ is Gaussian iff 
\begin{equation}
\label{measureG}
\widehat{\mu}(y)=\ee^{-\I\langle y,a\rangle-\frac{1}{2}\langle Ky,y\rangle},
\end{equation}
where $a\in\R^n$ and $K$ is a positive definite symmetric $n\times n$-matrix. In this case, when $\mu$ is nondegenerate, then 
\[
\dd\mu(x)=\frac{1}{\det\left(\frac{K}{2\pi}\right)^\frac{1}{2}}\ee^{-\frac{1}{2}\langle K^{-1}(x-a),K^{-1}(x-a)\rangle}\dd x.
\]
\end{prop}
\begin{proof}
Given a Borel measure $\mu$ on $\R^n$ and a linear map $\alpha:\R^n\to\R$, we get 
\begin{align}
\begin{split}
\widehat{\alpha_*\mu}(t)=\int_\R\ee^{\I ts}\dd(\alpha_*\mu)(s)&=\int_\R\ee^{\I t\alpha (x)}\dd\mu(x)\\
&=\int_{\R^n}\ee^{\I\langle t\alpha,x\rangle}\dd\mu(x)\\
&=\widehat{\mu}(t\alpha),
\end{split}
\end{align}
where we have used in the second equality that 
\[
\int_X (F^*\alpha)(x)\dd x=\int_Yf(y)\dd(\alpha_*\mu)(y).
\]
Assume that $\widehat{\mu}$ has the form \eqref{Ftrans}. Then 
\begin{equation}
\widehat{\alpha_*\mu}(t)=\widehat{\mu}(t\alpha)=\ee^{\I\langle t\alpha,a\rangle-\frac{1}{2}\langle K(t\alpha),t\alpha\rangle}=\ee^{\I t\langle\alpha,a\rangle-\frac{1}{2}t^2\langle K\alpha,\alpha\rangle}.
\end{equation}
By exercise \ref{exTrans}, $\widehat{\alpha_*\mu}$ is Gaussian on $\R$. Conversely, assume that $\alpha_*\mu$ is Gaussian on $\R$ for all linear maps $\alpha:\R^n\to\R$. By exercise \ref{exTrans} we get 
\[
\widehat{\alpha_*\mu}(t)=\ee^{\I ta(\alpha)-\frac{1}{2}\sigma(\alpha)t^2}.
\]
Moreover, by exercise \ref{exGauss}, we get 
\begin{align}
a(\alpha)&=\int_\R t\dd(\alpha_*\mu)(t),\\
\sigma(\alpha)&=\int_\R (t-a(\alpha)^2\dd(\alpha_*\mu)(t).
\end{align}
We can check that the application $\alpha\mapsto a(\alpha)$ defines a linear map $\R^n\to\R$. and hence it can be identified with $a\in\R^n$ as $a(\alpha)=\langle a,\alpha\rangle$. Moreover, the application $\alpha\mapsto\sigma(\alpha)$ defines a quadratic form on $\R^n$. Hence, there is a symmetric $n\times n$-matrix $K$ such that $\sigma(\alpha)=\langle K\alpha,\alpha\rangle$. Thus, $\sigma(\alpha)>0$ for all $\alpha\in\R^n$ implies that $K$ is a positive matrix. The last part of the proof is left as an exercise\footnote{It essentially follows from the one dimensional case and diagonalization of $K$.}. 
\end{proof}
Hence, we saw that this abstract definition of a Gaussian measure on $\R^n$ is equivalent to the usual notion of Gaussian measure.

\begin{exe}
Let $\mu$ be a Gaussian measure on $\R^n$ of the form 
\[
\dd\mu(x)=\det\left(\frac{K}{2\pi}\right)^\frac{1}{2}\ee^{-\frac{1}{2}\langle K(x-a),(x-a)\rangle}\dd x.
\]
Check that 
\[
a=\int_\R x\dd\mu(x)=\left(\int_\R x_1\dd\mu(x_1),..., \int_\R x_n\dd\mu(x_n)\right),
\]
and 
$$K^{-1}_{ij}=\int_{\R^n}(x_i-a_i)(x_j-a_j)\dd\mu(x).$$
\end{exe}
\begin{defn}[Covariance operator]
The vector $a\in\R^n$ is called the \textsf{mean} of the Gaussian measure and the matrix $K^{-1}$ is called the \textsf{covariance operator} of $\mu$. When the mean of a Gaussian measure is $0$, then it is called a \textsf{centered Gaussian} (see remark \ref{centeredGaussian}).
\end{defn}
\begin{prop}
Let $\mu$ be a centered Gaussian measure on $\R^n$ of the form $$\dd\mu(x)=\det\left(\frac{K}{2\pi}\right)^\frac{1}{2}\ee^{-\frac{1}{2}\langle Kx,x\rangle}\dd x.$$ Then 
\begin{enumerate}
\item{For all $\lambda\in\mathbb{C}^n$,
\[
\int_{\R^n}\ee^{\langle\lambda,x\rangle}\dd\mu(x)=\ee^{\frac{1}{2}\langle K^{-1}\lambda,\lambda\rangle}.
\]
}
\item{
\[
\int_{\R^n}f(x-\sqrt{t}y)\dd\mu(y)=\left(\ee^{\frac{t}{n}L^\mu}f\right)(x),
\]
where $L^\mu=\sum_{i=1}^nK^{-1}_{ij}\frac{\partial}{\partial x_i}\frac{\partial}{\partial y_j}$. Moreover, 
\[
\int_{\R^n}f(y)\dd\mu(y)=\left(\ee^{\frac{1}{2}L^\mu}f\right)(0).
\]
}
\item{
\[
\int_{\R^n}p(x)\dd\mu(x)=p(D_\lambda)\ee^{-\frac{1}{2}\langle A^{-1}\lambda,\lambda\rangle}\Big|_{\lambda=0},
\]
where $p(D_\lambda)$ is a polynomial in derivatives in $\lambda_i$-directions $\frac{\partial}{\partial \lambda_i}$ corresponding to the polynomial map $p(x)$, i.e. if $p(x)=x_1x_2$, then $p(D_\lambda)=\frac{\partial}{\partial \lambda_1}\frac{\partial}{\partial\lambda_2}$.
}
\end{enumerate}
\end{prop}

\begin{proof}
We prove each point separately:
\begin{enumerate}
\item{We have 
\[
\int_{\R^n}\ee^{\langle \lambda,x\rangle}\dd\mu(x)=\det\left(\frac{K}{2\pi}\right)^{\frac{1}{2}}\int_{\R^n}\ee^{\langle \lambda,x\rangle}\ee^{-\frac{1}{2}\langle Kx,x\rangle}\dd x=\frac{\det\left(\frac{K}{2\pi}\right)^{\frac{1}{2}}}{\det\left(\frac{K}{2\pi}\right)^{\frac{1}{2}}}\ee^{\frac{1}{2}\langle K^{-1}\lambda,\lambda\rangle}=\ee^{\frac{1}{2}\langle K^{-1}\lambda,\lambda\rangle}.
\]
}
\item{It is sufficient to check that $f(x)$ is of the form $\ee^{\langle\lambda,x\rangle}$ with $\lambda\in\mathbb{C}^n$ as these function are dense. For $f(x)=\ee^{\langle\lambda,x\rangle}$ we get 
\[
f(x-\sqrt{t}y)=\ee^{\langle\lambda,x\rangle}\ee^{\langle -\sqrt{t}\lambda,y\rangle},
\]
and thus 
\[
\int_{\R^n}f(x-\sqrt{t}y)\dd\mu(y)=\ee^{\langle\lambda,x\rangle}\int_{\R^n}\ee^{\langle-\sqrt{t}\lambda,y\rangle}\dd\mu(y)=\ee^{\langle\lambda,x\rangle}\ee^{\frac{t}{2}\langle K^{-1}\lambda,\lambda\rangle}.
\]
On the other hand 
\[
L^\mu\left(\ee^{\langle\lambda,x\rangle}\right)=\langle K^{-1}\lambda,\lambda\rangle\ee^{\langle\lambda,x\rangle}
\]
which means 
\[
\ee^{\frac{t}{2}L^\mu}\left(\ee^{\langle\lambda,x\rangle}\right)=\ee^{\frac{t}{2}\langle K^{-1}\lambda,\lambda\rangle}\ee^{\langle\lambda,x\rangle}.
\]
Thus we have 
\[
\int_\R f(x-\sqrt{t}y)\dd\mu(y)=\left(\ee^{\frac{t}{2}L^\mu}f\right)(x),
\]
when $f(x)=\ee^{\langle \lambda,x\rangle}$. The second part can be verified in a similar way.
}
\item{Left as an exercise.}
\end{enumerate}
\end{proof}
\begin{ex}
Consider 
\[
\int_{\R^n}x_ix_j\dd\mu(x)=K^{-1}_{ij}.
\]
More generally, 
\[
\int_{\R^n}\langle u,x\rangle\langle v,x\rangle\dd\mu(x)=\langle K^{-1} u,v\rangle.
\]
Graphically, it can be represented as 
\begin{figure}[!ht]
\centering
\tikzset{
particle/.style={thick,draw=black},
particle2/.style={thick,draw=black, postaction={decorate},
    decoration={markings,mark=at position .9 with {\arrow[black]{triangle 45}}}},
gluon/.style={decorate, draw=black,
    decoration={coil,aspect=0}}
 }
\begin{tikzpicture}[x=0.04\textwidth, y=0.04\textwidth]
\node[](1) at (-0.5,0){$u$};
\node[](2) at (5.5,0){$v$};
\draw[fill=black] (0,0) circle (.07cm);
\draw[fill=black] (5,0) circle (.07cm);
\draw[particle] (1)--(2){};
\node[](3) at (2.5,0.5){$K^{-1}$};
\end{tikzpicture}
\end{figure}
\end{ex}
\begin{ex}
Consider 
\[
\int_{\R^n}\left(\prod_{i=1}^4\langle u_i,x\rangle\right)\dd\mu(x)=\langle K^{-1}u_1,u_2\rangle\langle K^{-1}u_3,u_4\rangle+\langle K^{-1}u_1,u_3\rangle\langle K^{-1}u_2,u_4\rangle+\langle K^{-1}u_1,u_4\rangle\langle K^{-1}u_2,u_3\rangle.
\]
Graphically, it can be represented as the sum of 

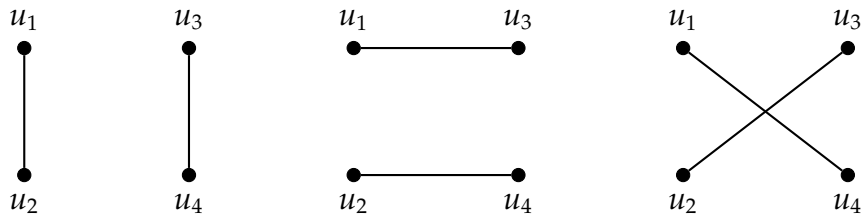
\begin{figure}[!ht]
\centering
\tikzset{
particle/.style={thick,draw=black},
gluon/.style={decorate, draw=black,
    decoration={coil,aspect=0}}
 }
\begin{tikzpicture}[node distance=1.5cm and 2cm]
\coordinate[vertex, label=above: $u_1$](v1);
\coordinate[vertex, below=of v1, label=below: $u_2$](v2);
\coordinate[vertex, right=of v2, label=below: $u_4$](v3);
\coordinate[vertex, right=of v1, label=above: $u_3$](v4);
\coordinate[vertex, right=of v4, label=above: $u_1$](v5);
\coordinate[vertex, right=of v3, label=below: $u_2$](v6);
\coordinate[vertex, right=of v5, label=above: $u_3$](v7);
\coordinate[vertex, right=of v6, label=below: $u_4$](v8);
\coordinate[vertex, right=of v8, label=below: $u_2$](v9);
\coordinate[vertex, right=of v7, label=above: $u_1$](v10);
\coordinate[vertex, right=of v10, label=above: $u_3$](v11);
\coordinate[vertex, right=of v9, label=below: $u_4$](v12);
\draw[fill=black](v1) circle (.07cm);
\draw[fill=black](v2) circle (.07cm);
\draw[fill=black](v3) circle (.07cm);
\draw[fill=black](v4) circle (.07cm);
\draw[fill=black](v5) circle (.07cm);
\draw[fill=black](v6) circle (.07cm);
\draw[fill=black](v7) circle (.07cm);
\draw[fill=black](v8) circle (.07cm);
\draw[fill=black](v9) circle (.07cm);
\draw[fill=black](v10) circle (.07cm);
\draw[fill=black](v11) circle (.07cm);
\draw[fill=black](v12) circle (.07cm);
\draw[particle] (v1) -- (v2); 
\draw[particle] (v4) -- (v3);
\draw[particle] (v5) -- (v7);
\draw[particle] (v6) -- (v8);
\draw[particle] (v10) -- (v12);
\draw[particle] (v9) -- (v11);
\end{tikzpicture}
\end{figure}
where each edge represents $K^{-1}$.
\end{ex}
In general, there is the following theorem.
\begin{thm}[Wick]
\label{WIck}
\[
\int_{\R^n}\prod_{i=1}^k\langle u_i,x\rangle\dd\mu(x)=\begin{cases}\sum\langle K^{-1}u_{j_1},u_{j_2}\rangle\dotsm\langle K^{-1} u_{j_{m-1}},u_{j_m}\rangle,&\text{if $k$ even}\\0,&\text{otherwise}\end{cases}
\]
\end{thm}
\begin{exe}
Prove theorem \ref{WIck}.
\end{exe}
\subsection{Gaussian measures on real seperable Hilbert spaces} 

Let $\mathcal{H}$ be a real seperable Hilbert space.
\begin{defn}[Borel measure on $\mathcal{H}$]
A \textsf{Borel measure} $\mu$ on $\mathcal{H}$ is a measure defined on $\calB(\mathcal{H})$, which is the Borel $\sigma$-algebra of $\mathcal{H}$.
\end{defn}
In the previous subsection, we saw that a Gaussian measure on a finite dimensional vector space $V$ is determined by $a\in V$ and a positive symmetric matrix $K^{-1}$, called the covariance of the Gaussian mean. In this section we will see whether this is the case in the infinite dimensional case as well. Let $\mu$ be a Borel measure on $\mathcal{H}$. Define an operator $S_\mu$ on $\mathcal{H}$ by 
\begin{equation}
\label{inner1}
\langle S_\mu(x),y\rangle=\int_\mathcal{H}\langle x,z\rangle\langle y,z\rangle\dd\mu(z)
\end{equation}
\begin{rem}
It may happen that $S_\mu$ does not exists.
\end{rem}
Let us recall some background material.
\begin{enumerate}
\item{(Trace class operators) Let $A:\mathcal{H}\to\mathcal{H}$ be a bounded operator. We define the \textsf{squareroot} of $A$ by 
\[
\vert A\vert:=\sqrt{A^*A},
\]
which exists by the spectral theorem. Note that $\vert A\vert \geq 0$. Let $A$ be a nonegative operator on $\mathcal{H}$. Then 
\[
\sum_{n=1}^\infty\langle Ae_n,e_n\rangle
\]
is independent of the choice of an orthonormal basis $\{e_n\}$. In this case, one defines the \textsf{trace} of $A$ as 
\[
Tr(A)=\sum_{n=1}^\infty\langle Ae_n,e_n\rangle.
\]
\begin{defn}[Trace class]
The operator $A$ is called \textsf{trace class} if $Tr(\vert A\vert)<\infty$.
\end{defn}
If $A$ is a trace class operator, then $\sum_{n=1}^\infty\langle Ae_n,e_n\rangle$ does not depend on the choice of an orthonormal basis $\{e_n\}$. In this case, we define $Tr(A)=\sum_{n=1}^\infty\langle Ae_n,e_n\rangle$.
}
\item{(Bilinear forms/quadratic forms) A \textsf{bilinear form} $B$ with domain $D(B)$ is a bilinear map 
\begin{align*}
B\colon D(B)\times D(B)&\longrightarrow \R\\
(x,y)&\longmapsto B(x,y),
\end{align*}
where $D(B)$ is a dense subspace of $\mathcal{H}$. Given a bilinear form $B$ on $\mathcal{H}$, we can define a \textsf{quadratic form} $q(x)=B(x,x)$. A bilinear form $B$ is \textsf{bounded} if there is some $\varepsilon>0$ such that for all $x,y\in D(B)$
\[
\vert B(x,y)\vert\leq \varepsilon \| x\|\|y\|.
\]
We call $B$ \textsf{symmetric} if $B(x,y)=B(y,x)$ for all $x,y$. Moreover, $B$ is called \textsf{positive (definite)} if $q(x)\geq 0$ (and $q(x)=0$ iff $x=0$) for all $x$. If $B$ is a bounded, positive and symmetric bilinear form, then there is a bounded linear operator $S_B:\mathcal{H}\longrightarrow\mathcal{H}$ such that $B(x,y)=\langle S_B(x),y\rangle$. 
}
\end{enumerate}
This is the end of the background materials. Next, we want to investigate when $S_\mu$ exists. We first need some notation. We define
\[
\mathcal{T}:=\{\text{Trace class, positive, self adjoint operators on $\mathcal{H}$}\}.
\]
\begin{prop}
$$S_\mu\in\mathcal{T}\Longleftrightarrow\int_\mathcal{H}\|x\|^2\dd\mu(x)<\infty.$$
\end{prop}
\begin{proof}
Assume $S_\mu\in\mathcal{T}$. Let $\{e_n\}$ be an orthonormal basis of $\mathcal{H}$. Then 
\[
Tr(S_\mu)=\sum_{n=1}^\infty\langle S_\mu(e_n),e_n\rangle=\sum_{n=1}^\infty\int_\mathcal{H}\langle x,e_n\rangle^2\dd\mu(x)=\int_\mathcal{H}\sum_{n=1}^\infty\langle x,e_n\rangle^2\dd\mu(x)=\int_\mathcal{H}\|x\|^2\dd\mu(x),
\]
by the monotone convergence theorem. Conversely, assume $\int_{\mathcal{H}}\|x\|^2\dd\mu(x)<\infty$. Then, define 
\[
B(x,y)=\int_\mathcal{H}\langle x,z\rangle\langle y,z\rangle\dd\mu(z).
\]
Then 
\[
\vert B(x,y)\vert=\left\vert\int_\mathcal{H}\langle x,z\rangle\langle y,z\rangle\dd\mu(z)\right\vert\leq\|x\|\|y\|\int_\mathcal{H}\|z\|^2\dd\mu(z),
\]
and thus $B$ is a bounded bilinear form. Moreover, $B$ is symmetric and positive. Hence, there is a positive self adjoint operator $S_\mu$ such that $B(x,y)=\langle S_\mu(x),y\rangle$. Now, we can check 
\[
\sum_{n=1}^\infty\langle S_\mu(e_n),e_n\rangle=\int_\mathcal{H}\|x\|^2\dd\mu(x)<\infty
\]
for any orthonormal basis $\{e_n\}$. Thus $S_\mu\in\mathcal{T}$.
\end{proof}

\subsubsection{Characteristic Functionals}

\begin{defn}[Positive definite function]
A function $\phi:\mathcal{H}\to\mathbb{C}$ is called a \textsf{positive definite} if for all $c_1,...,c_n\in\mathbb{C}$ and $h_1,...,h_n\in\mathcal{H}$ with $n=1,2,...$ we have 
\begin{equation}
\label{char_fun}
\sum_{j,k=1}^nc_k\phi(h_k-h_j)\bar c_j\geq 0.
\end{equation}
\end{defn}

\begin{defn}[Characteristic functional III]
Let $\mu$ be a Borel measure on $\mathcal{H}$. The \textsf{characteristic functional} (or \textsf{Fourier transform}) $\widehat{\mu}$ of $\mu$ is a function $\widehat{\mu}:\mathcal{H}\to\mathbb{C}$ defined by 
\begin{equation}
\label{char_functional}
\widehat{\mu}(y)=\int_\mathcal{H}\ee^{\I\langle y,x\rangle}\dd\mu(x).
\end{equation}
\end{defn}
\begin{rem}
It is easy to check that:
\begin{enumerate}
\item{$\vert\widehat{\mu}(x)\vert\leq \mu(\mathcal{H})$ for all $x\in \mathcal{H}$.
}
\item{If $\mu$ is a probability measure, then $\widehat{\mu}(0)=1$.}
\item{If $\mu$ is a finite measure, then $\widehat{\mu}$ is uniformly continuous on $\mathcal{H}$.
}
\end{enumerate}
\end{rem}
\begin{lem}
Let $\mu$ be a Borel measure on $\mathcal{H}$. Then $\widehat{\mu}$ is a positive definite functional on $\mathcal{H}$.
\end{lem}

\begin{proof}
Let $h_1,...,h_n\in\mathcal{H}$ and $c_1,...,c_n\in\mathbb{C}$. Then we get 
\begin{align*}
\sum_{j,k=1}^nc_j\widehat{\mu}(h_j-h_k)\bar c_k&=\int_\mathcal{H}\sum_{j,k=1}^n c_j\ee^{\I\langle h_j,x\rangle}\ee^{-\I\langle h_k,x\rangle}\bar c_k\dd\mu(x)\\
&=\int_\mathcal{H}\sum_{j,k=1}^nc_j\ee^{\I\langle h_j,x\rangle}\overline{\ee^{\I\langle h_k,x\rangle}c_k}\dd\mu(x)\\
&=\int_\mathcal{H}\left\vert\sum_{j=1}^nc_j\ee^{\I\langle h_j,x\rangle}\right\vert^2\dd\mu(x)\geq 0.
\end{align*}
\end{proof}

\begin{defn}[Gaussian measure III]
A Borel measure $\mu$ on $\mathcal{H}$ is called a \textsf{Gaussian measure on} $\mathcal{H}$ if for all $h\in\mathcal{H}$ we get that $(\alpha_h)_*\mu$ is a Gaussian measure on $\R$, wehere $\alpha_h:\mathcal{H}\to\R$ is given by $\alpha_h(x)=\langle h,x\rangle$.
\end{defn}

\begin{lem}
Let $\mu$ be a Gaussian measure on $\mathcal{H}$. Then there are functions $m$ and $\sigma$ on $\mathcal{H}$ such that $\widehat{\mu}(y)=\ee^{\I m(y)-\frac{1}{2}\sigma(y)}$.
\end{lem}

\begin{proof}
Recall that $\widehat{(\alpha_h)_*\mu}(t)=\widehat{\mu}(th)$. Since $(\alpha_h)_*\mu$ is a Gaussian measure, we have 
\[
\widehat{(\alpha_h)_*\mu}(t)=\ee^{\I m(h)t-\frac{1}{2}t^2\sigma(h)},
\]
and thus 
\[
\widehat{\mu}(h)=\widehat{(\alpha_h)_*\mu}(1)=\ee^{\I m(h)-\frac{1}{2}\sigma(h)}.
\]
\end{proof}
\begin{exe}
Check that 
\begin{align}
m(y)&=\int_\mathcal{H}\langle x,y\rangle \dd\mu(x)\\
\sigma(y)&=\int_\mathcal{H}\langle y,x\rangle^2\dd\mu(x)
\end{align}
\end{exe}

\begin{thm}[Bochner-Kolmogorov-Milnor-Prokhorov]
Let $\phi$ be a positive definite functional on $\mathcal{H}$. Then $\phi$ is a characteristic functional of a Borel probability measure $\mu$ on $\mathcal{H}$ if and only if 
\begin{enumerate}
\item{$\phi(0)=0$}
\item{for all $\varepsilon>0$ there is an $S_\varepsilon\in \mathcal{T}$ such that
$$1-Re(\phi(x))\leq\langle S_\varepsilon x,x\rangle+\varepsilon$$
for all $x\in\mathcal{H}$.
}
\end{enumerate}
\end{thm}
\begin{proof}
See \cite{HK}.
\end{proof}

\begin{thm}[Prokhorov]
The following hold:
\begin{enumerate}
\item{Let $\mu$ be a Gaussian measure on $\mathcal{H}$. Then $S_\mu\in\mathcal{T}$.
}
\item{Let $m\in\mathcal{H}$ and $S\in\mathcal{T}$. Then $\phi(x)=\ee^{\I\langle m,x\rangle-\frac{1}{2}\langle Sx,x\rangle}$ is the characteristic functional of a Gaussian measure.
}
\end{enumerate}
\end{thm}
\begin{proof}
We will only consider the centered Gaussian measure, i.e. $m(y)=0$, i.e. $\widehat{\mu}(x)=\ee^{-\frac{1}{2}\sigma(x)}$. We want to show that 
\begin{equation}
\label{norm_like}
\int_\mathcal{H}\| x\|^2\dd\mu(x)<\infty.
\end{equation}
The idea now is to try to find some $S\in\mathcal{T}$ and $C_S>0$ such that 
\begin{equation} 
\label{eq1}
\int_\mathcal{H}\langle x,y\rangle^2\dd\mu(y)\leq C_S\langle Sx,x\rangle,
\end{equation}
for all $x\in\mathcal{H}$. Before we discuss a construction of $S$, let us observe why \eqref{eq1} implies \eqref{norm_like}. Let $\{e_n\}$ be an orthonormal basis of $\mathcal{H}$. Then by \eqref{eq1} 
\[
 \sum_{n=1}^\infty\int_\mathcal{H}\langle e_n,y\rangle^2\dd\mu(y)\leq C_S\sum_{n=1}^\infty\langle Se_n,e_n\rangle=C_S Tr(S).
\]
This implies that
\[
\int_\mathcal{H}\| y\|^2\dd\mu(y)=\int_\mathcal{H}\sum_{n=1}^\infty\langle e_n,y\rangle^2\dd\mu(y)=\sum_{n=1}^\infty\int_\mathcal{H}\langle e_n,y\rangle^2\dd\mu(y)\leq C_S Tr(S)<\infty. 
\]
Hence, our goal will be to construct $S$ such that \eqref{eq1} holds. Since $\mu$ is a probability measure, by the previous theorem we get that for all $\varepsilon>0$ there is some $S_\varepsilon\in\mathcal{T}$ such that 
\begin{equation} 
\label{eq2}
1-\widehat{\mu}(x)\leq \langle S_\varepsilon x,x\rangle+\varepsilon, 
\end{equation} 
for all $x\in\mathcal{H}$. Assume now that $\ker(S_\varepsilon)=\{0\}$. In this case we claim that for all $x\in\mathcal{H}\setminus\{0\}$, we have 
\begin{equation}
\label{eq3}
\int_\mathcal{H}\langle x,y\rangle^2\dd\mu(x)\leq \frac{4}{\varepsilon}\log\left(\frac{1}{1-2\varepsilon}\right)\langle S_\varepsilon x,x\rangle.
\end{equation} 
Obviously \eqref{eq3} implies \eqref{eq1}. To verify \eqref{eq3}, we proceed as follows: If $y\in\mathcal{H}$ such that $\langle S_\varepsilon y,y\rangle<\varepsilon$, then from \eqref{eq2} we get $\sigma(y)\leq2\log\left(\frac{1}{1-2\varepsilon}\right)$. Given $x\in\mathcal{H}\setminus\{0\}$, take $y=\left(\frac{\varepsilon}{2\langle S_\varepsilon x,x\rangle}\right)^{\frac{1}{2}}x$. Then we can check that $\langle S_\varepsilon y,y\rangle<\varepsilon$ and hence we have 
\begin{equation}
\label{var1}
\sigma(y)\leq 2\log\left(\frac{1}{1-2\varepsilon}\right).
\end{equation}
Note that we use $\ker(S_\varepsilon)=\{0\}$ to define $y$. Also, we can check that 
\[
\sigma(y)=\frac{\varepsilon}{2\langle S\varepsilon x,x\rangle}\sigma(x).
\]
Thus for $x\in\mathcal{H}\setminus\{0\}$, we get from \eqref{var1} that 
\[
\sigma(x)\leq \frac{4}{\varepsilon}\log\left(\frac{1}{1-2\varepsilon}\right)\langle S_\varepsilon x,x\rangle.
\]
Now using $\sigma(x)=\int_\mathcal{H}\langle x,y\rangle^2\dd\mu(y)$, we verify \eqref{eq3} when $\ker(S_\varepsilon)=\{0\}$. If $\ker(S_\varepsilon)=\{0\}$, then we can construct $S\in\mathcal{T}$ with $S_\varepsilon\leq S$ and $\ker(S)=\{0\}$ as follows: Let $\{\lambda_n\}$ be positive eigenvalues of $S_\varepsilon$ and $\phi_n$ eigenvectors corresponding to $\lambda_n$ such that $\|\phi_n\|=1$ and $\phi_n\perp \phi_m$ for $m\not=n$. Moreover, let $\{\psi_j\}$ be an orthonormal basis of $\ker(S_\varepsilon)$. Then $\{\phi_n,\psi_j\}$ form an orthonormal basis of $\mathcal{H}$. Define the map 
\begin{align*}
S:\mathcal{H}&\longrightarrow\mathcal{H}\\
x&\longmapsto S(x)=\sum_n\lambda_n\langle \phi_n,x\rangle\phi_n+\sum_j\frac{1}{j^2}\langle \psi_j,x\rangle\psi_j.
\end{align*}
Then we can check that $S\in\mathcal{T}$, $\ker(S)=\{0\}$ and thus \eqref{eq2} holds if we replace $S_\varepsilon$ by $S$. Hence repeating the argument above, \eqref{eq3} holds for $S$. This completes the proof.
\end{proof}

Now let $\mathcal{H}$ be a seperable Hilbert space. Let $\F$ be the set of finite rank projections of $\mathcal{H}$, i.e. $p\in\F$ iff $p:\mathcal{H}\to\mathcal{H}$ is a projection and $\dim p(\mathcal{H})<\infty$. We define the set 
\[
\calR=\{p^{-1}(B)\mid p\in\F,B\subseteq p(\mathcal{H}),\text{$B$ is Borel measurable}\}.
\]
Then it is easy to check that $R$ is an algebra. However, $\calR$ is not a sigma algebra, which can be seen as follows. Let $\overline{B(0,1)}$ be the closed unit ball in $\mathcal{H}$. When $\mathcal{H}$ is infinite dimensional $\overline{B(0,1)}$ is not a cylinder set, i.e. $\overline{B(0,1)}\not\in \calR$, as $C\in \calR$ implies that $C$ is unbounded. We claim that $\overline{B(0,1)}$ can be written as countable intersections of elements of $\calR$. Let $\{h_n\}$ be a countable dense subset of $\mathcal{H}$ with $h_n\not=0$ for all $n$. Moreover, for $N\in\N$, we define the set
\[
K_N=\{h\in\mathcal{H}\mid \vert\langle h,h_n\rangle\vert\leq\|h_n\|,\forall n=1,2,...,N\}.
\]
\begin{exe}
Show that $K_N\in \calR$ for all $N\in\N$.
\end{exe}
It is easy to see that $\overline{B(0,1)}\subseteq \bigcap_{N=1}^\infty K_N$. Assume that $h\not\in \overline{B(0,1)}$. Then there is some $h'\in\mathcal{H}\setminus\{0\}$ such that $\frac{\vert\langle h,h'\rangle\vert}{\| h'\|}\geq \delta+1$ for some $\delta> 0$. Choose then a subsequence $\{h_{n_k}\}$ such that $h_{n_k}\to h'$. Then $\frac{\vert\langle h,h_{n_k}\rangle\vert}{\|h_{n_k}\|}\to \frac{\vert\langle h,h'\rangle\vert}{\|h'\|}$ and thus 
\[
\frac{\vert\langle h,h_{n_k}\rangle\vert}{\|h_{n_k}\|}\geq \delta+1
\]
as $k\to \infty$. This shows that $h\not\in\bigcap_{N=1}^\infty K_N$ and hence we have showed that $\bigcap_{N=1}^\infty K_N\subseteq \overline{B(0,1)}$. This means $\overline{B(0,1)}=\bigcap_{N=1}^\infty K_N$. Next we define a finitely additive measure $\mu$ on $\calR$ as follows. Let $p\in \F$ and $B$ be a Borel subset of $p(\mathcal{H})$ and $\dim p(\mathcal{H})=n$. Define 
\[
\mu(p^{-1}(B))=\frac{1}{(2\pi)^{\frac{n}{2}}}\int_B \ee^{-\frac{1}{2}\| x\|^2}\dd x.
\]
\begin{exe}
Show that $\mu$ is a finitely additive measure on $R$.
\end{exe}
\begin{exe}
Show directly that $\mu$ can not be countably additive.
\end{exe}
Hence, there is no hope to try to construct the standard Gaussian measure on $\mathcal{H}$ (in the infinite dimensionl case the identity operator is not a trace class operator). We ask ourself whether there is a wayto make sense of the standard Gaussian measure on $\mathcal{H}$. The answer is \textsf{yes}. There is a way to understand the \textsf{standard Gaussian} measure on $\mathcal{H}$. The idea is to \textsf{expand} $\mathcal{H}$ so that it supports a countably additive Gaussian measure.

\subsection{Standard Gaussian measure on $\mathcal{H}$}
How do we \textsf{expand} $\mathcal{H}$? The technical tool we use here is Kolmogorov's theorem. Let us briefly recall this without a proof. For this, let $\{X_i\}_{i\in \mathcal{I}}$ be a family of topological spaces. Assume that for each $I\subseteq \mathcal{I}$ finite, we have a Borel probability measure $\mu_I$ on $X_I:=\prod_{i\in I}X_i$. Given $J\subseteq I\subseteq \mathcal{I}$, with $I$ finite, let $\pi_{IJ}:X_I\to X_J$ denote the projection onto the first $J$ coordinates. 
\begin{defn}[Compatible family]
The family $\{X_I,\mu_J\}_{I\subseteq \mathcal{I},\atop\text{$I$ finite}}$ is said to form a \textsf{compatible family} if for all $J\subseteq I$ we have $(\pi_{IJ})_*\mu_I=\mu_J$.
\end{defn}
\begin{thm}[Kolmogorov]
\label{Kolmogorov}
Let $\{X_I,\mu_J\}$ be a compatible family. Then there is a unique probability measure $\mu_\mathcal{I}$ on $X_\mathcal{I}=\prod_{i\in\mathcal{I}}X_i$ and measurable maps $\pi_I:X_\mathcal{I}\to X_I$ for $I\subseteq\mathcal{I}$ finite such that $(\pi_I)_*\mu_\mathcal{I}=\mu_I$.
\end{thm}
To apply theorem \ref{Kolmogorov} in our sitation, we proceed as follows. Let $\{e_n\}$ be an orthonormal basis of $\mathcal{H}$. Define a measure $\mu_n$ on $\R^n$ by 
\[
\mu_n(B)=\mu(p_n^{-1}(B)),
\]
where $p_n:\mathcal{H}\to span \{e_1,...,e_n\}\cong\R^n$ is the projection and $\mu$ the cylindrical measure defined before. Then it is easy to check that $\{\R^n,\mu_n\}$ form a compatible family of probability measures. Hence, by theorem \ref{Kolmogorov} there is a probability space $(\Omega,\widetilde{\mu})$ and random variables $\xi_1,...,\xi_n$ on $\Omega$ such that 
\begin{multline}
\label{}
\widetilde{\mu}\left(\{\omega\in\Omega\mid(\xi_1(\omega),...,\xi_n(\omega))\in B,B\subseteq\R^n\text{ Borel measurable}\}\right)\\
=\mu(\{ h\in\mathcal{H}\mid(\langle h,e_1\rangle,...,\langle h,e_n\rangle)\in B\})(\text{or simply $\mu_n(B)$})
\end{multline}
\begin{lem}
The $\{\xi_i\}$ are independent and identitcally distributed random variables with mean $0$ and  variance $1$.
\end{lem}
Note that using the $\xi_i$s we can define $\mathcal{H}$-valued random variables $X_n$ by 
\begin{align*}
X_n:\Omega&\longrightarrow \mathcal{H}\\
\omega&\longmapsto X_n(\omega)=\sum_{i=1}^n\xi_i(\omega)e_n.
\end{align*}
Moreover, 
\begin{equation}
\label{proj1}
(p_n\circ X_n)_*\widetilde{\mu}=\mu_n.
\end{equation}
If $\{X_n\}$ \textsf{converges in probability} (convergence in measure), then it would induce a random variable $X:\Omega\to\mathcal{H}$ and hence we would get a measure $X_*\widetilde{\mu}$ on $\mathcal{H}$ and by construction it would be the standard Gaussian measure on $\mathcal{H}$. Unfortunately, the bad thing is that the sequence $\{X_n\}$ does not converge in probability. We already know that this is not possible because we have seen that there can not exist a Gaussian measure $\mu$ on $\mathcal{H}$ (assuming $\mathcal{H}$ is infinite dimensional) whose characteristic functional is $\widehat{\mu}(x)=\ee^{-\frac{1}{2}\| x\|^2}$. Let us see directly how $\{X_n\}$ fails to converge in probability. For this it is sufficient to show that $\{X_n\}$ is not Cauchy in probability.
\begin{lem}
$\{X_n\}$ is not Cauchy in probability.
\end{lem}
\begin{proof}
Let $\varepsilon>0$ and $n>m$. Then 
\begin{align}
\widetilde{\mu}\left(\left\{\omega\in\Omega\Big|\left\|\sum_{i=m+1}^n\xi_i(\omega)e_i\right\|>\varepsilon\right\}\right)&=\mu_{n-m}\left(\R^{n-m}\setminus \overline{B(0,\varepsilon)}\right)\\
&=1-\mu_{n-m}(\overline{B(0,\varepsilon)})\\
&\geq 1-\mu_{n-m}([-\varepsilon,\varepsilon]^{n-m})\\
&=1-(\mu_1([-\varepsilon,\varepsilon]))^{n-m}.
\end{align}
Note that $\mu_1([-\varepsilon,\varepsilon])<1$ implies that $(\mu_1([-\varepsilon,\varepsilon]))^{n-m}\xrightarrow{n,m\to\infty}0$. Here $\mu_1([-\varepsilon,\varepsilon])=\frac{1}{\sqrt{2\pi}}\int_{-\varepsilon}^\varepsilon \ee^{-\frac{1}{2}x^2}\dd x$. This implies that $\{X_n\}$ is not Cauchy in probability.
\end{proof}
Hence, our strategy to construct a Banach space containing $\mathcal{H}$ would be the following. First consider a new norm $\|\enspace\|_W$ for which the sequence $\{X_n\}$ is Cauchy in probability. In this case $\{X_n\}$ converges in probability if we consider the Banach space obtained by completing $\mathcal{H}$ with respect to this new norm. This motivates the following definition. 
\begin{defn}[Measurable norm]
A norm $\|\enspace\|_W$ on $\mathcal{H}$ is said to be \textsf{measurable} if for all $\varepsilon>0$ there is some $p_0\in\F$ such that 
\[
\mu(\{h\in \mathcal{H}\mid \|p_h\|_W>\varepsilon\})<\varepsilon
\]
for all $p\in \F$ such that $p\perp p_0$.
\end{defn}
Geometrically it means that $\|\enspace\|_W$ is such that $\mu$ is concentrated in a tubular neighborhood of some $p_0\in\F$. A non-example would be the norm $\|\enspace\|_\mathcal{H}$ on $\mathcal{H}$, which is not measurable.
\begin{thm}[Gross]
\label{gross1}
Let $\|\enspace\|_W$ be a measurable norm on $\mathcal{H}$. Let $W$ be the Banach space obtained by the completion of $\mathcal{H}$ with respect to $\|\enspace\|_W$. Then the sequence $\{X_n\}$ converges in probability in $W$.
\end{thm}
\begin{thm}[Gross]
\label{gross2}
Given a seperable real Hilbert space $\mathcal{H}$, there is a seperable Banach space $W$ with a linear continuous dense embedding $\iota:\mathcal{H}\hookrightarrow W$ and a Gaussian measure $\mu_W$ on $W$ such that 
\begin{multline}
\mu_W(\{ w\in W\mid (f_1(w),...,f_n(w))\in B,B\subseteq\R^n\text{ Borel measurable}\})\\
=\mu(\{h\in\mathcal{H}\mid (\langle h,f_1\rangle,...,\langle h,f_n\rangle)\in B\})
\end{multline}
for all $f_1,...,f_n\in W^*\hookrightarrow \mathcal{H}^*\cong\mathcal{H}$. In particular, for all $f\in W^*\subseteq \mathcal{H}$
\[
\widehat{\mu}(f)=\ee^{-\frac{1}{2}\| f\|_\mathcal{H}^2}.
\]
\end{thm}
Here Gaussian measure means that for all $f\in W^*$ we have that $f_*\mu$ is Gaussian on $\R$.
\begin{rem}
If $\|\enspace\|_W$ is a measurable norm on $\mathcal{H}$, then there is some $c>0$ such that $\|h\|_W\leq c\|h\|$ for all $h\in \mathcal{H}$ (see \cite{HK}). It was expected that $\|\enspace\|_W$ is dominated by $\|\enspace\|$ because we needed a bigger topology on $\mathcal{H}$ to allow convergence of $\{X_n\}$.
\end{rem}
\begin{rem}
\label{norm1}
Let $A$ be a positive Hilbert-Schmidt operator on $\mathcal{H}$. Define a new norm by 
\[
\|h\|_{W_A}=\|Ah\|.
\]
Then $\|\enspace\|_{W_A}$ is a measurable norm on $\mathcal{H}$ (see \cite{HK}).
\end{rem}
\begin{rem}
\label{norm2}
In the view of remark \ref{norm1} we see that there can be many Banach spaces. In other words, we have no uniqueness. However, we do not care.
\end{rem}
To elaborate on remark \ref{norm2}, a slogan here is that $\mathcal{H}$ contains all information about the measure. Our next goal is to make this slogan a little more precise, and this requires some effort. Given a seperabe Hilbert space $\mathcal{H}$, we saw that there is a Banach space $W$, a linear continuous dense embedding $\iota:\mathcal{H}\hookrightarrow W$ and a Gaussian measure $\mu$ on $W$ such that $\widehat{\mu}(f)=\ee^{-\frac{1}{2}\|f\|_\mathcal{H}^2}$, where $f\in W^*\subseteq\mathcal{H}^*\cong\mathcal{H}$. Next, we would like to understand whether it is possible to identify $\mathcal{H}$ from a seperable Banach space $W$ and a centered Gaussian measure $\mu$. More precisely, given a seperable Banach space $W$ and a centered Gaussian measure $\mu$ on $W$, is it possible to find a Hilbert space $\mathcal{H}(\mu)$ together with a linear continuous dense embedding $\iota:\mathcal{H}(\mu)\hookrightarrow W$ such that $\widehat{\mu}(f)=\ee^{-\frac{1}{2}\|f\|^2_{\mathcal{H}(\mu)}}$. We will start with a seperable real Hilbert space $\mathcal{H}$ and a Banach space $W$ and a Gaussian measure $\mu$ given by theorem \ref{gross2}. Then we will try to understand how to recover $\mathcal{H}$ from $W$ and $\mu$. This will give a hint on the construction of $\mathcal{H}(\mu)$ out of $W$ and $\mu$. Let $\mathcal{H},W$ and $\mu$ be as in theorem \ref{gross1}. Given $f\in W^*$, we have $q_\mu(f)=\int_W f(w)^2\dd\mu(w)$. More generally,
\begin{align*}
q_\mu:W^*\times W^*&\longrightarrow \R\\
(f,g)&\longmapsto q_\mu(f,g)=\int_Wfg\dd\mu.
\end{align*}
\begin{defn}[Covariance of a measure]
The map $q_\mu$ is called the \textsf{covariance} of $\mu$.
\end{defn}
\begin{exe}
Show that $q_\mu(f,g)=\langle f,g\rangle_\mathcal{H}$, where $f,g\in W^*\subseteq\mathcal{H}^*\cong\mathcal{H}$.
\end{exe}
First, we would like to show that $q_\mu$ is a continuous positive definite symmetric bilinear form on $W^*$. To see this we need a technical tool: \textsf{Fernique}'s theorem, which we state without proof.
\begin{thm}[Fernique]
Let $W$ be a seperable Banach space and $\mu$ be a Gaussian measure on $W$. Then there is some $\varepsilon=\varepsilon(\mu)>0$ such that
\[
\int_W\ee^{\varepsilon\| w\|_W^2}\dd\mu(w)<\infty.
\]
\end{thm}
\begin{cor}
$$\int_W\|w\|_W^p\dd\mu(w)<\infty,\hspace{0.5cm}\forall p\geq 1.$$
\end{cor}
\begin{prop}
\label{prop_bil}
$q_\mu$ is a continuous bilinear form on $W^*$.
\end{prop}
\begin{proof}
We have 
\[
\vert q_\mu(f,g)\vert\leq \int_W\vert f(w)g(w)\vert\dd\mu(w)\leq \|f\|_{W^*}\|g\|_{W^*}\underbrace{\int_W\| w\|_W^2\dd\mu(w)}_{c}.
\]
\end{proof}
Note that $f\in W^*$ implies that $q_\mu(f)<\infty$ and thus $f\in L^2(W,\mu)$. Hence, we have a canonical linear map
\begin{align*}
T:W^*&\longrightarrow L^2(W,\mu)\\
f&\longmapsto f.
\end{align*}
\begin{lem}
The map $T$ is continuous.
\end{lem}
\begin{proof}
We have 
$$\|T(f)\|_{L^2(W,\mu)}=\int_W f(w)^2\dd\mu(w)\leq \| f\|_{W^*}^2\int_W\|w\|_W^2\dd\mu(w).$$
\end{proof}
\begin{cor}
The norm on $W^*$ induced by $q_\mu$ is weaker then $\|\enspace\|_{W^*}$.
\end{cor}
\begin{lem}
Let $J$ be the map $J:W^*\to \mathcal{H}$ given as the composition $W^*\hookrightarrow\mathcal{H}^*\xrightarrow{\sim}\mathcal{H}$. Then $J:(W^*,q_\mu)\to\mathcal{H}$ is a linear continuous dense isometric embedding.
\end{lem}
\begin{proof}
It is a direct consequence of the previous corollary.
\end{proof}
\begin{exe}
Given $h\in\mathcal{H}$, define $\alpha_h:W^*\to\R$ by $\alpha_h(f)=f(h)$. Show that $\alpha_h$ is continuous on $W^*$ with respect to the topology given by $q_\mu$.
\end{exe}
\begin{cor}
If $h\in \mathcal{H}$, then $\|h\|_\calH=\sup_{f\in W^*\setminus\{0\}}\frac{\vert f(h)\vert}{\sqrt{q_\mu(f,f)}}$.
\end{cor}
\begin{proof}
Since $J(W^*)$ is dense in $\mathcal{H}$, we know that 
\[
\|h\|_\calH=\sup_{f\in W^*\setminus\{0\}}\frac{\vert f(h)\vert}{\| J(f)\|}=\sup_{f\in W^*\setminus\{0\}}\frac{\vert f(h)\vert}{\| f\|_{L^2(W,\mu)}}=\sup_{f\in W^*\setminus\{0\}}\frac{\vert f(h)\vert}{\sqrt{q_\mu(f,f)}}.
\]
\end{proof}
Let $K$ be the completion of $T(W^*)$ in $L^2(W,\mu)$. Then we see that $J$ extends to an isometry $J:K\to \mathcal{H}$.
\begin{exe}
Show that $J:K\to\mathcal{H}$ is an isomorphism of Hilbert spaces. 
\end{exe}
Let us summerize what we have seen so far:
\begin{enumerate}
\item{We have seen that
$$h\in \mathcal{H}\Longrightarrow \|h\|_\calH=\sup_{f\in W^*\setminus\{0\}}\frac{\vert f(h)\vert}{\sqrt{q_\mu(f,f)}}.$$ 
This relation can be thought of as constructing the norm on $\mathcal{H}$ out of $W$ and $\mu$. It will be the key in order to construct $\mathcal{H}$ out of $W$ and $\mu$.
}
\item{The map $J:K\to\mathcal{H}$ is an isomorphism of Hilbert spaces. In particular, it is an isomorphism of Banach spaces.
}
\end{enumerate}
Given a seperable Banach space on $W$ and a contered Gaussian measure $\mu$ on $W$, (1) will be used to define a normed space $\mathcal{H}(\mu)$ and (2) will be used to give an inner product on $\mathcal{H}(\mu)$. This way we will be able to construct $\mathcal{H}(\mu)$ out of $W$ and $\mu$.
\begin{defn}[$\mathcal{H}(\mu)$-norm]
Let $W$ be a real seperable Banach space and $\mu$ a centered Gaussian measure on $W$. Define a \textsf{norm} $\|\enspace\|_{\mathcal{H}(\mu)}$ by 
\[
\|w\|_{\mathcal{H}(\mu)}=\sup_{f\in W^*\setminus\{0\}}\frac{\vert f(w)\vert}{\sqrt{q_\mu(f,f)}},
\]
and $\mathcal{H}(\mu)=\{w\in W\mid \|w\|_{\mathcal{H}(\mu)}<\infty\}$. The space $\mathcal{H}(\mu)$ is called the \textsf{Cameron-Martin} space.
\end{defn}
\begin{exe}
Show that $\mathcal{H}(\mu)$ is a normed space.
\end{exe}
\begin{exe}
Show that $w\in \mathcal{H}(\mu)$ iff $f\mapsto f(w)$ is continuous on $W^*$ if $W^*$ has the topology induced by $q_\mu$.
\end{exe}
\begin{prop}
$\mathcal{H}(\mu)$ is a Banach space, i.e. $\|\enspace\|_{\mathcal{H}(\mu)}$ is complete.
\end{prop}
\begin{proof}
We first show that there is some $c>0$ such that for all $w\in \mathcal{H}(\mu)$
\[
\|w\|_W\leq c\|w\|_{\mathcal{H}(\mu)}.
\]
In other words $\iota:(\mathcal{H}(\mu),\|\enspace\|_{\mathcal{H}(\mu)})\hookrightarrow W$ is continuous. Let $w\in W\setminus\{0\}$. By the \textsf{Hahn-Banach} theorem we can choose $f\in W^*$ such that $\|f\|_{W^*}=1$ and $f(w)=\|w\|_W$. Moreover, by proposition \ref{prop_bil} we have that $\|f\|_{q_\mu}\leq \tilde c\|f\|_{W^*}$ and thus $\|f\|_{q_\mu}\leq \tilde c$. Now 
\[
\|w\|_W=f(w)=\vert f(w)\vert=\frac{\vert f(w)\vert}{\|f\|_{W^*}}\leq c\frac{\vert f(w)\vert}{\|f\|_{q_\mu}}\leq c\|w\|_{\mathcal{H}(\mu)},
\]
with $c=\frac{1}{\tilde c}$. To show that $\mathcal{H}(\mu)$ is complete, let $\{h_n\}$ be a Cauchy sequence in $\mathcal{H}(\mu)$ with respect to $\|\enspace\|_{\mathcal{H}(\mu)}$. By the previuous section it is Cauchy in $(W,\|\enspace\|_W)$. Since $W$ is complete, there is some $h\in W$ such that $h_n\to h$ in $(W,\|\enspace\|_W)$. We claim that $h\in \mathcal{H}(\mu)$ and $h_n\to h$ in $\mathcal{H}(\mu)$. Let $\varepsilon>0$. Choose $m\geq n$ large enough such that 
\[
f(h_n-h)<\varepsilon\sqrt{q_\mu(f,f)},
\]
(use $h_n\to h$ in $W$ and $f\in W^*$). Therefore 
\[
\frac{\vert f(h_n-h)\vert}{\sqrt{q_\mu(f,f)}}\leq \frac{\vert f(h_n-h_m)\vert}{\sqrt{q_\mu(f,f)}}+\frac{\vert f(h_m-h)\vert}{\sqrt{q_\mu(f,f)}}\leq \|h_n-h_m\|_{\mathcal{H}(\mu)}+\varepsilon.
\]
This shows that $\|h_n-h\|_{\mathcal{H}(\mu)}<\infty$ implies that $h\in\mathcal{H}(\mu)$ and $\|h_n-h\|_{\mathcal{H}(\mu)}\xrightarrow{n\to\infty}0$. 
\end{proof}
Even though we haven been able to show that $\mathcal{H}(\mu)$ is a Banach space, so far, we haven't done anything to show that $\mathcal{H}(\mu)\not=\{0\}$. In order to see this, and that $\mathcal{H}(\mu)$ is a Hilbert space, we will use \textsf{Bochner integrals}. 

\subsubsection{Digression on Bochner integrals}
Let $(\Omega,\sigma(\Omega),\mu)$ be a measure space and $W$ a Banach space. 
\begin{defn}[Bochner integrable]
Let $f:\Omega\to W$ be a measurable map, where we consider the Borel $\sigma$-algebra on $W$. We say $f$ is \textsf{Bochner integrable} if 
\[
\int_\Omega\| f(\omega)\|_W\dd\mu(\omega)<\infty.
\]
\end{defn}
If $f$ is Bochner integrable, it is possible to define $\int_\Omega f(\omega)\dd\mu(\omega)\in W$, i.e. a $W$-valued integral on $\Omega$, which is called the \textsf{Bochner integral} of $f$. For us $(\Omega,\sigma(\Omega),\mu)$ will be $(W,\calB(W),\mu)$. Let $f\in W^*$, then 
\[
\int_W\|wf(w)\|_W\dd\mu(w)\leq \|f\|_{W^*}\int_W\|w\|_W^2\dd\mu(w)<\infty,
\]
by Fernique's theorem. This implies that, given $f\in W^*$, the map $w\mapsto wf(w)$ is Bochner integrable. Hence $\int_W wf(w)\dd\mu(w)\in W$. For $f\in W^*$ we define 
\[
J(f):=\int_Wwf(w)\dd\mu(w)\in W.
\]
\begin{exe}
Show that $J(f)\in \mathcal{H}(\mu)$ and $\|J(f)\|_{\mathcal{H}(\mu)}\leq \|f\|_{q_\mu}$.
\end{exe}
As a consequence of the exercise we see that $J:W^*\to \mathcal{H}(\mu)$ is a contraction if $W^*$ is endowed with the norm induced by $q_\mu$, i.e.
\[
\|f\|_{q_\mu}=\|f\|_{L^2(W,\mu)}.
\]
This shows that $\mathcal{H}(\mu)\not=\{0\}$. Since $W^*$ is dense in $K$, we have an isometry $J:K\to\mathcal{H}(\mu)$. Next, we show that $J$ is surjective. Given $h\in\mathcal{H}(\mu)$ such that the map $f\mapsto f(h)$ is continuous on $(W^*,q_\mu)$. We call this map $\hat h$. Note that $\hat h$ extends to a continuous linear functional on $K$. Hence, by Riesz's theorem, $\hat h$ can be identified with a element of $\hat h\in K$. It is easy to check that $J(\hat h)=h$. Thus $J$ is an isomorphism of Banach spaces. Now we give $\mathcal{H}(\mu)$ the Hilbert space structure induced by $J$. Since $K$ is a seperable Hilbert space, so is $\mathcal{H}(\mu)$. We have the following theorem
\begin{thm}[Gross]
\label{gross3}
Given a real seperable Banach space and a Gaussian measure $\mu$, there exist a Hilbet space $\mathcal{H}(\mu)\subseteq W$ such that $\mathcal{H}(\mu)\hookrightarrow W$ is a linear continuous dense embedding and 
\[
\widehat{\mu}(f)=\ee^{-\frac{1}{2}\|f\|^2_{\mathcal{H}(\mu)}},\hspace{0.5cm}\forall f\in W^*.
\]
\end{thm}
\begin{exe}
\label{Ex10}
Let $\mathcal{H}(\mu)$ be the Cameron-Martin space of $(W,\mu)$. Let $\{e_n\}\subseteq W^*$ be an orthonormal basis of $\mathcal{H}(\mu)$. Show that if $w\in W$, then $w\in \mathcal{H}(\mu)$ iff 
\[
\sum_{n=1}^\infty e_n(w)<\infty,
\]
and in this case 
\[
\|w\|^2_{\mathcal{H}(\mu)}=\sum_{n=1}^\infty e_n(w)^2.
\]
\end{exe}
\begin{cor}
$\mu(\mathcal{H}(\mu))=0$ if $\mathcal{H}(\mu)$ is infinite dimensional.
\end{cor}
\begin{proof}
Let $\{e_n\}\subseteq W^*$ be as in the previuous exercise. Then $\{e_n\}$ is a sequence of independent, identically distributed random variables with mean $0$ and variance $1$. Hence, by the law of large numbers, we get
\[
\sum_{n=1}^\infty e_n(w)^2=\mu\hspace{0.5cm}a.e.
\]
Whereas by exercise \ref{Ex10} we have 
\[
\mathcal{H}(\mu)=\left\{w\in W\Big| \sum_{n=1}^\infty e_n(w)^2<\infty\right\}.
\]
This implies that $\mu(\mathcal{H}(\mu))=0$.
\end{proof}
The is a different way to understand the Cameron-Martin space.
\begin{thm}[Cameron-Martin]
Let $\mathcal{H}(\mu)$ be the Cameron-Martin space of $(W,\mu)$, $h\in W$ and $T_h:W\to W$ given by $T_h(w)=w-h$. If $h\in\mathcal{H}(\mu)$, then $\mu_h=(T_h)_*\mu$ is absolutely continuous with respect to $\mu$ and 
\begin{equation}
\label{cam_mart}
\frac{\dd\mu_h}{\dd\mu}=\ee^{-\frac{1}{2}\|h\|^2_{\mathcal{H}(\mu)}}\ee^{\langle h,\enspace\rangle}.
\end{equation}
\end{thm}
\begin{proof}
We will compute the Fourier transform of $\mu_h$ and $\ee^{-\frac{1}{2}\|h\|^2_{\mathcal{H}(\mu)}}\ee^{\langle h,\enspace\rangle}\mu$. Let thus $f\in W^*$. Then we have 
\[
\widehat{\mu}_h(f)=\int_W\ee^{\I f(w)}\dd\mu_h(w)=\int_W\ee^{\I f(w-h)}\dd\mu(w)=\ee^{-\I f(h)}\ee^{-\frac{1}{2}q_\mu(f,f)}.
\]
Ont he other hand, setting $\tilde h=J^{-1}(h)$, we get
\begin{align}
\begin{split}
\int_W\ee^{-\frac{1}{2}\| h\|^2}\ee^{\tilde h(w)}\ee^{\I f(w)}\dd\mu(w)&=\ee^{-\frac{1}{2}\|h\|^2}\int_W\ee^{\I(f-\I\tilde h)(w)}\dd\mu(w)\\
&=\ee^{-\frac{1}{2}\|h\|^2}\ee^{-\frac{1}{2}q_\mu(f-\I h,f-\I h)}\\
&=\ee^{\I f(h)}\ee^{-\frac{1}{2}q_\mu(f,f)}.
\end{split}
\end{align}
Thus $(T_h)_*\mu$ is absolutely continuous with respect to $\mu$ and \eqref{cam_mart} holds.
\end{proof}
\begin{rem}
$h\in W\setminus \mathcal{H}(\mu)$ implies that $\mu$ and $\mu_h$ are mutually singular. Hence $h\in\mathcal{H}(\mu)$ iff $(T_h)_*\mu\ll\mu$.
\end{rem}
\begin{ex}[Cameron-Martin space of the classical Wiener space]
Recall the classical Wiener space: we had $W=C_0([0,1])$ endowed with the Wiener measure $\mu$ together with the covariance 
\[
q_\mu(\ev_t,\ev_s)=\min\{s,t\},
\]
where $\ev$ denotes the evaluation map. One can check that $E=span\{\ev_t\mid t\in[0,1]\}$ is dense in $W^*$. First, let us compute $J:E\to W$. By definition we have 
\[
(J(\ev_t))(s)=\int_{C_0([0,1])}x(s)\ev_t(x)\dd\mu(x)=\int_{C_0([0,1])}x(s)x(t)\dd\mu(x)=\min\{s,t\},
\]
and thus $\frac{\dd}{\dd s}(J(\ev_t))=\chi_{[0,t]}$. Note htat $q_\mu(\ev_t,\ev_s)=\min\{s,t\}$. On the other hand
\[
\int_0^1\dot{J}(\ev_t)(u)\dot{J}(\ev_s)(u)\dd u=\int_0^1\chi_{[0,1]}(u)\chi_{[0,s]}(u)\dd u=\min\{s,t\}. 
\]
This shows that if $x,y\in \mathcal{H}(\mu)$, then 
\begin{equation}
\label{innerprod}
\langle x,y\rangle_{\mathcal{H}(\mu)}=\int_0^1\dot{x}(u)\dot{y}(u)\dd u.
\end{equation}
Thus we can write down the space $\mathcal{H}(\mu)=\{x\in C_0([0,1])\mid \dot{x}\in L^2([0,1])\}$ and the inner product on $\mathcal{H}(\mu)$ is given by \eqref{innerprod}. This means that the Winer measure is the standard Gaussian measure on $H^1([0,1])$, which is the space of $1$-Sobolev paths. This observation is due to Cameron-Martin, which was later generalized by Gross.
\end{ex}

\section{Wick ordering}
\subsection{Motivating example and construction}
Let $H_n(x)$ be the degree $n$ \textsf{Hermite polynomial} on $\R$. It can be defined recursively as follows 
\begin{align*}
H_0(x)&=1\\
\frac{\dd}{\dd x}H_n(x)&=n\frac{\dd}{\dd x}H_{n-1}(x)\\
\int_\R H_n(x)\dd\mu(x)&=0,\hspace{0.5cm}\text{where $\mu$ is the standard Gaussian measure on $\R$.}
\end{align*}
Moreover, Hermite polynomials are given by the generating function 
\begin{equation}
\label{generating}
\ee^{tx-\frac{1}{2}t^2}=\sum_{n=0}^\infty t^n H_n(x),
\end{equation}
i.e. $\ee^{tx-\frac{1}{2}t^2}$ is the generating function for $H_n(x)$. This is used as the best way to study properties of $H_n(x)$. 
\begin{exe}
Show that $H_n(x)=\ee^{-\frac{\Delta}{2}}(x^n)$, where $\Delta=-\frac{\dd^2}{\dd x^2}$ (this is a conceptual way to think about Wick ordering).
\end{exe}
The following statements hold:
\begin{enumerate}
\item{$H_n(x)$ form an orthonormal basis of $L^2(\R,\mu)$.
}
\item{$$\int_\R H_n(x)^2\dd\mu(x)=n!.$$
}
\end{enumerate}
Note that here $L^2(\R,\mu)=\bigoplus_{n=0}^\infty span(H_n(x))$, i.e. $H_n(x)\perp \bigoplus_{k=0}^{n-1}span(H_k(x))$. More generally, 
$$H_n(x^{t_1},...,x^{t_r})=H_{t_1}(x)\dotsm H_{t_r}(x),\hspace{0.5cm}t_1+\dotsm +t_r=n,$$
and one can check that Hermite polynomials form an orthogonal basis of $L^2(\R^n,\mu)$. On the other hand $\R$ is a Hilbert space, in fact it is also a Cameron-Martin space (CMS) for standard Gaussian measure $\mu$. We can then talk about the \textsf{Bosonic Fock space} of $\R$.
We have

\[
\begin{tikzcd}
\widetilde{Sym}^\bullet(\R)\ar[d,equal]\ar[r,equal]&\ar[d,equal]L^2(\R,\mu)\\
\bigoplus_{n\geq 0}Sym^n(\R)\arrow[r,""]&\bigoplus_{n\geq 0}span(H_n(x))
\end{tikzcd}
\]


where the arrow on the bottom represents a canonical isomorphism. More generally, we want to prove
\[
L^2(W,\mu)\cong \widetilde{Sym}^\bullet(H(\mu)),
\]
where the isomorphism is canonical. Let thus $W$ be a seperable Banach space and $\mu$ a centered Gaussian measure on $W$ and $H(\mu)$ be its Cameron-Martin space with covariance $q_\mu$. Let $f\in W^*$ (or $f\in H(\mu)$, since it doesn't matter). Define\footnote{In the physics literature, this is usually written without the brackets, i.e. $:f^n:$. We write the brackets to avoid confusion with the ``double dot'' of a function or when two such objects are multiplied.} then $(:f^n:)$ recursively as follows:
\begin{align*}
(:f^0:)&=1\\
\frac{\partial}{\partial f}(:f^n:)&=n(:f^{n-1}:)\\
\int_W(:f^n:)\dd\mu&=0.
\end{align*}
This definition sais that $(:f^n:)$ comes from $f^n$ in the same way as $H_n(x)$ comes from $x^n$.
Let us use generating functions to observe properties of $(:f^n:)$. Define 
\[
(:\ee^{\alpha f}:)=\sum_{k=0}^\infty\frac{\alpha^k}{k!}(:f^k:)
\]
One can then check that 
\begin{equation}
\label{identity1}
\int_W(:\ee^{\alpha f}:)\dd\mu=1.
\end{equation}
Moreover, by definition
\[
\frac{\partial}{\partial f}(:\ee^{\alpha f}:)=\alpha (:\ee^{\alpha f}:),
\]
which implies that $(:\ee^{\alpha f}:)=c\ee^{\alpha f}$, where $c$ is a constant, which needs to be determined. Using \eqref{identity1}, we see that $c=\frac{1}{\int_W \ee^{\alpha f}\dd\mu}$.
\begin{exe}
Use Wick's theorem to show that 
\begin{align*}
\int_W f^{2n+1}\dd\mu&=0,\\
\int_W f^{2n}\dd\mu&=\frac{(2n)!}{2^n n!}q_\mu(f,f)^n.
\end{align*}
\end{exe}
Hence, regarding $\ee^{\alpha f}$ as a formal power series in $\alpha$, we see that 
\[
\int_W\ee^{\alpha f}\dd\mu=\sum_{k=0}^\infty\frac{\alpha^k}{k!}\int_W f^k\dd\mu=\sum_{k=0}^\infty\frac{\alpha^{2k}}{(2k)!}\frac{(2k)!}{k!}(q_\mu(f,f))^k=\ee^{\frac{1}{2}\alpha^2 q_\mu(f,f)},
\]
because $(:\ee^{\alpha f}:)=\ee^{\alpha f}\ee^{-\frac{1}{2}\alpha^2 q_\mu(f,f)}$, which is exactly \eqref{generating} when $W=\R$ and $f(x)=x$. One can use the generating function for $(:f^n:)$ to show that 
\begin{align*}
(:f^n:)&=\sum_{k=0}^{[\frac{n}{2}]}\frac{n!}{k!(n-2k)!}f^{n-2k}\left(-\frac{1}{2}q_\mu(f,f)\right)^k\\
f^n&=\sum_{k=0}^\frac{n}{2}\frac{n!}{k!(n-2k)!}(:f^{n-2k}:)\left(\frac{1}{2}q_\mu(f,f)\right)^k
\end{align*}

\begin{prop}
$$\int_W(:f^n:)(:g^m:)\dd\mu=\delta_{nm}n!\langle f,g\rangle^n.$$
\end{prop}
\begin{proof} We have
\begin{align*} 
(:\ee^{\alpha f}:)&=\ee^{\alpha f-\frac{1}{2}\alpha^2 q_\mu(f,f)}\\
(:\ee^{\beta f}:)&=\ee^{\beta g-\frac{1}{2}\beta^2 q_\mu(g,g)}.
\end{align*}
Then 
\[
(:\ee^{\alpha f}:)(:\ee^{\beta g}:)=(:\ee^{\alpha f+\beta g}:)\ee^{\alpha \beta q_\mu(f,g)}.
\]
Thus 
\[
\int_W (:\ee^{\alpha f}:)(:\ee^{\beta g}:)\dd\mu=\ee^{\alpha \beta q_\mu(f,g)}\underbrace{\int_W(:\ee^{\alpha f+\beta g})\dd\mu}_{=1}=\sum_{k=0}^\infty \frac{(\alpha\beta)^k}{k!}q_\mu(f,g)^k.
\]
Comparing the coeefficients of $(\alpha\beta)^k$ we get 
\[
\frac{1}{(k!)^2}\int_W (:f^k:)^k(:g^k:)\dd\mu=\frac{1}{(k!)!}q_\mu(f,g)^k.
\]
\end{proof}
The questions is how to define $(:f^{k_1}g^{k_2}:)$? The idea is to use the recursive definition. Therefore we get for $n=n_1+\dotsm +n_k$
\begin{align*} 
(:f_1^0\dotsm f_k^0:)&=1,\\
\int_W(:f_1^{n_1}\dotsm f_k^{n_k}:)\dd\mu&=0,\\
\frac{\partial}{\partial f_i}(:f_1^{n_1}\dotsm f_k^{n_k}:)&=n_i(:f_1^{n_1}\dotsm f_i^{n_{i-1}}\dotsm f_k^{n_k}:).
\end{align*}
One can easily check that $(:(\alpha f+\beta g)^n:)=\sum_{k=1}^n\binom{n}{k}\alpha^k\beta^{n-k}(:f^kg^{n-k}:)$ by using generating functions.
\begin{exe}
Show that 
\[
\int_W(:f_1-f_n:)(:g_1-g_m:)\dd\mu=\begin{cases}0,&\text{if $m\not=n$}\\ \sum_{\sigma\in S_n}\prod_{i=1}^n\langle f_i,g_{\sigma(i)}\rangle,&\text{if $m=n$}\end{cases}
\]
\end{exe}
As in the finite dimensional case, given a polynomial function $P$ on $W$, i.e. $$P(w)=\sum_{i_1,...,i_k} a_{i_1}\dotsm a_{i_k}f_{i_1}^{\alpha_1}(w)\dotsm f_{i_k}^{\alpha_k}(w),$$ where $f_{i_1},...,f_{i_k}\in W^*$, formally 
\[
(:P:)=\ee^{-\frac{\Delta}{2}}P,
\]
where we need to check what $\Delta$ is. One can use the Cameron-Martin space to make sense of $\Delta$. Take an orthonormal basis $\{e_n\}$ of $\mathcal{H}(\mu)$. Then loosely speaking $\Delta=-\sum_{n=1}^\infty\frac{\partial^2}{\partial e_n^2}$. 

\subsection{Wick ordering as a value of Feynman diagrams}
Given $f_1,f_2,f_3,f_4\in W^*$, and $\gamma\in\{\{1,2\},\{3,4\}\}$, we can construct a diagram as follows 

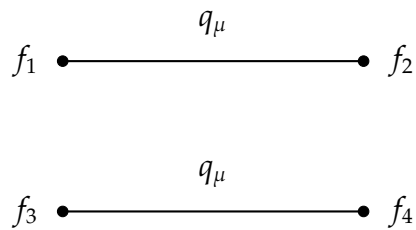
\begin{figure}[!ht]
\centering
\tikzset{
particle/.style={thick,draw=black},
gluon/.style={decorate, draw=black,
    decoration={coil,aspect=0}}
 }
\begin{tikzpicture}[node distance=1.5cm and 2cm]
\node[](f1) at (0,0){};
\node[](f1L) at (-0.2,0){};
\node[](f11) at (-0.5,0){$f_1$};
\node[](f2) at (4,0){};
\node[](f2L) at (4.2,0){};
\node[](f22) at (4.5,0){$f_2$};
\node[](f3) at (0,-2){};
\node[](f3L) at (-0.2,-2){};
\node[](f33) at (-0.5,-2){$f_3$};
\node[](f4) at (4,-2){};
\node[](f4L) at (4.2,-2){};
\node[](f44) at (4.5,-2){$f_4$};
\node[](q1) at (2,0.5){$q_\mu$};
\node[](q2) at (2,-1.5){$q_\mu$};
\draw[fill=black](f1) circle (.07cm);
\draw[fill=black](f2) circle (.07cm);
\draw[fill=black](f3) circle (.07cm);
\draw[fill=black](f4) circle (.07cm);
\draw[particle] (f1L) -- (f2L);
\draw[particle] (f3L) -- (f4L);
\end{tikzpicture}
\end{figure}

The value of such a diagram will be $q_\mu(f_1,f_2)q_\mu(f_3,f_4)$.
\begin{defn}[Feynman diagram]
A \textsf{Feynman diagram} with $n$ vertices and rank $r$, where $r\leq\frac{n}{2}$, consists of a set $V$ called the \textsf{set of vertices} (thus $\vert V\vert=n$), and a set $H$ called the \textsf{set of half edges}, which consists of $r$ disjoint pair of vertices. The remaining vertices are called \textsf{unpaired vertices}, which will be denoted by $A$.
\end{defn}
\begin{ex}

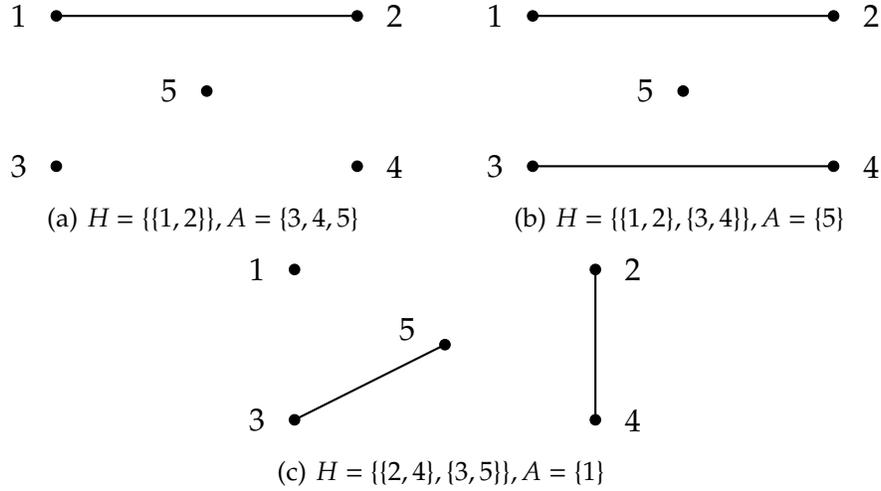
\begin{figure}[!ht]
\centering
\subfigure[$H=\{\{1,2\}\}$, $A=\{3,4,5\}$]{
\tikzset{
particle/.style={thick,draw=black},
gluon/.style={decorate, draw=black,
    decoration={coil,aspect=0}}
 }
\begin{tikzpicture}[node distance=1.5cm and 2cm]
\node[](f1) at (0,0){};
\node[](f1L) at (-0.2,0){};
\node[](f11) at (-0.5,0){$1$};
\node[](f2) at (4,0){};
\node[](f2L) at (4.2,0){};
\node[](f22) at (4.5,0){$2$};
\node[](f3) at (0,-2){};
\node[](f3L) at (-0.2,-2){};
\node[](f33) at (-0.5,-2){$3$};
\node[](f4) at (4,-2){};
\node[](f4L) at (4.2,-2){};
\node[](f44) at (4.5,-2){$4$};
\node[](f5) at (2,-1){};
\node[](f5L) at (1.5,-1){};
\node[](f55) at (1.5,-1){$5$};
\draw[fill=black](f1) circle (.07cm);
\draw[fill=black](f2) circle (.07cm);
\draw[fill=black](f3) circle (.07cm);
\draw[fill=black](f4) circle (.07cm);
\draw[fill=black](f5) circle (.07cm);
\draw[particle] (f1L) -- (f2L);
\end{tikzpicture}
}
\quad
\subfigure[$H=\{\{1,2\},\{3,4\}\}$, $A=\{5\}$]{
\tikzset{
particle/.style={thick,draw=black},
gluon/.style={decorate, draw=black,
    decoration={coil,aspect=0}}
 }
\begin{tikzpicture}[node distance=1.5cm and 2cm]
\node[](f1) at (0,0){};
\node[](f1L) at (-0.2,0){};
\node[](f11) at (-0.5,0){$1$};
\node[](f2) at (4,0){};
\node[](f2L) at (4.2,0){};
\node[](f22) at (4.5,0){$2$};
\node[](f3) at (0,-2){};
\node[](f3L) at (-0.2,-2){};
\node[](f33) at (-0.5,-2){$3$};
\node[](f4) at (4,-2){};
\node[](f4L) at (4.2,-2){};
\node[](f44) at (4.5,-2){$4$};
\node[](f5) at (2,-1){};
\node[](f5L) at (1.5,-1){};
\node[](f55) at (1.5,-1){$5$};
\draw[fill=black](f1) circle (.07cm);
\draw[fill=black](f2) circle (.07cm);
\draw[fill=black](f3) circle (.07cm);
\draw[fill=black](f4) circle (.07cm);
\draw[fill=black](f5) circle (.07cm);
\draw[particle] (f1L) -- (f2L);
\draw[particle] (f3L) -- (f4L);
\end{tikzpicture}
}
\quad
\subfigure[$H=\{\{2,4\},\{3,5\}\}$, $A=\{1\}$]{
\tikzset{
particle/.style={thick,draw=black},
gluon/.style={decorate, draw=black,
    decoration={coil,aspect=0}}
 }
\begin{tikzpicture}[node distance=1.5cm and 2cm]
\node[](f1) at (0,0){};
\node[](f1L) at (-0.2,0){};
\node[](f11) at (-0.5,0){$1$};
\node[](f2) at (4,0){};
\node[](f2L) at (4,0.1){};
\node[](f22) at (4.5,0){$2$};
\node[](f3) at (0,-2){};
\node[](f3L) at (-0.2,-2.1){};
\node[](f33) at (-0.5,-2){$3$};
\node[](f4) at (4,-2){};
\node[](f4L) at (4,-2.1){};
\node[](f44) at (4.5,-2){$4$};
\node[](f5) at (2,-1){};
\node[](f5L) at (2.1,-0.95){};
\node[](f55) at (1.5,-0.8){$5$};
\draw[fill=black](f1) circle (.07cm);
\draw[fill=black](f2) circle (.07cm);
\draw[fill=black](f3) circle (.07cm);
\draw[fill=black](f4) circle (.07cm);
\draw[fill=black](f5) circle (.07cm);
\draw[particle] (f4L) -- (f2L);
\draw[particle] (f3L) -- (f5L);
\end{tikzpicture}
}
\caption{Examples of Feynman graphs with $V=\{1,2,3,4,5\}$ and different $H$ and $A$.}
\label{Feynman_graph}
\end{figure}
\end{ex}

If we are given $f_1,...,f_n\in W^*$, we can think of them as vertices of a Feynman diagram. Hence, given $f_1,...,f_n\in W^*$ and a Feynman diagram of rank $r$, i.e. 
\[
\gamma(H)=\{\{i_1,j_1\},...,\{i_n,j_n\}\},
\]
we define the value $F(f_1,...,f_n;\gamma)$ of the Feynman diagram as 
\[
F(f_1,...,f_n;\gamma)=\left(\prod_{k=1}^rq_\mu(f_{i_k},f_{j_k})\right)\prod_{i\in A}f_i.
\]
Moreover, we say that $\gamma$ is \textsf{complete} if $n=2r$. With this notation, we can rephrase Wick's theorem as 
\[
\int f_1\dotsm f_n\dd\mu=\sum_{\text{$\gamma$ complete}\atop\text{Feynman diagram}}F(\gamma).
\]

\subsection{Abstract point of view on Wick ordering}
\begin{defn}[Gaussian Hilbert space]
A \textsf{Gaussian Hilbert space} $\mathcal{H}$ is a Hilbert space of random variables on some probability space $(\Omega,\sigma(\Omega),\mu)$ such that each $f\in\mathcal{H}$ is a Gaussian on $\R$, i.e. $f\in \mathcal{H}$, then $f_*\mu$ is Gaussian on $\R$. 
\end{defn}
\begin{ex}
If $K$ denotes the completion of $W^*$ in $L^2(W,\mu)$, then $K$ is a Gaussian Hilbert space. We assume that $\sigma(\Omega)$ is generated by elements of $\mathcal{H}$.
\end{ex}
Given a Gaussian Hilbert space $\mathcal{H}\subseteq L^2(\Omega,\sigma(\Omega),\mu)$, define the set 
\[
P_n(\mathcal{H})=\{P(\xi_1,...,\xi_n)\mid\xi_1,...,\xi_n\in \mathcal{H},\text{ $P$ is a polynomial of degree $\leq n$}\},
\]
and define 
\[
\mathcal{H}^{:n:}=\overline{P_n}(\mathcal{H})\cap (\overline{P_{n-1}}(\mathcal{H}))^\perp.
\]
We can then observe\footnote{We will define the \textsf{completed} direct sum at some later point.} that 
\begin{enumerate}
\item{
\[
\overline{P_n}(\mathcal{H})=\bigoplus_{k=0}^n\mathcal{H}^{:k:},
\]
}
\item{
\[
\bigoplus_{n=0}^\infty\mathcal{H}^{:n:}=\overline{\bigcup_{n=0}^\infty P_n(\mathcal{H})}.
\]
}
\end{enumerate}
\begin{thm}
\[
L^2(\Omega,\sigma(\Omega),\mu)=\bigoplus_{n=0}^\infty\mathcal{H}^{:n:}.
\]
\end{thm}
\begin{proof}
See 
\end{proof}
\begin{rem}
This is just a way to say that ``polynomial'' random variables are dense in $L^2(\Omega,\sigma(\Omega),\mu)$.
\end{rem}
\begin{thm}
Let $\xi_1,...,\xi_n\in\mathcal{H}$. Then 
\[
(:\xi_1\dotsm \xi_n:)=\pi_n(\xi_1\dotsm \xi_n),
\]
where $\pi_n:L^2(\Omega,\sigma(\Omega),\mu)\to \mathcal{H}^{:n:}$ is the projection.
\end{thm}
Hence, we saw that $(:\xi_1,...,\xi_n:)$ is nothing but an orthogonal projection of $\xi_1\cdot\xi_n$ onto $\mathcal{H}^{:n:}$. The idea here is that $\xi_1\dotsm\xi_n$ is not orthogonal to lower degree polynomials where as $(:\xi_1\dotsm\xi_n:)$ is orthogonal to lower degree polynomials. Hence, Wick ordering can be thought of as taking a polynomial and changing it into a new polynomial in such a way that the result is orthogonal to lower degree polynomials.

\section{Bosonic Fock Spaces}
Let $\mathcal{H}_1,\mathcal{H}_2$ be a seperable Hilbert spaces. Then we can look at the tensor product $\mathcal{H}_1\otimes\mathcal{H}_2$, where we have the inner product defined as 
\[
\langle h_1\otimes h_2,h_1'\otimes h_2'\rangle_{\mathcal{H}_1\otimes\mathcal{H}_2}:=\langle h_1,h_1'\rangle_{\mathcal{H}_1}\langle h_2,h_2'\rangle_{\mathcal{H}_2}.
\]
Moreover, we denote by $\mathcal{H}_1\widehat{\otimes}\mathcal{H}_2$ the completion of $\mathcal{H}_1\otimes\mathcal{H}_2$ with respect to $\langle\enspace,\enspace\rangle$. We call $\widehat{\otimes}$ the \textsf{Hilbert-Schmidt tensor product}. The space $\mathcal{H}^*_1\widehat{\otimes}\mathcal{H}_2$ is isomorphic to the space of Hilbert-Schmidt operators from $\mathcal{H}_1$ to $\mathcal{H}_2$. Let $\{\mathcal{H}\}_{i=0}^\infty$ be a familiy of Hilbert spaces. Then $\widehat{\bigotimes}_i \mathcal{H}_i$ is the completion of $\bigotimes_i\mathcal{H}_i$ with respect to the norm $\sum_i\|x_i\|_{\mathcal{H}_i}^2$, i.e. 
\[
\widehat{\bigotimes}_i\mathcal{H}_i=\left\{(x_i)\Big| \sum_{i=0}^\infty \|x_i\|_{\mathcal{H}_i}^2<\infty\right\}.
\]
We will drop $\widehat{\phantom{a}}$ from now on. We can also define the space $\mathcal{H}_1\otimes\dotsm\otimes \mathcal{H}_n$ and so on. Let now $\mathcal{H}$ be a real and sepereable Hilbert space. Then $T^n\mathcal{H}=\mathcal{H}^{\otimes n}$. Define the map $P:\mathcal{H}^{\otimes n}\to \mathcal{H}^{\otimes n}$ by $P(h_1\otimes\dotsm\otimes h_n)=\frac{1}{n!}\sum_{\sigma\in S_n}h_{\sigma(1)}\otimes\dotsm\otimes h_{\sigma(n)}$. Then $P$ is a projection. We define $Sym^n(\mathcal{H})=P(\mathcal{H}^{\otimes n})$. Now we can see that $S_n$ acts on $T^n\mathcal{H}$ and thus $Sym^n(\mathcal{H})$ is the invariant subspace of $T^n\mathcal{H}$ under this action. Given $h_1,...,h_n$, define 
\[
h_1\otimes_s\dotsm\otimes_s h_n=\frac{1}{\sqrt{n!}}\sum_{\sigma\in S_n}h_{\sigma(1)}\otimes\dotsm \otimes h_{\sigma(n)},
\]
i.e. $h_1\otimes_s\dotsm\otimes_s h_n=\sqrt{n!}P(h_1\otimes\dotsm\otimes h_n)$.
\begin{exe}
Show that
\[
\langle h_1\otimes_s\dotsm\otimes_sh_n,h_1'\otimes_s\dotsm\otimes_sh_n'\rangle_{\mathcal{H}^{\otimes n}}=\sum_{\sigma\in S_n}\prod_{i=1}^n\langle h_i,h_{\sigma(i)}\rangle_\mathcal{H}.
\]
In particular $\|h^{\otimes_s n}\|_{\calH^{\otimes_{s}n}}^2=n!\| h\|_{\calH}^{2n}$ (i.e. $\|h^{\otimes_sn}\|_{\calH^{\otimes_sn}}=\sqrt{n!}\|h\|_\calH$).
\end{exe}
Let us try to give an alternative definition for $Sym^n(\mathcal{H})$. For this, recall that closed subspaces are generated by elements of the form $h_1\otimes_s\dotsm \otimes_sh_n$. 
\begin{rem}
From now on we will not indicate the inner products.
\end{rem}
\begin{defn}[Bosonic Fock space]
The space $\widetilde{Sym}^\bullet(\mathcal{H})=\bigoplus_{n=0}^\infty Sym^n(\mathcal{H})$ is called the \textsf{Bosonic Fock space} of $\mathcal{H}$.
\end{defn}
\begin{rem}
We sometimes also write $\Gamma(\mathcal{H})$ or $\Exp(\mathcal{H})$ for the Bosonic Fock space.
\end{rem}
\begin{rem}
Similarly, one can define the \textsf{Fermionic Fock space} of $\mathcal{H}$ by 
\[
u_1\land\dotsm \land u_n=\frac{1}{\sqrt{n!}}\sum_{\sigma\in S_n}\sgn(\sigma)u_{\sigma(1)}\otimes\dotsm \otimes u_{\sigma(n)}.
\]
\end{rem}
One can check that $\Gamma(\mathcal{H}_1\otimes\mathcal{H}_2)=\Gamma(\mathcal{H}_1)\otimes\Gamma(\mathcal{H}_2)$. Now one can ask whether there is a functor $\mathcal{H}\mapsto\Gamma(\mathcal{H})$. Thus, given a bounded operator $A:\mathcal{H}_1\to\mathcal{H}_2$, we need to know whether $\Gamma(A)$ is bounded. As a matter of fact, this is not the case. On the other hand, If $\|A\|\leq 1$ then $\|\Gamma(A)\|\leq 1$. This leads to the functor 
\[
\Gamma:{\bf Hilb_B^{\leq 1}}\longrightarrow{\bf Hilb},
\]
where ${\bf Hilb_B^{\leq 1}}$ is the category with Hilbert spaces as objects and bounded linear operators of norm $\leq 1$ as morphisms and ${\bf Hilb}$ the category of Hilbert spaces.
\begin{rem}
In the Fermionic case, no restriction on $A$ is required. 
\end{rem}
Given $h\in\mathcal{H}$, we can define 
\[
\Exp(h)=\sum_{n=0}^\infty\frac{h^{\otimes n}}{n!}\in\Gamma(\mathcal{H}).
\]
Then, for $h_1,h_2\in\calH$, we observe
\[
\langle \Exp(h_1),\Exp(h_2)\rangle=\sum_{n=0}^\infty\frac{1}{(n!)^2}\langle h_1^{\otimes_sn},h_2^{\otimes_sn}\rangle=\sum_{n=0}^\infty\frac{1}{n!}\langle h_1,h_2\rangle^n=\Exp(\langle h_1,h_2\rangle).
\]
\begin{exe}
Show that $\Exp:\mathcal{H}\to \Gamma(\mathcal{H})$ is continuous (\emph{be aware that it is not linear}). Moreover, show that $\Exp$ is injective.
\end{exe}
\begin{lem}
The elements $\{\Exp(h)\mid h\in\mathcal{H}\}\subseteq\Gamma(\mathcal{H})$ are linearly independent in $\Gamma(\mathcal{H})$.
\end{lem}
\begin{proof}
Let $h_1,...,h_n\in\mathcal{H}$. We want to show
\[
\sum_{i=1}^n\lambda_i\Exp(h_i)=0\Rightarrow \lambda_i=0,\hspace{0.3cm}\forall i\geq 1,2,...,n.
\]
For this, choose $h\in\mathcal{H}$ such that 
\[
\langle h,h_i\rangle\not=\langle h,h_j\rangle,\hspace{0.3cm}\forall i\not=j.
\]
Then we get that $\sum_{i=1}^n\lambda_i\Exp(h_i)=0$ implies $\sum_{i=1}^n\lambda_i\ee^{\langle h_i,h\rangle}=0$ for all $h\in\mathcal{H}$. Thus $\sum_{i=1}^n\lambda_i\ee^{z\langle h_i,h\rangle}=0$ for all $z\in\mathbb{C}$ if we choose $h$ as above. Hence, $\lambda_i=0$ for all $i=1,2,...,n$. 
\end{proof}
\begin{exe}
Show that $\{\Exp(h)\mid h\in\mathcal{H}\}$ span $\Gamma(\mathcal{H})$.
\end{exe}
Recall that 
\begin{enumerate}
\item{$L^2(W,\mu)=\bigoplus_{n=0}^\infty K^{:n:}$.
}
\item{there is a canonical isomorphism of Hilbert spaces $\mathcal{H}(\mu)\xrightarrow{T}K$.
}
\end{enumerate}
Thus we can observe that there is a canonical isomorphism of Hilbert spaces 
\begin{align*}
Sym^n(K)&\longrightarrow K^{:n:}\\
\xi_1\otimes_s\dotsm\otimes_s\xi_n&\longmapsto (:\xi_1\dotsm \xi_n:), 
\end{align*}
which leads to a map
\begin{align*}
\Gamma(K)&\longrightarrow \bigoplus_{n=0}^\infty K^{:n:}\\
\Exp(\xi)&\longmapsto \sum_{n=0}^\infty\frac{(:\xi^n:)}{n!}=(:\Exp(\xi):)=\ee^{\xi}\ee^{-\frac{1}{2}\|\xi\|^2},
\end{align*}
that comes from the \textsf{Segal-Ito isomorphism}
\begin{align*}
\Gamma(\mathcal{H})&\longrightarrow\Gamma(K)\cong L^2(W,\mu)\\
\Exp(h)&\longmapsto \ee^{\langle h,\enspace\rangle}\ee^{-\frac{1}{2}\|h\|^2}.
\end{align*}

\part{Construction of Quantum Field Theories}
\section{Free Scalar Field Theory}

Recall that a classical scalar field theory on $\R^n$ consists of the following data. A \textsf{space of fields} $\mathcal{F}=C^\infty(\R^n)$ and an \textsf{action functional}, which is a map $S:\mathcal{F}\to\R$ that is \textsf{local}, i.e. it only depends on fields and derivatives of fields. In free theory we are interested in the action functional $S$ which is of the form 
\[
S(\phi)=\int_{\R^n}\phi(\Delta+m^2)\phi\dd x,
\]
where $\dd x$ denotes the Lebesgue measure on $\R^n$. In quantum theory, we are interested in defining a measure of the form 
\begin{equation}
\label{QFTmeasure}
\ee^{-\frac{1}{2}S(\phi)}\mathscr{D}\phi
\end{equation}
on $\mathcal{F}$. We will see that it is possible to define a measure of the form \eqref{QFTmeasure} but it lives on a much larger space than $\mathcal{F}$. Next, we will discuss Gaussian measures on locally convex spaces and as a consequence we will define a measure of the form \eqref{QFTmeasure}.
\subsection{Locally convex spaces}
\begin{defn}[seperating points]
Let $V$ be a vector space. A family $\{\rho_\alpha\}_{\alpha\in A}$ of seminorms on $V$ is said to \textsf{seperate points} if $\rho_\alpha(x)=0$ for all $\alpha\in A$ implies $x=0$.
\end{defn}
\begin{defn}[Natural topology]
Given a family of seminorms $\{\rho_\alpha\}_{\alpha\in A}$ on $V$ there exists a smallest topology for which each $\rho_\alpha$ is continuous and the addition operation is continuous. This topology, which is denoted by $\mathcal{O}(\{\rho_\alpha\})$, is called the \textsf{natural topology} on $V$.
\end{defn}
\begin{defn}[Locally convex space]
A \textsf{locally convex space} is a vector space $V$ together with a family $\{\rho_\alpha\}$ of seminorms that seperate points.
\end{defn}
\begin{exe}
Show that the natural topology on a locally convex space is Hausdorff.
\end{exe}
Let $\varepsilon>0$ and $\alpha_1,...,\alpha_n\in A$. Define the set 
\[
N(\alpha_1,...,\alpha_n;\varepsilon)=\{v\in V\mid \rho_{\alpha_i}(v)<\varepsilon,\hspace{0.2cm}\forall i=1,2,...,n\}.
\]
One can check that 
\begin{enumerate}
\item{$N(\alpha_1,...,\alpha_n;\varepsilon)=\bigcap_{i=1}^nN(\alpha_i;\varepsilon)$.
}
\item{$N(\alpha_1,...,\alpha_n;\varepsilon)$ is convex.
}
\end{enumerate}
\begin{exe}
Check that the elements of
\[
\{N(\alpha_1,...,\alpha_n;\varepsilon)\mid \alpha_1,...,\alpha_n\in A,n\in\N,\varepsilon>0\}
\]
form a neighbourhood basis at $0\in V$.
\end{exe}
From $(2)$ it follows that a locally convex space $V$ has a neighbourhood basis at $0\in V$, where each open set in this basis is convex. This justifies the name \textsf{locally convex space}. One can define the notion of a Cauchy sequence and the notion of convergence in a locally convex space. Let $V$ be a locally convex space. The following are equivalent:
\begin{enumerate}
\item{$V$ is metrizable.}
\item{$0\in V$ has a countable neighbourhood basis.}
\item{The natural topology on $V$ is generated by some countable family of seminorms.}
\end{enumerate}
\begin{defn}[Fr\'echet space]
A complete metrizable locally convex space is called a \textsf{Fr\'echet space}.
\end{defn}
\begin{ex}[Schwartz space]
Let $\phi\in C^\infty(\R^n)$ and let $\alpha=(\alpha_1,...,\alpha_k)\in (\N\cup\{0\})^k$ and $\beta=(\beta_1,...,\beta_\ell)\in (\N\cup\{0\})^\ell$. Let $D^\beta=\frac{\partial^{\beta_1+\dotsm +\beta_\ell}}{\partial x_{i_1}^{\beta_1}\dotsm\partial x_{i_\ell}^{\beta_\ell}}$. Moreover, define 
\[
\|\phi\|_{\alpha,\beta}=\sup_{x\in\R^n}\vert x^\alpha D^\beta\phi(x)\vert,
\]
and 
\[
\mathcal{S}(\R^n)=\{\phi\in C^\infty(\R^n)\mid \|\phi\|_{\alpha,\beta}<\infty,\hspace{0.2cm}\forall \alpha,\beta\}.
\]
The space $\mathcal{S}(\R^n)$ is called the \textsf{Schwartz space} on $\R^n$. One can easily check that $\mathcal{S}(\R^n)$ is a locally convex space. In general $\mathcal{S}(\R^n)$ is a Fr\'echet space.
\end{ex}
\subsection{Dual of a locally convex space}
Let $V$ be a locally convex space and $V^*$ be the set of continuous linear functionals on $V$, i.e. $\ell\in v^*$ iff $\ell:V\to\R$ is linear and continuous. Given $x\in V$, define $\rho_x:V^*\to\R$ by $\rho_x(\ell)=\vert\ell(x)\vert$. One can easily check that $\rho_x$ is a seminorm. In fact $\{\rho_x\mid x\in V\}$ is a family of seminorms on $V^*$ that seperates points. Hence, $(V^*,\{\rho_x\mid x\in V\})$ is a locally convex space. The natural topology on $V^*$ induces by $\{\rho_x\mid x\in V\}$ is called the \textsf{weak*-topology} on $V^*$. A sequence $\{\ell_n\}$ in $V^*$ converges to $\ell\in  V^*$ in the weak-*topology iff $\rho_x(\ell_n)\to\rho_x(\ell)$ for all $x\in V$, i.e. $\ell_n(x)\to\ell(x)$ for all $x\in V$. The weak-*topology on $V^*$ is denoted by $\mathcal{O}(V^*,V)$. 
\begin{rem}
The space of linear functionals on $(V^*,\mathcal{O}(V^*,V))$ is exactly $V$.
\end{rem}

\subsection{Gaussian measures on the dual of Fr\'echet spaces}
\begin{thm}
Let $V$ be a Fr\'echet space. Then there is a bijection between the following sets
\begin{multline*}
\Big\{\text{Continuous poisitve definite symmetric bilinear forms on V}\Big\}\\
\longleftrightarrow\Big\{\text{Centered Gaussian measures on $(V^*,\mathcal{O}(V^*,V))$}\Big\}
\end{multline*}
\end{thm}
\begin{proof}
See \cite{VB,AG}.
\end{proof}
Let $C$ be a continuous positive definite symmetric bilinear form on $V$. The construction of the associated Gaussian measure on $V^*$ goes as follows. Let $F\subseteq V$ be a finite dimensional subspace of $V$. Let $C_F$ be the restriction of $C$ on $F$. Then $C_F$ is a symmetric bilinear form on $F$, which is positive definite. Hence $C_F$ defines a Gaussian measure $\mu_{C_F}$ on $F^*$ ($F^*$ can be identified with $F$) of the form 
\[
\dd\mu_{C_F}(x)=\left(\det\left(\frac{C_F}{2\pi}\right)\right)^{-\frac{1}{2}}\ee^{-\frac{1}{2}C_F(x,x)},
\]
where $C_F$ is identified with a positive definite matrix. In fact, one can think of $\mu_{C_F}$ to be a measure on the $F$ cylinder subsets of $V^*$. One can check that if $E\subseteq F$, then $\mu_{C_F}$ agrees with $\mu_{C_E}$ when restricted to the $E$ cylinder subsets of $V^*$. Now we can proceed as in the construction of the Wiener measure and show that there is a Gaussian measure $\mu_C$ on the $\sigma$-algebra of $V^*$ generated by cyclinder sets. This gives the construction of $\mu_C$.
\begin{cor}
Let $C$ be the bilinear form on $\mathcal{S}(\R^n)$ defined by 
\[
C(f,g)=\int_{\R^n}f(\Delta+m^2)^{-1}g\dd x.
\]
Then there is a Gaussian measure $\mu$ on $\mathcal{S}(\R^n)$ whose covariance is $C$.
\end{cor}
In this example the reproducing kernel space $K(\mu)$ of $\mu$ is $H^{-1}(\R^n)$, where $H^{-1}(\R^n)$ is the completion of $\mathcal{S}(\R^n)$ with respect to $C$. Hence, we have succeeding in defining the measure of the form $\ee^{-S(\phi)}\mathscr{D}\phi$, where 
\[
S(\phi)=\frac{1}{2}\int_{\R^n}\phi(\Delta+m^2)\phi\dd x.
\]
In other words, we have constructed the Gaussian measure associated to the free theory.
\begin{rem}
In this example, the Cameron-Martin space of $\mu$ is $H^1(\R^n)$, which is the completion of $\mathcal{S}(\R^n)$ with respect to the map 
\[
(f,g)\mapsto \int_{\R^n}f(\Delta+m^2)g\dd x.
\]
\end{rem}

\subsection{The operator $(\Delta+m^2)^{-1}$}

We regard $(\Delta+m^2)^{-1}$ as an operator on $L^2(\R^n)$. It is known that $(\Delta+m^2)^{-1}$ is a positive operator, and that it is an integral operator. Let $C(x,y)$ be the integral kernel of $(\Delta+m^2)^{-1}$. Then 
\[
C(f,g)=\iint_{\R^n\times\R^n}f(x)C(x,y)g(y)\dd x\dd y.
\]
In fact, one can show that $C(x,y)$ is the unique solution of 
\begin{equation}
\label{rel1}
\Delta_yC(x,y)=\delta_x(y).
\end{equation}
Using Fourier transform, we can give an explicit representation of $C(x,y)$. Formally we have the following chain of implications.
\[
(\Delta+m^2)^{-1}f=g\Rightarrow(\Delta+m^2)g=f\Rightarrow (\xi^2+m^2)g^{-1}=f^{-1}\Rightarrow g=\mathcal{F}^{-1}\left(\frac{1}{\xi^2+m^2}\widehat{f}\right).
\]
These hold since 
\[
g(x)=\frac{1}{(2\pi)^{n}}\iint_{\R^n\times \R^n}\frac{\ee^{\I\xi(x-y)}}{\xi^2+m^2}f(y)\dd y\dd\xi,
\]
and thus
\[
C(x,y)=\frac{1}{(2\pi)^n}\int_{\R^n}\frac{\ee^{\I\xi(x-y)}}{\xi^2+m^2}\dd\xi.
\]
For $x\not=y$ one can show that 
\[
C(x,y)=(2\pi)^{-\frac{1}{2}}\left(\frac{m}{\|x-y\|}\right)^{\frac{n-2}{2}}K_{\frac{n-2}{2}}(m\|x-y\|),
\]
where $K_\nu$ is a modified \textsf{Bessel function}. Next we will study $C(x,y)$ in more details. In particular the behaviour of $C(x,y)$ when $\|x-y\|$ is large and $\|x-y\|$ is small.
\begin{rem}
For $n=1$ we have $C(x,y)=\frac{\ee^{-m\vert x-y\vert}}{m}$ and for $n=3$ we have $C(x,y)=\frac{\ee^{-m\|x-y\|}}{4\pi\|x-y\|}$.
\end{rem}
\begin{prop}[Properties of $C(x,y)$]
The following hold:
\begin{enumerate}
\item{For every $m\vert x-y\vert$ bounded away from zero, there exists some $M\geq 0$ such that we have 
\[
C(x,y)\leq Mm^\frac{n-3}{2}\vert x-y\vert^{\frac{n-1}{2}}\ee^{-m\vert x-y\vert}.
\]
}
\item{For $n\geq 3$ and $m\vert x-y\vert$ in a neighbourhood of zero we get 
\[
C(x,y)\sim \vert x-y\vert^{-n+2}
\]
}
\item{For $n=2$ and $m\vert x-y\vert$ in a neighborhood of zero we get 
\[
C(x,y)\sim -\log\left(m\vert x-y\vert\right).
\]
}
\end{enumerate}
\end{prop}
\begin{proof}
Recall first 
\begin{equation}
\label{rel2}
C(x,y)=\frac{1}{2\pi}\int_\R\frac{\ee^{\I\xi(x-y)}}{\xi^2+m^2}\dd\xi.
\end{equation}
\begin{exe}
Show that in general we have 
\[
C(x,y)=\frac{1}{(2\pi)^n}\int_{\R^n}\frac{\ee^{\I\xi\|x-y\|}}{\xi^2+m^2}\dd\xi.
\]
Hint: Choose an orthonormal basis $\{e_1,...,e_n\}$ of $\R^n$ such that $e_1=\frac{x-y}{\|x-y\|}$ and do a change of variables.
\end{exe}
Now, using the residue theorem, we have 
\begin{multline*}
C(x,y)=\frac{1}{(2\pi)^n}\int_{\R^n}\frac{\ee^{\I t\xi_1}}{\xi_1^2+\left(\sqrt{m^2+\xi_2^2+\dotsm +\xi_n^2}\right)^2}\dd\xi_1\dd\xi_2\dotsm \dd\xi_n\\=\frac{1}{(2\pi)^n}\int_{\R^{n-1}}\frac{\pi\ee^{-t\sqrt{m^2+\xi_2^2+\dotsm+\xi_n^2}}}{\sqrt{m^2+\xi_2^2+\dotsm +\xi^2_n}}\dd\xi_2\dotsm\dd\xi_n.
\end{multline*}
Without loss of generality, assume that $m=1$. Recall that for $f:\R^n\to\R$ with $f(x)=g(\vert x\vert)$ we have 
\begin{equation}
\label{relation1}
\int_{\R^n}f(x)\dd x=v(S^{n-1})\int_0^\infty r^{n-1}g(r)\dd r,
\end{equation}
where $v(S^{n-1})$ is the volume of $S^{n-1}$. Using \eqref{relation1}, we can write 
\begin{equation}
\label{relation1}
C(x,y)=\frac{\pi A_{n-1}}{(2\pi)^{\frac{n}{2}}}\int_0^\infty \frac{r^{n-2}\ee^{-t\mu(r)}}{\mu(r)}\dd r,
\end{equation}
where $\mu(r)=\sqrt{1+r^2}$. 
\begin{exe}
Show that there is some $\varepsilon>0$ such that 
\[
\mu(r)\geq\begin{cases}1+\varepsilon r^2,&\text{if $r\leq 1$}\\1+\varepsilon r,&\text{if $r\geq 1$}\end{cases}
\]
\end{exe}
We will claim that 
\begin{equation}
\label{rel3}
C(x,y)\leq k\ee^{-t}(t^{-\frac{n-1}{2}}+t^{-(n-1)}),
\end{equation}
where $t=\|x-y\|$ and $k$ some constant. Note then that 
\begin{equation*}
\int_0^1\frac{r^{n-2}\ee^{-t\mu(r)}}{\mu(r)}\leq\int_0^1r^{n-2}\ee^{-t(1+\varepsilon r^2)}\dd r
\leq k\ee^{-t}t^{-\frac{n-1}{2}}
\end{equation*}
and 
\[
\int_1^\infty\frac{r^{n-2}\ee^{-t\mu(r)}}{\mu(r)}\leq \int_1^\infty r^{n-2}\ee^{-t(1+\varepsilon r)}\dd r\leq k \ee^{-t}t^{-(n-1)}.
\]
If $t\geq 1$ then 
\[
C(x,y)\leq k \ee^{-t}t^{-\frac{n-1}{2}}.
\]
For $0< t\leq 1$ we have 
\begin{align*}
\int_0^\infty r^{n-2}\frac{\ee^{-t\sqrt{1+r^2}}}{1+r^2}\dd r&=t^{-(n-2)}\int_0^\infty s^{n-2}\frac{\ee^{-\sqrt{s^2+t^2}}}{\sqrt{s^2+t^2}}\dd s\\
&\sim_{\text{\{as $t\to 0$ in the integral\}}} t^{-(n-2)}\int_0^\infty \frac{s^{n-2}\ee^{-s}}{s}\dd s\\
&=t^{-(n-2)}\int_0^\infty s^{n-3}\ee^{-s}\dd s
\end{align*}
If $n=2$, let $s=t\mu(r)$. Then $\mu(r)=\frac{s}{t}$, $1+r^2=\frac{s^2}{t^2}$ and $r=\sqrt{\frac{s^2-t^2}{t^2}}$. Thus 
\[
C(x,y)=\int_t^\infty\frac{\ee^{-s}}{\sqrt{s^2+t^2}}\dd s\sim \int_t^\infty \frac{1}{\sqrt{s^2+t^2}}\dd s\sim -\log(t).
\]
\end{proof}

\section{Construction of self interacting theory}

To construct a theory with polynomial interaction, we want to define a measure of the form 
\[
\ee^{-S(\phi)}\mathscr{D}\phi
\]
rigorously where
\[
S(\phi)=\frac{1}{2}\int_{\R^n}\phi(\Delta+m^2)\phi\dd x+\int_{\R^n}P(\phi)\dd x=S_f(\phi)+S_I(\phi),
\]
with $P(y)=\sum_i a_iy^{i}$ some polynomial function on $\R$. We have succeeded in defining a measure $\mu$ of the form $\ee^{-S_f(\phi)}\mathscr{D}\phi$, but the price we had to pay was that it lives on $\mathcal{S}(\R^n)$. In fact $\mu(\mathcal{S}(\R^n))=0$ because for such measures the Cameron-Martin space is $H^1(\R^n)$ and $\mathcal{S}(\R^n)\subseteq H^1(\R^n)$. Hence, it is not obvious that we have to view $\phi\mapsto \int_{\R^n}\phi^n(x)\dd x$ as a measurable function on $\mathcal{S}(\R^n)$. Let us now define a measure of the form 
\[
\ee^{-S_f(\phi)}\ee^{-S_I(\phi)}\mathscr{D}\phi.
\]
First, we will try to ``define'' measurable functions of the form 
\begin{equation}
\label{map1}
\phi\longmapsto \int_{\R^n}\phi(x)^k\dd x.
\end{equation}
In particular, we will try to bypass the difficulties in making sense of \eqref{map1}. Let us pretend that we can define \eqref{map1}. Formally we have 
\begin{align*}
\left\|\int_{\R^n}\phi(x)^k\dd x\right\|^2_{L^2(\mathcal{S}(\R^n),\mu)}&=\int_{\phi\in \mathcal{S}(\R^n)}\left(\int_{\R^n}\phi(x)^k\dd x\right)\left(\int_{\R^n}\phi(y)^k\dd y\right)\dd\mu(\phi)\\
&=\iint_{\R^n\times\R^n}\left(\int_{\mathcal{S}(\R^n)}\phi(x)^k\phi(y)^k\dd\mu(\phi)\right)\dd x\dd y.
\end{align*}
Now, formally thinking of $\phi(x)$ as $\delta_x[\phi]=\langle \delta_x|\phi\rangle$, which is defined in terms of the Heaviside step function, and using Wick's theorem, we see that 
\begin{multline*}
\left\|\int_{\R^n}\phi(x)^k\dd x\right\|_{L^2(\mathcal{S}(\R^n),\mu)}=\text{Linear combination of integrals of the form} \\ \iint_{\R^n\times\R^n}C(x,x)^\alpha C(y,y)^\beta C(x,y)^\gamma \dd x\dd y. 
\end{multline*}
The existence of $\int_{\R^n}\phi(x)^k\dd x$ depends on the properties of $C(x,y)$ and hence it is \textsf{dimension sensitive}. In fact we can not define \eqref{map1} because $C(x,x)$ is not integrable. We want to try two different attempts to make sense of $\int_{\R^n}\phi(x)^k\dd x$:
\begin{enumerate}
\item{(Approximation of delta function) Let $h\in\mathcal{S}(\R^n)$ such that $h\geq 0, h(0)>0$ and $\int_{\R^n}h=1$. Consider e.g. $h_\varepsilon(y)=\frac{1}{\varepsilon^2}h\left(\frac{y}{\varepsilon}\right)$ for some $\varepsilon>0$. Then $h_\varepsilon\to \delta_0$ as $\varepsilon\to 0$. Similarly\footnote{With $x=\frac{1}{\varepsilon}$ we get $h_x(x)=x^2h(xx)\to \delta_0$ as $x\to\infty$.} we can construct $\delta_{x,\varepsilon}\in\mathcal{S}'(\R^n)$ (space of \textsf{Schwartz distributions on $\R^n$}) such that $\delta_{x,\varepsilon}\to\delta_x$. Then $\phi\mapsto \delta_{x,\varepsilon}[\phi]=\langle\delta_{x,\varepsilon}|\phi\rangle$ is a polynomial function on $\mathcal{S}(\R^n)$. We denote the polynomial type by $\phi(\delta_{x,\varepsilon})$. We know how to compute 
\[
\int_{\phi\in\mathcal{S}(\R^n)}\phi(\delta_{x,\varepsilon})^k\phi(\delta_{y,\varepsilon})^m\dd\mu(\phi),
\]
which is equal to the sum of terms of the form $A_{\alpha\beta\gamma}C(\delta_{x,\varepsilon},\delta_{x,\varepsilon})^\alpha C(\delta_{y,\varepsilon},\delta_{y,\varepsilon})^\beta C(\delta_{x,\varepsilon},\delta_{y,\varepsilon})^\gamma$. Formally we have that
\[
\left\langle \int_{\R^n}\phi(\delta_{x,\varepsilon})^k\dd x,\int_{\R^n}\phi(\delta_{y,\varepsilon})^m\right\rangle_{L^2(\mathcal{S}(\R^n),\mu)}
\]
is equal to sum of expressions of the form 
\[
A_{\alpha\beta\gamma}\iint_{\R^n\times\R^n}C(\delta_{x,\varepsilon},\delta_{x,\varepsilon})^\alpha C(\delta_{y,\varepsilon},\delta_{y,\varepsilon})^\beta C(\delta_{x,\varepsilon},\delta_{y,\varepsilon})^\gamma \dd x\dd y,
\]
and we can try to take the limit $\varepsilon\to 0$. The conclusion here is that this attempt does not lead to anywhere, since we still get a diagonal contribution. 
}
\item{(Redefine obseravbles) Let us try to get rid of diagonal contribution (i.e. terms like $C(x,x)$). This is where the Wick ordering comes into the play. We can think of Wick ordering as a renormalization process. Consider the map 
\[
\phi\longmapsto \int_{\Lambda\subset \R^n\atop\Lambda \text{ compact}}(:\phi(\delta_{x,\varepsilon})^k:)\dd x.
\]
Thus we get 
\[
\iint_{\R^n\times\R^n}\left(\int_{\mathcal{S}(\R^n)}(:\phi(\delta_{x,\varepsilon})^k:)(:\phi(\delta_{y,\varepsilon})^k:)\dd\mu(\phi)\right)\dd x\dd y=k!\iint_{\R^n\times\R^n} C(\delta_{x,\varepsilon},\delta_{y,\varepsilon})^k\dd x\dd y.
\]
Taking the limit $\varepsilon\to 0$ formally, it converges to 
\begin{equation}
\label{limit1}
k!\iint_{\R^n\times\R^n}C(x,y)^k\dd x\dd y.
\end{equation}
}
\end{enumerate}
Let us list what we know so far:
\begin{enumerate}[$(i)$]
\item{Wick ordering allows us to get rid of diagonal contribution of $C(x,y)$.}
\item{If $\iint_{\R^n\times\R^n}C(x,y)^k\dd x\dd y<\infty$, then there is a hope that we can define ``Wick ordered polynomial'' functions of the form 
\[
\phi\longmapsto \int_{\R^n}(:\phi(x)^k:)\dd x.
\]
}
\end{enumerate}
Recall that if $n\geq 3$, then $C(x,y)\sim \frac{1}{\|x-y\|^{n-2}}$ for $\|x-y\|\to 0$ and hence the integral of the form \eqref{limit1} will diverge in general. This means that if $n\geq 3$, Wick ordering renormalization may not kill all the infinites appearing in Feynmann amplitudes. However, if $n=2$, we get $C(x,y)\sim -\log(\|x-y\|)$ as $\|x-y\|\to 0$ and in this case it is possible to define observables of the form $\phi\mapsto \int_{\Lambda\subset \R^2\atop\Lambda\text{ compact}}(:P(\phi):)$. Define $\widetilde{S}_{I,\Lambda}(\phi)=\int (:P(\phi)(x):)\dd x$ as a measurable function on $\calS'(\R^2)$, where $P$ is any polynomial. Let $P(x)=x^4$ and $\Lambda\subset \R^2$. Then we can consider $\widetilde{S}_{I,\Lambda}^\varepsilon(\phi)=\int_\Lambda (:P(\phi,\delta_{\varepsilon,x}):)\dd x$, where $\delta_{\varepsilon,x}$ is a smooth approximation of $\delta_x$. More precisely, $\delta_{\varepsilon,x}$ can be constructed as follows. Let $h\in C_0^\infty(\R^2)$ with $h\geq 0$, $h(0)>0$, and $\int_{\R^2}h=1$. Then consider $\delta_{\varepsilon,x}(y)=\varepsilon^{-2}h\left(\frac{x-y}{\varepsilon}\right)$. In fact, $\delta_{\varepsilon,x}\to \delta_x$ in $\calS'(\R^2)$. If we take $\varepsilon=\frac{1}{k}$, we will get $\delta_{\varepsilon,x}=\delta_{k,x}$ and thus one can observe that the sequence $\{\widetilde{S}^k_{I,\Lambda}\}_k$ is Cauchy in $L^2(\calS'(\R^2,\mu))$. Define $\widetilde{S}_{I,\Lambda}(\phi):=\lim_{k\to\infty} \widetilde{S}_{I,\Lambda}^k(\phi)$.  

\begin{rem}
Recall \
$$\left\langle \int_\Lambda(:\phi(\delta_{k,x}):)^n,\int_{\Lambda}(:\phi(\delta_{k,y}):)^n\right\rangle_{L^2(\calS'(\R^2,\mu))}=n!\int_{\Lambda\times \Lambda}C(\delta_{k,x},\delta_{k,y})^n\dd x\dd y.$$
To see that $\{\widetilde{S}_{I,\Lambda}^k\}_k$ is Cauchy we only have to understand how $C$ behaves. In the Fourier picture it is easier understood. We have the Fourier transform of $\delta_{k,x}$ is $\left(\frac{1}{\xi^2+m^2}\right)\widehat{h}^2\left(\frac{\xi}{k}\right)$, which can be understood very easily. Then one can use the properties of $C_k(x,y):=C(\delta_{k,x},\delta_{k,y})$, which is a smooth approximation of Green function, to show that $\int_\Lambda (:\phi(\delta_{k,x}):)^n\dd x$ is Cauchy.
\end{rem}

\subsection{More random variables}
Let $f\in C_0^\infty(\underbrace{\R^2\times\dotsm \times\R^2}_k)$. Then 
$$\widetilde{S}_{I,\Lambda}(f,k)=\int_{\R^2\times\dotsm \times \R^2}(:\phi(x_1)\dotsm \phi(x_k):)f(x_1,...,x_k)\dd x_1\dotsm \dd x_k$$
can be defined as before. More generally we can take $f\in L^2(\R^2\times \dotsm \times \R^2)$. Moreover, we can also define 
$$A(\phi)=\prod_{i=1}^nS_I(f_i,k_i).$$
We want to know how we can compute $\int_{\calS'(\R^2)}A(\phi)\dd\mu(\phi)$. For that, we recall: For $(W,\mu)$ a measure space and $f\in W^*$ we have 
$$(:f:)^n=\sum_{k=0}^{\lfloor \frac{n!}{2}\rfloor}\frac{n!}{k!(n-2k)!}f^{n-2k}\left(\frac{-q_\mu
(f,f)}{2}\right)^k.$$
We would like to know if an expression of the form
$$(:f_1\dotsm f_k:)(:g_{k+1}\dotsm g_n:)$$
can be written as a linear combination of Wick ordered polynomials. The answer is yes, and the advantage is that it allows us to compute integration of product of Wick ordered polynomials easily.

\begin{ex} We can write
$$(:f_1\dotsm f_n:)=f_1(:f_2\dotsm f_n:)-\sum_{j=2}^n q_n(f_1,f_j)(:f_2\dotsm \widehat{f}_j\dotsm f_n:),$$
which is similar to integration by parts. Here the symbol $\widehat{\enspace}$ means that the element is omitted.
\end{ex}

\subsection{Generalized Feynman diagrams}

Let $I=I_1\sqcup I_2\sqcup\dotsm \sqcup I_n$, be the disjoint union of finite sets $I_i$ for $i\in\{1,...,n\}$. 
\begin{defn}[Generlized Feynman diagram] A \textsf{generalized Feynman diagram} is a pair $(I,E)$, where 
$$E\subseteq \{(a,b)\mid \text{ $a$ and $b$ do not belong to some $I_i$, $i\in\{1,...,n\}$}\}.$$ 
\end{defn}
We denote by $A_E$ the remaining vertices. Let $F=(V,E)$ be a generalized Feynman diagram associated to $(:f_1\dotsm f_k:)(:g_{k+1}\dotsm g_n:)$. Then 
\[
V(F)=\left(\prod_{e\in E}q_\mu(f_{\ell(e)},g_{r(e)})\right)\left(:\prod_{v\in A_E}\alpha_v:\right),
\]
where $\alpha_v$ is either $f_v$ or $g_v$ and $\ell(e)$ is the left end point and $r(e)$ the right end point.

\begin{cor}
$$\int (:f_1\dotsm f_k:)(: g_{k+1}\dotsm g_n:)\dd\mu=\text{sum of value of complete Feynman diagrams}.$$
\end{cor}

Consider again the integral $\int_{\calS'(\R^2)}A(\phi)\dd\mu(\phi)$. This can be computed using generalized Feynman diagrams. Our goal is to show that $\ee^{-\widetilde{S}_{I,\Lambda}}\in L^1(\calS'(\R^n))$. Consider thus $\widetilde{S}^k_{I,\Lambda}(\phi)=\int_\Lambda (:\phi^{2k}(x):)\dd x$. We want to know whether $\ee^{-\widetilde{S}^k_{I,\Lambda}}\in L^1(\calS'(\R^n))$. E.g. $(:x^4:)=x^4-6x^2+3$, then $\ee^{-(:x^4:)}$ can behave bad. 

\begin{lem}
$$\widetilde{S}^k_{I,\Lambda}\geq -b(\log k)^n,$$
as $k\to\infty$ for some $b>0$.
\end{lem}
\begin{rem}
This shows that $\widetilde{S}_{I,\Lambda}^k$ does not gneralize to a polynomial, which is is not bounded from below.
\end{rem}
\begin{proof}
Let $Q(y)=\sum_{k=0}^{2n}a_ky^k$, for $a_{2n}>0$. Then 
\[
\inf_{y\in\R}Q(y)\geq -b,
\]
for some $0\leq b<\infty$, and 
\begin{multline*}
(:\phi(\delta_{k,x})^{2n}:)=\sum_{k=0}^{2n}=\frac{(2n)!}{k!}\phi(\delta_{k,x})^{2n-2k}\left((-1)\cdot\frac{C(\delta_{k,x},\delta_{k,x})}{2}\right)^k\\=C_k(x,x)^n\sum_{k=0}^{2n}\frac{(2n)!}{k!(2n-2k)!}\frac{(-1)^k}{2^k}\left(\frac{\phi(\delta_{k,x})}{\sqrt{C_k(x,x)}}\right)^{2n-2k}.
\end{multline*}
Thus $(:\phi(\delta_{k,x}):)^{2n}\geq -b\int_\Lambda C_k(x,x)^n$ for some $b>0$ and hence $\widetilde{S}_{I,\Lambda}^k(\phi)\geq -b\int_\Lambda C_k(x,x)^n\geq -\tilde{b}(\log k)^n$ as $k\to \infty$.
\end{proof}
\begin{cor}
$\ee^{-\widetilde{S}_{I,\Lambda}^k}\in L^p(\calS'(\R^n))$ for all $p$.
\end{cor}
Consider $S^{I}(P,f)(\phi)=\int_{\R^2}f(x)(:P(\phi(x)):)\dd x$, where $P(x)=\sum_n a_nx^n$ is a function on $\calS'(\R^2)$ and $f\in L^2(\R^2)$. We showed, $S^{I}(P,f)\in L^2(\calS'(\R^2),\mu)$. Let $S^{I,k}(P,f)=\int f(x)P(\phi,\delta_{k,x})\dd x$, where $\delta_{k,x}$ is a smooth approximation of $\delta_x$. We showed, if $S^{I,k}$ is Cauchy, then $S^{I}(P,f)=\lim_{k\to \infty}S^{I,k}(P,f)$. In fact, 
\[
\| S^{I,k}(P,f)-S^{I}(P,f)\|_{L^2(\calS'(\R^2),\mu)}\leq Cf^{-\delta},
\]
for some $\delta>0$ as $k\to\infty$. Moreover, $\ee^{-S^{I,k}(P)}\in L^1(\calS'(\R^n))$, where $S^{I,k}_\Lambda(P)(\phi)=\int_\Lambda(:P(\phi(x)):)\dd x$ with $P(x)=x^{2k}$. The idea is that $S^{I,k}_\Lambda(\phi)\geq -C(\log k)^n$ for some $C>0$. We can observe that $S^{I,k}_\Lambda(\phi)\geq 1-\widetilde{C}\log (k)^n$ for some $\widetilde{C}>0$ for large $k$ (take $\widetilde{C}=\frac{2}{3}C$). The goal was to show that $\ee^{-S^{I}_\Lambda(P)}\in L^1(\calS'(\R^2),\mu)$. The strategy is to study the sets, where $S^{I}_\Lambda(P)$ is bad, and then show that these sets have measure zero. Define a ``bad set''
\[
X(k):=\{\phi\in\calS'(\R^2)\mid S^{I}_\Lambda(P)(\phi)\leq \widetilde{C}(\log k)^n\}.
\]
\begin{lem}
$$X(k)\subseteq \{\phi\in\calS'(\R^2)\mid \vert S^{I}_\Lambda(P)(\phi)-S^{I,k}_\Lambda(P)(\phi)\vert\geq 1\}.$$
\end{lem}
\begin{proof}
Let $\phi\in X(k)$. Then $$S^{I}_\Lambda(P)(\phi)-S^{I,k}_\Lambda(P)(\phi)\leq S_\Lambda^{I}(P)(\phi)-(1-\widetilde{C}(\log k)^n=\underbrace{S^{I}_\Lambda(P)(\phi)+\widetilde{C}(\log k)^n}_{\leq 1}-1\leq -1.$$
\end{proof}

\begin{prop}
There is a $B>0$ and $\delta>0$ such that $\mu(X(k))\leq Bk^{-\delta}$ as $k\to\infty$.
\end{prop}

\begin{proof}
We have 
\begin{equation*}
\mu(X(k))=\int_{X(k)}\dd\mu\leq \int_{X(k)}\vert S^{I}_\Lambda(P)-S^{I,k}_\Lambda(P)\vert^2\dd\mu\leq \int_{\calS'}\vert S^{I}_\Lambda(P)-S^{I,k}_\Lambda(P)\vert^2\dd\mu\leq B_\Lambda k^{-\delta}
\end{equation*}
as $k\to\infty$.
\end{proof}

\begin{rem}
One can show that $\mu(X(k))\leq C \Exp\left({-k^\alpha}\right)$ for some $\alpha>0$ as $k\to\infty$.
\end{rem}

Let $(\Omega,\sigma(\Omega),\mu)$ be a probability space and $f\colon \Omega\to\R$ a measurable function on $\Omega$. Denote by
\[
\mu_f(x)=\mu(\{\omega\in\Omega\mid f(\omega)\geq x\}).
\]
Let $F$ be an increasing positive function on $\R$ such that $\lim_{x\to \infty}F(x)=\infty$. Then $$\int_\Omega F(f(\omega))\dd\mu(\omega)=\int_\R F(x)\mu_f(x)\dd x$$.

\begin{thm}
Let $f$ be a measurable function on $\Omega$ such that $\mu(\{\omega\in \Omega\mid -f(\omega)\geq C(\log k)^n\})\leq C\ee^{-k_0^\alpha}$ for $k\geq k_0$. Then 
\[
\int_\Omega \ee^{-f(\omega)}\dd\mu(\omega)<\infty.
\]
\end{thm}
\begin{proof}
We have 
\begin{multline*}
\int_\Omega \ee^{-f(\omega)}\dd\mu(\omega)=\int_{\{\omega\in \Omega\mid f(\omega)<C(\log k)^n\}}\ee^{-f(\omega)}\dd\mu(\omega)+\int_{\{\omega\in \Omega\mid -f(\omega)\geq C(\log k_0)^n\}}\ee^{-f(\omega)}\dd\mu(\omega)\\
\leq B_1\int \ee^{x}\mu_f(x)\dd x\leq B_1+\int \ee^x\Exp\left(-\ee^{\alpha\left(\frac{x}{C}\right)^{1/n}}\right)\dd x<\infty.
\end{multline*}
\end{proof}

\begin{cor}
$\ee^{-S^{I}_\Lambda(P)}\in L^1(\calS'(\R^n))$.
\end{cor}
\begin{proof}
Take $f=S^{I}_\Lambda(P)$ and $\Omega=\calS'(\R^2)$.
\end{proof}

\begin{rem}
If $P$ is a polynomial of the form $P(x)=\sum_{k=0}^{2n}a_kx^k$ with $a_{2n}>0$, and $f\in L^1(\R^2)\cap L^2(\R^2)$ with $f\geq 0$, then we can show that $\ee^{-S^{I}(P,f)}\in L^1(\calS'(\R^2),\mu)$.
\end{rem}

\begin{cor}
$\frac{\ee^{-S^{I}_\Lambda(P,f)}}{\int_{\calS'(\R^2)} \ee^{-S^{I}_\Lambda(P,f)}}$ is a probability measure on $\calS'(\R^2)$.
\end{cor}

\subsection{Theories with exponential interaction}
Consider the potetial $$V_g^\alpha(\phi)=\int g(x)(:\Exp(\alpha\phi(x)):)\dd x.$$ We want to define a theory for this type of interaction. Moreover, we want to show that $V_g^\alpha\in L^2(\calS'(\R^2),\mu)$ with certain assumption on $\alpha$ and $g$. Define 
\[
V_k^\alpha(g)=\int_{\R^2}g(x)(:\Exp(\alpha(\phi,\delta_{k,x})):)\dd x.
\]
Recall $(:\Exp(\alpha f):)=\sum_{k=0}^\infty\frac{\alpha^k}{k!}(:f^k:)$ for $f\in \calS(\R^2)$.

\begin{lem}
We get $V_k^\alpha(g)\in L^2(\calS'(\R^2),\mu)$, whenever $g\in L^1(\R^2)\cap L^2(\R^2)$ and $0\leq \alpha^2\leq 4\pi$.
\end{lem}

\begin{proof}
We have 
\begin{equation*}
\left\langle (:\Exp(\alpha f):),(:\Exp(\alpha g):)\right\rangle=\Exp(\alpha^2C(f,g)),
\end{equation*}
and thus 
\[
\|V_k^\alpha(g)\|^2=\iint_{\R^2\times\R}g(x)g(y)\Exp(\alpha^2C(\delta_{k,x},\delta_{k,x}))\dd x\dd y=\iint_{\R^2\times \R^2}g(x)g(y)\Exp(\alpha^2 C_k(x,y))\dd x\dd y,
\]
where $C_k(x,y)=C(\delta_{k,x},\delta_{k,y})$. We know $C_k(x,y)\leq C(x,y)$ and $\iint_{\R^2\times\R^2} g(x)g(y)\exp(\alpha^2C(x,y))\dd x\dd y<\infty$, whenever $0\leq \alpha^2<4\pi$, and $g\in L^1(\R^2)\cap L^2(\R^2)$. The latter is true for $\| x-y\|\geq 1$ and $\| x-y\|<1$ gives $\exp\left(\alpha^2-\frac{\log\| x-y\|}{2\pi}\right)=\|x-y\|-\frac{\alpha^2}{4\pi}$. Hence $\|V_k^\alpha(g)\|^2<\infty$.
\end{proof}

\begin{prop}
$\{V_k\}_k$ converges in $L^2(\calS'(\R^2),\mu)$.
\end{prop}

\begin{proof}
Recall that $V_k=\sum_{k=0}^\infty\frac{\alpha^k}{k!}\int g(x)(:f_k(x)^k:)\dd x$. The Weierstrass M-test tells us that for any metric space $(X,d)$, a Banach space $W$, $f_k\colon X\to W$ with $\vert f_k(x)\vert \leq M_k$ with numbers $M_k>0$ such that $\sum_{k=0}^\infty M_k<\infty$, then $\sum_{k=0}^\infty f_k(x)$ converges uniformly for $x$.
\end{proof}

\begin{exe}
Show that there is a $C_k>0$ such that $\left\|\frac{\alpha^2}{k!}\int g(x)(:\phi_k(x)^k:)\dd x\right\|^2_{L^2(\calS'(\R^2),\mu)}\leq C_k$.
\end{exe}

This implies that $V_k$ converges uniformly on $X$ (by the M-test). Recall that
$$\frac{\alpha^k}{k!}\int g(x)(:\phi_k(x)^k:)\dd x\to \frac{\alpha^k}{k!}\int g(x)(:\phi(x)^k:)\dd x,$$
where $\phi_k=(\phi,\delta_{k,x})$. If $V=\lim_{k\to\infty}V_k$, then $$V=\sum_{k=0}^\infty\frac{\alpha^k}{k!}\int g(x)(:\phi(x)^k:)\dd x=\int g(x)(:\Exp(\alpha\phi(x)):)\dd x.$$
Thus we have shown that $V\in L^2(\calS'(\R^2),\mu)$.
We can observe the following:
\begin{enumerate}
\item{$V_k\geq 0$ for all $k$, whenever $g\geq 0$. this implies that for such $g$, $V\geq 0$ and hence $\ee^{-V}\in L^1$.
}
\item{For $0\geq \alpha^2<4\pi$ and $g\geq 0$ with $g\in L^1(\R^2)\cap L^2(\R^2)$, we showed $\ee^{-V^\alpha_g}\in L^1(\calS'(\R^2),\mu)$. Let $\nu$ be a measure on $[-\alpha,\alpha]$. Then $$\int_{[-\alpha,\alpha]}\ee^{-V_g^{\alpha'}}\dd\nu(\alpha')\in L^1(\calS'(\R^2),\mu).$$
}
\end{enumerate}

\subsection{The Osterwalder-Schrader Axioms}
Let $\mu$ be a Borel measure on $\calS'(\R^n)$.
\subsubsection{Analyticity (OS0)}
Let $f_1,...,f_k\in\calS'(\R^n)$. Define a function $\widehat{\mu}(f_1,...,f_k)\colon \mathbb{C}^k\to \mathbb{C}$ 

by $\widehat{\mu}(f_1,...,f_k)(z_1,...,z_k)=\widehat{\mu}\left(\sum_jz_jf_j\right)$, where $\widehat{\mu}$ is the characteristic function of $\mu$, i.e. 
\[
\widehat{\mu}(f)=\int_{\calS'(\R^n)}\ee^{\I\phi(f)}\dd\mu(\phi).
\]
\begin{defn}[Analyticity]
We say that $\mu$ is \textsf{analytic} or $\mu$ has \textsf{analyticity} if $\widehat{\mu}(f_1,...,f_k)$ is entire on $\mathbb{C}^k$ for all $f_1,...,f_k\in \calS(\R^n)$ and $k\in \N$. This means that $\mu$ decays faster than any exponential map. 
\end{defn}
\begin{rem}
An immediate consequence is that $\int_{\calS'(\R^n)}\phi(f)\dd\mu(\phi)<\infty$ for all $k$ and for all $f\in \calS(\R^n)$. Then 
$$\widehat{\mu}(\I f)=\int_{\calS'(\R^n)}\ee^{-\phi(f)}\dd\mu(\phi)<\infty$$
and 
$$\widehat{\mu}(-\I f)=\int_{\calS'(\R^n)}\ee^{\phi(f)}\dd\mu(\phi)<\infty,$$
if $\mu$ is analytic.
\end{rem}

\begin{ex}
Let $\mu$ be the Gaussian measure on $\calS'(\R^n)$, whose covariance is given by $(\Delta+m^2)^{-1}$, and let $C(f,g)=\int f(x)C(x,y)g(y)\dd x\dd y=\langle f,(\Delta+m^2)^{-1}f\rangle_{C^2(\R^n)}$. We claim that $\mu$ has analyticity. We prove this via an example. Consider $\widehat{\mu}(z_1f_1+z_2f_2)=\int_{\calS'(\R^n)}\ee^{\I\phi (z_1f_1+z_2f_2)}\dd\mu(\phi)=\ee^{-\frac{1}{2}(2z_1z_2C(f_1,f_2)+z_1^2C(f_1,f_1)+z_2^2C(f_2,f_2))}$, which is obviously an entire function, and $\widehat{\mu}(z_1f_1+z_2f_2)=\widehat{\mu}(f_1,f_2)(z_1,z_2)$ which is analytic.
\end{ex}

\begin{prop}
Let $\nu$ be any Gaussian measure on $\calS'(\R^n)$. Then $\nu$ has analyticity.
\end{prop}

\subsubsection{Euclidean invariance (OS1)}
Let $E(n)$ be the Euclidean group of $\R^n$, i.e. the group generated by rotations, reflections and translations. Let $R\in O(n)$ and $a\in \R^n$. Let $T(a,R)\in E(n)$ be defined by $(T(a,R))(x)=Rx+a$. Notice that $E(n)$ acts on $\calS(\R^n)$ by $(T(a,R)f)(x)=f(T(a,R)^{-1}x)$. $E(n)$ acts also on $\calS'(\R^n)$ by $(T(a,R)\phi)(f)=\phi(T(a,R)f)$.

\begin{defn}[Euclidean invariance I]
We say that $\mu$ is \textsf{Euclidean invariant} if $(T(a,R))_*\mu=\mu$ for all $T(a,R)\in E(n)$.
\end{defn}

\begin{lem}
$\mu$ is Euclidean invariant if and only if $\widehat{\mu}(f)=\widehat{\mu}(T(a,R)f)$ for all $f\in \calS(\R^n)$. 
\end{lem}

\begin{defn}[Euclidean invariance II]
Let $\nu$ be a Gaussian measure on $\calS'(\R^n)$ whose covariance is $C_\nu\colon \calS(\R^n)\times \calS(\R^n)\to \R$. We say $C_\nu$ is \textsf{Euclidean invariant} if $Cov(T(a,R)f,T(a,R)g)=C_\nu(f,g)$ for all $T(a,R)\in E(n)$ and $f,g\in\calS(\R^n)$. 
\end{defn}

\begin{lem}
Let $\nu$ be Gaussian. Then $\nu$ is Euclidean invariant if and only if $C_\nu$ is Euclidean invariant.
\end{lem}

\begin{ex}
Let $\mu$ be the Gaussian measure on $\calS'(\R^n)$ with covariance $(\Delta+m^2)^{-1}$. Then $\mu$ is Euclidean invariant. We have 
\[
C(x,y)=\frac{1}{(2\pi)^n}\int\frac{\ee^{\I\xi_1\|x-y\|}}{m^2+\xi^2}\dd\xi.
\]
\end{ex}

Next, we want to construct a Hilbert space $\calE=L^2(\calS'(\R^n),\mu)$.

\subsubsection{Reflection positivity (OS3)}
Let $f_1,...,f_k\in\calS(\R^n)$, such that $supp(f_i)\subseteq \R^n_+$. Write $\R^n=\R^{n-1}\times\R$ and $\R^n_+=\R^{n-1}\times (0,\infty)$.
\begin{defn}[Reflection positivity I]
We say that $\mu$ has \textsf{reflection positivity} if for all $z_1,...,z_k\in\mathbb{C}$ we have $\sum_{i,j}\bar z_i\widehat{\mu}(f_i \cdot\text{im}(\theta)\cdot f_j)z_j\geq 0$, where $\theta(x,t)=(x,-t)$. for all $k\in\N$, $f_1,...,f_k\in\calS(\R^n)$ with $supp(f_j)\subseteq \R^n_+$.
\end{defn}

Assume that $\nu$ is Gaussian and let $C_\nu$ be its covariance. 
\begin{defn}[Reflection Positivity II]
We say that $C_\nu$ has \textsf{reflection positivity} if $$C_\nu(f,\theta f)=\iint_{\R^n\times\R^n}f(x,t)f(y,-s)C((x,t),(y,s))\dd x\dd y\dd t\dd s\geq 0$$ for all $f\in\calS(\R^n)$ such that $supp(f)\subseteq \R^n_+$.
\end{defn}

\begin{exe}
Let $\nu$ be a Gaussian measure on $\calS'(\R^)$. Then $\nu$ has reflection positivity if and only if $C_\nu$ has reflection positivity.
\end{exe}

\begin{ex}
Let $\mu$ be the Gaussian measure with covariance $(\Delta+m^2)^{-1}$. Then $\mu$ has reflection positivity.
\begin{proof}
We will show that $C(f,g)=\langle f,(\Delta+m^2)^{-1}g\rangle_{L^1(\R^n)}$ is reflection positive. Let $f\in \calS(\R^n)$ with $supp(f)\subseteq\R^n_+$, and $x=(\bar x,t)$. Then 
$$C(f,g)=\iint_{\R^n\times\R^n}f(x)C(x,y)g(y)\dd y=\frac{1}{(2\pi)^n}\iint_{\R^n\times\R^n}\left(\int_{\R^n}\frac{\ee^{\I\xi(x-y)}}{\xi^2+m^2}\dd\xi\right)g(y)\dd x\dd y=\int \frac{\widehat{f}(\xi)\widehat{g}(\xi)}{\xi^2+m^2}\dd\xi.$$
Moreover, we have $C(f,\theta f)=\int_{\R^n}\frac{\overline{\widehat{\theta}}(f(\xi))\widehat{f}(\xi)}{\xi^2+m^2}\dd\xi$, which we want to be positive. We have $\widehat{f}(\widetilde{\xi},\I\xi_n)=\frac{1}{(2\pi)^{n/2}}\int_0^\infty\left(\int_{\R^{n-1}}f(x,t)\ee^{\I\xi\cdot x-\xi_n\cdot t}\right)\dd t$. Similarly $\widehat{\theta f}(\widetilde{\xi},\I\xi_n)=\frac{1}{(2\pi)^{n/2}}\int_0^\infty\left(\int_{\R^{n-1}}\overline{f(x,t)}\ee^{-I\xi\cdot x-\xi_n\cdot t}\right)\dd t$. Thus we get $$\widehat{f}(\widetilde{\xi},\I\xi_n)=\overline{\widehat{\theta f}(\widetilde{\xi},\I\xi_n)}.$$
Using these relations, we can show 
$$C(f,\theta f)=\int_{\R^{n-1}}\frac{\vert \widehat{f}(\widetilde{\xi},\I\mu(\widetilde{\xi}))\vert^2}{\sqrt{m^2+\widetilde{\xi}^2}}\dd\widetilde{\xi}\geq 0,$$
where $\mu(\widetilde{\xi})=\sqrt{m^2+\widetilde{\xi}^2}$.
\end{proof}
\end{ex}

Considering $\calE=L^2(\calS'(\R^n),\nu)$, we assume that $\nu$ has analyticity and Euclidean invariance. Consider the set 
\[
\calA=\left\{ A(\phi)=\sum_{j=1}^k c_j\ee^{\I\phi(f_j)}\Big|c_j\in\mathbb{C},k\in\N\right\}
\]
In fact, $\calA\subseteq \calE$, because of analyticity, and $\calA$ is an algebra. Moreover, let $\calE_+=\{A(\phi)\in \calA\mid supp(f_i)\subseteq \R^n_+\}\subseteq \calE$. We define a bilinear form $b$ on $\calE_+$ by $b(A,B)=\int \overline{\theta A}B\dd\nu(\phi)$.

\begin{exe}
The measure $\nu$ is reflection positive if and only if for all $A\in\calE_+$, $b(A,A)\geq 0$. Let $N=\{A\in\calE_+\mid b(A,A)=0\}$, and let $\calH$ be the completion of $\calE_+/N$.
\end{exe}

\begin{defn}
$\calH$ is called the physical Hilbert space. 
\end{defn}

We can observe that if we have $T\colon \calE\to\calE$ such that $T(\calE_+)\subseteq \calE_+$ and $T(N)\subseteq N$, then $T$ induces a map $T(t)\colon \calH\to\calH$, where $T(t)(\vec{x},s)=(\vec{x},s+t)$ for $t\geq 0$. We know that $T(t)$ acts on $\calE$ unitarly. 
\begin{lem}
We have $T(t)\calE_+\subseteq \calE_+$ and $T(t)N\subseteq N$.
\end{lem}

\begin{proof}
The first part is obvious. For the second part, observe that $\theta\circ T(t)=T(-t)\circ \theta$. Let $A\in N$. Then 
\begin{align*}
\langle T(t)A,\theta T(t)A\rangle_\calE&=\langle T(t)A,T(-t)\theta A\rangle_\calE=\langle T(2t)A,\theta, \theta A\rangle_\calE\\
&=b(T(2t)A,A)\leq \underbrace{b(A,A)^{1/2}}_{=0}b(T(2t)A,T(2t)A)^{1/2}=0,
\end{align*}
which iplies that $T(t)A\in N$ because of reflection positivity.
\end{proof}
One can also check that the map $T(t)^\land\colon\mathcal{H}\to\calH$ is a semigroup for $t\geq 0$.

\begin{lem}
We have $\|T(t)\|_\calH\leq 1$, for $t\geq 0$. Moreover, $t\mapsto T(t)$ is strongly continuous.
\end{lem}
\begin{cor}
$T(t)=\ee^{-t H}$, where $H$ is a positive self adjoint operator on $\calH$. Moreover, $H(1)=0$.
\end{cor}

\begin{ex}[Free massive scalar field theory]
Consider a measure $\mu$ and the Green's functions $C(x,y)$. We want to know whether we can find an ``explicit'' representation of $\calH$ in terms of time zero hypersurfacesin $\R^{n-1}$. We can indeed write $\calH\cong L^2(\calS'(\R^{n-1}),\nu)\subseteq \Gamma(H^{-1}(\R^n))=L^2(\calS'(\R^n),\mu)$, where $\nu$ is a Gaussian measure.
\end{ex}

Let $f\in\calS(\R^{n-1})$. Define then $j_0 f=f\otimes \delta_0$, where $(f\otimes \delta_0)(\vec{x},t)=f(\vec{x})\delta_0(t)$. We claim that $f\otimes\delta_0\in H^{-1}(\R^n_+)$. Indeed, we have 
$\widehat{f\otimes\delta_0}(\vec{\xi},\xi_n)=\widehat{f}(\vec{\xi})$, and we know 
\begin{multline*}
\langle f\otimes \delta_0,C(f\otimes\delta_0)\rangle=\frac{1}{2\pi}\int_{\R^n}\frac{\vert \widehat{f}(\vec{\xi})\vert^2}{\xi^2+m^2}\dd\xi\\=\frac{1}{2\pi}\int_{\R^{n-1}}\vert\widehat{f}(\vec{\xi})\vert^2\left(\int_\R\frac{1}{\xi^2+m^2}\dd\xi_n\right)\dd\vec{\xi}=\frac{1}{2}\int_{\R^{n-1}} \frac{\vert\widehat{f}(\vec{\xi})\vert^2}{\sqrt{\vec{\xi}^2+m^2}}\dd\vec{\xi}
\end{multline*}
Thus $f\otimes\delta_0\in H^{-1}(\R^n)$. Moreover, $\langle f\otimes \delta_0,C(f\otimes\delta_0)\rangle_{L^2(\R^n)}=\frac{1}{2}\left\langle f,\left(\sqrt{\Delta_{\R^{n-1}}+m^2}\right)^{-1}f\right\rangle_{L^2(\R^{n-1})}$. If we define $B(f,g)=\frac{1}{2}\int_{\R^{n-1}}f\left(\Delta_{\R^{n-1}}+m^2\right)^{-1/2}g \dd x$, we can see that $j_0$ defines an isometry $K_{B(\nu)}\to H^{-1}(\R^n)$, where $K_{B(\nu)}$ is the completion of $\calS(\R^n)$ with respect to $B$.

\begin{lem}
For $t\in\R$, define $(j_tf)=f\otimes\delta_t$, with $f\in\calS(\R^{n-1})$. Then for $t\geq s$, $\langle j_t f,j_sg\rangle_{L_2(\R^n)}=\frac{1}{2}\left\langle f,\left(\Delta_{\R^n}+m^2\right)^{-1/2}\ee^{-(t-s)\sqrt{\Delta+m^2}}g\right\rangle_{L^2(\R^n)}$.
\end{lem}
Let $\nu$ be the Gaussian measure on $\calS'(\R^{n-1})$ whose covariance is $B$. Denote by $H^{-1/2}(\R^{n-1}):=K_{B(\nu)}$. Then we know $L^2(\calS'(\R^{n-1}),\nu)\cong \Gamma(H^{-1/2}(\R^{n-1}))$. Given an operator $A$ on $\calH$, one can define an operator $\dd\Gamma(A)$ on $\Gamma(\calH)$ as follows: on $Sym^n(\calH)$ we get $\dd\Gamma(A)=A\otimes I\otimes\dotsm\otimes I+I\otimes A\otimes\dotsm \otimes I+...$, and on $Sym^0(\calH)=\mathbb{C}$ we get $\dd\Gamma(A)=0$. If we identify $\calH$ with $L^2(\calS'(\R^{n-1}),\nu)$ or $\Gamma(H^{-1/2}(\R^{n-1}))$, then $$\dd\Gamma(\sqrt{\Delta_{\R^n}+m^2}).$$

\section{QFT as operator valued distribution}

The motivation of this section is to get a better understanding of relativistic quantum meachanics. Recall the data for a quantum mechanical system:
\begin{itemize}
\item{Hilbert space of states $\calH$ (e.g. $L^2(\R^n)$)
}
\item{Obsrevables, which are represented by self adjoint operatos on $\calH$, 
}
\item{
``symmetries'', which are unitary representations on $\calH$, and $1$-parameter group of symmetries, leading to specific observables (e.g. time translation $\rightsquigarrow$ Hamiltonian of the system).
}
\item{Dynamics is controlled by the Schr\"odinger equation $\I\hbar\frac{\partial\psi}{\partial t}=\widehat{H}\psi$.
}
\end{itemize}

\subsection{Relativistic quantum mechanics}
In relativistic quantum mechanics we want to have unitary representation of the \textsf{Poincar\'e group} $\calP$, which is the group of all ``space-time'' symmetries. Recall that Minkowski space-time is given by $\mathbb{M}^n=\R^{1,n-1}$, where we can have coordinates in position space (such as ($t,\vec{x}$)) or in momentum space (such as ($\xi_0,\vec{\xi}$)). Denote by $\calL$ the \textsf{Lorentz group}, which is the set of all linear isometries of $\mathbb{M}^n$, i.e. $\{(\Lambda_{ij})\mid \Lambda^T g\Lambda =g\}$, thus for $\Lambda\in \calL$ we have $\det\Lambda\in \{\pm1\}$. Moreover, we can write $\calL$ as a union of subspaces:
$$\calL=\calL^{\uparrow}_+\cup\calL_1^\uparrow\cup\calL^\downarrow_+\cup\calL^\downarrow_-,$$
where the label $\uparrow$ ($\downarrow$) means the determinant is $+1$ ($-1$), and the label $+$ ($-$) means $\Lambda_{00}>0$ ($<0$). Note that $I\in \calL_+^\uparrow$, which we call the restricted Lorentz group. We define the Poincar\'e group by 
\[
\calP=\{ T(\Lambda,a)\mid \Lambda\in \calL,a\in\R^n\},
\]
where $T(\Lambda,a)(x)=\Lambda x+a$. Thus $\calP=\calL\ltimes \R^n$. We can write $\calP$ as the union of subspaces in the same way as for $\calL$. We call $\calP^\uparrow_+$ the restricted Poincar\'e group. We want to have a projective unitary representation of $\calP^\uparrow_+$. 

\subsubsection{Bergmann's construction ($n=4$)}
In this construction, the Projective unitary representation of $\calP^\uparrow_+$ come from unitary representation of $\widetilde{\calP}^\uparrow_+$, which is the universal cover of $\calP^\uparrow_+$. In fact $SL(2,\mathbb{C})$ is the universal cover of $\calL^\uparrow_+$. Hence, in this case $G=\widetilde{\calP}^\uparrow_+=SL(2,\mathbb{C})\ltimes\R^4$.

\subsubsection{Wigner's construction}
Take $p\in\R^4$, and let $H_p$ be the stabilizer of $p$ by the action of $SL(2,\mathbb{C})$. Moreover, take a unitary irreducible representation $\calH_{\sigma,p}$ of $H_p$. One can then use the Mackey machine: choose a $G$ invariant measure on $G/H_p$ and define the Hilbet space $\calH$ to be the $\calH_{\sigma,p}$ valued functions on $G/H_p$ anf use the invariant measure to define an inner product.

\begin{prop}[Wigner]
$\calH$ is an irreducible unitary representation of $G$. Moreover, all irreducible unitary representations of $G$ arise this way.
\end{prop}

\begin{rem}
$H_p$ can have $(2s+1)$-dimensional irreducible representations. Here $s\in\{0,\frac{1}{2},1\}$ represents the ``spin'' of the particle.
\end{rem}

Assume $s=0$. We start with the trivial representation, which is a $1$-dimensional representation of $H_p$. Consider the sets
\begin{align*}
X_m^+&=\{\xi^1-m^2=0\mid \xi_0>0\}\\
X_m^-&=\{\xi^1-m^2=0\mid \xi_0<0\}
\end{align*}
and write $X_m=X_m^+\cup X_m^-$ and $X=\bigcup_{m\geq 0}X_m$. Take $p=(m,0,0,0)$. Then we have $G/H_p=X_m^+$. We want to construct an invariant measure on $X_m^+$. Let $f$ be a positive function on $(0,\infty)$. Then $f(\xi^2)\dd\xi$ is an invariant measure on $X$. We would like to have an invariant measure of the form $\delta(\xi^1-m^2)\dd\xi$. We define $\phi\colon (0,\infty)\times \R^3\to X^+_m$, $(y,\vec{\xi})\mapsto (\sqrt{y+\vert\vec{\xi\vert^2,\vec{\xi}}})$. Then $\phi^*(f(\xi^2)\dd\xi) =\frac{f(y)\dd y\dd\vec{\xi}}{\sqrt{y^2+\vert\vec{\xi}\vert^2}}$. We want to have the pushforward of $\delta_m=\frac{\dd y\dd\vec{\xi}}{\sqrt{y^1+\vec{\xi}\vert^2}}$ to be our measure on $X_m^+$. More precisely, define $\alpha\colon \R^3\to X_m^+$ by $\alpha(\vec{\xi})=(\sqrt{m^2+\vert\vec{\xi}\vert^2},\vec{\xi})$. As an invariant measure, we want the pushforward of $\frac{1}{2\sqrt{m^2+\vert\vec{\xi}\vert^2}}\dd\vec{\xi}$ on $\R^3$ to $X_m^+$. Wigner's theorem gives us $\calH=L^2(X_m^+,\mu_m)\cong L^2(\R^3,\nu)$, where $\frac{\dd\nu}{\dd\vec{\xi}}=\frac{1}{2\sqrt{m^2+\vert\vec{\xi}\vert^2}}$. The position operator is then given by $\frac{1}{2}(\Delta+m^2)^{-1/2}$. One can summarize the result by saying that the Hilbert space for spin zero particles is given by $L^2(\R^3,\nu)\cong H^{-1/2}(\R^3)$.

\subsection{Garding-Wightman formulation of QFT}
We want to give the axioms of the so-called \textsf{Garding-Wightman} formulation of QFT. We have the following axioms:
\begin{enumerate}[(GW1)]
\item{We have a Hilbet space $\calH$ , a vaccum state $\Omega\in\calH$, and a unitary representation $\calP_+^\uparrow$ on $\calH$.
}
\item{We have a field operator $\Phi\colon \calS(\R^+)\to \text{Operators on $\calH$}$ together with a dense subspace $D$ of $\calH$ such that
\begin{enumerate}
\item{$\Omega\in D$}
\item{$D\subseteq D(\Phi(f))$ for all $f$,}
\item{$f\mapsto \Phi(f)\mid_D$ is linear,}
\item{for all $\Omega_1,\Omega_2\in D$, the assignment $f\mapsto \langle \Phi(f)\Omega_1,\Omega_2\rangle$ is a Schwarz distribution (regularity),}
\item{$\Phi(f)^*=\Phi(\bar f)$.}
\end{enumerate}
}
\item{(Covariance) We have 
\begin{enumerate}
\item{$U(a,\Lambda)D\subseteq D$, where $U(a,\Lambda)$ is the unitary representation of $T(a,\Lambda)\in \calP_+^\uparrow$ on $\calH$,
}
\item{$U(a,\Lambda)\cdot \Phi(f)\cdot U(a,\Lambda)^{-1}=\Phi(T(a,\Lambda)f)$.
}
\end{enumerate}
}
\item{(Spectrum)
Since $\R^4$ acts unitarly on $\calH$ via $U(a,\Lambda)$, we can take $P_1,...,P_4$ to be the infinitesimal generators of this action. One can show that $P_1,...,P_4$ are essentially self adjoint.
The axio is then given by: The joint spectrum of $(P_1,...,P_4)$ lies in $X^+=\{\xi^2\geq 0\mid \xi_0>0\}$, where physically $\xi^2=E^2-\vec{p}^2$.
}
\item{(Locality)
If $f$ and $g$ have space-like disjoint support, then $[\Phi(f),\Phi(g)]=0$.
}
\end{enumerate}

\begin{rem}
By the axioms, one can show that the vaccum is unique: If $U(a,\Lambda)\Omega'=\Omega'$ for all $T(a,\Lambda)\in\calP_+^\uparrow$, then $\Omega'=c\Omega$, where $c\in\mathbb{C}$.
\end{rem}
Given $f_1,...,f_k\in\calS(\R^4)$, we can define $(f_1,...,f_k)\mapsto \langle\Phi(f_i)\dotsm \Phi(f_k)\Omega,\Omega\rangle$. By (GW2) this assignment is a distribution in $\calS'(\R^4)$, i.e. 
$$\langle\Phi(f_1)\dotsm \Phi(f_k)\Omega,\Omega\rangle=W_k(f_1\otimes\dotsm \otimes f_k)=\int_{\R^4}W_k(x_1,...,x_k)f_1(x_1)\dotsm f_k(x_k)\dd x_1\dotsm \dd x_k.$$
\begin{defn}[Wightman distribution]
$W_k(x_1,...,x_k)$ are called \textsf{Wightman distribution}.
\end{defn}

We can now formulate the Wightman axioms:

\begin{enumerate}[(W1)]
\item{$W_k$ are $\calP_+^\uparrow$-invariant,}
\item{If $f_1\in\calS(\R^4),...,f_k\in\calS(\R^{4k})$, then $\sum_{i,j=0}^kW_{i+j}(\bar f_i\otimes f_j)\geq 0$.
}
\item{(Locality)
$W_k(x_1,...,x_j,x_{j+1},...,x_k)=W_k(x_1,...,x_{j+1},x_j,...,x_k)$, whenever $x_j$ and $x_{j+1}$ are space-like seperated.

We recall the Euclidean setting: We have a measure $\ee^{-S(\phi)}\mathscr{D}(\phi)$ on $\calS'(\R^4)$, a two point function $C(f,g)$ and for $f_1,...,f_k\in\calS(\R^4)$ we have a map $(f_1,...,f_k)\mapsto \int \phi(f_1)\dotsm \phi(f_k)\ee^{-S(\phi)}\mathscr{D}(\phi)$. We would like to know how we can relate the Minkowski to the Euclidean setting.

We can observe that $W_k(x_1,...,x_k)=\omega_k(x_1-x_2,...,x_{k-1}-x_k)$ because of translation invariance.
}
\item{(Spectral condition)
The Fourier transform of $\widehat{\omega}_k$ of $\omega_k$ has support in $\underbrace{X^+\times\dotsm\times X^+}_{k}$. 
}
\end{enumerate}

\begin{rem}
There is one more Wightman axiom, called ``Cluster property'' (W6), which is related to uniqueness of vaccum.
\end{rem}

\begin{thm}[Wightman reconstruction theorem]
If we have distributions $(W_k)_k$ satisfying (W1)-(W6), then there is a ``unique'' GW field theory, whose Wightman distributions are $W_k$. 
\end{thm}

\begin{proof}[Sketch of the proof]
Let $f\in V=\bigoplus_{k\geq 0}\calS(\R^4)$ (for $k=0$ we get $V=\mathbb{C}$), such that $f=(f_0,f_1,...,f_j,...)$, where everything is zero except for finitely many $j$s. Moreover, let $(f,g)=\sum_{i,j}W_{i+j}(\bar f_i\otimes g_j)$ and $N=\{f\in V\mid (f,f)=0\}$. Let $\calH$ be the completion of $V/N$ and $\Omega=(1,0,0,...)$, and for $h\in\calS(\R^4)$, we have $\Phi\colon V\to V$, $\Phi(h)(f_0,f_1,...)=(0,f_0\otimes h,f_1\otimes h,...)$.
\end{proof}

\subsubsection{Wick rotation}
Consider the distribution $W_k(x_1,...,x_k)$ for $x_1,..,x_k\in\mathbb{M}^4$, where $x_1=(t_1,\vec{x}_1),...,x_k=(t_k,\vec{x}_k)$ such that $x_i^2=t_i^2-\vec{x}_i^2$. Formally, we want to define $W_k(\I t_1,\vec{x}_1,\I t_2,\vec{x}_2,...,\I t_k,\vec{x}_k)$, with $x_i^2=-i_i^2-\vec{x}_i^2$. We want to do this by considering complex variables $z_i=x_i+\I y_i$ and pass from $W_k(x_1,...,x_k)$ to $W_k(z_1,...,z_k)$. First, we canalytically continue $W_k(x_1,...,x_k)$ to a holomorphic function. Next, we think of $W_k(x_1,...,x_k)$ as a boundary value of an analytic function.

\begin{defn}[Boundary value]
Let $\phi$ be a distribution in $\calS'(\R^n)$. Let $F$ be a holomorphic function. We say $\phi$ is a \textsf{boundary value} of $F$ if for fixed $y_0\in\R^n$, we have 
\[
\phi(f)=\lim_{t\to 0}\int_{\R^n}F(x+\I ty_0)f(x)\dd x,
\]
or equivalently we say $F$ is an analytic continuation of $\phi$.
\end{defn}

\begin{rem}
It is not clear whether all $\phi\in\calS'(\R^n)$ have analytic continuations. In fact, let $T\in\calS'(\R^n)$ such that $supp(T)\subseteq$ some cone $C$, where $C$ is the intersection of two hyperplanes. Then $T$ is a boundary value of an analytic function on $\R^n-\I C^*$, where $C^*$ is the dual cone.
\end{rem}

\begin{cor}
Recall $\widehat{\omega}_k$ has support in $X^+\times\dotsm\times X^+$. Then $\omega_k$ can be analytically continued to a holomorphic function $\omega_k(z_1,...,z_k)$ on $\underbrace{(\R^4-\I X^+)\times\dotsm\times (\R^4-\I X^+)}_{k-1}$.
\end{cor}

Now we can observe that the $W_k$ have analytic continuation to $W_k(z_1,...,z_k)$ on $\calT_k=\{(z_1,...,z_k)\mid \text{Im}(z_{i+1}-z_i)\in X^+\}$.

\subsubsection{Schwinger functions}

We want to construct $W_k(\I t_1,\vec{x}_1,...,\I t_k,\vec{x}_k)$ having $W_k(x_1,...,x_k)$ with $x_j=(t_j,\vec{x}_j)$. The problem is that all the points $(\I t_1,\vec{x}_i),...,(\I t_k,\vec{x}_k)\not\in \calT_k$.
\begin{exe}
Take $k=1$ and show $(\I t,\I)\in\calT_1$ if and only if $t_1>0$.
\end{exe}

Hence, we want to enlarge $\calT_k$ and extend $W_k$ to this bigger set. Take 
\[
\calT_k^{e}=\{(\Lambda z_1,...,\Lambda z_k)\mid \Lambda\in\calL, \det \Lambda=1\}.
\]
E.g. we had $\Lambda=-I$ before. Then we can extend $W_k(w_1,...,w_k)=W_k(z_1,...,z_k)$ if $(w_1,...,w_k)=(\Lambda z_1,...,\Lambda z_k)$ for some $\Lambda\in\calL$.

\begin{lem}
It is possible to extend $W_k$ as before (using a lot of assumptions). 
\end{lem}

Denote by $$\Sigma_k(\calT_k^{e})=\{(z_{\sigma(1)},...,z_{\sigma(n)})\mid \sigma\in\Sigma_k, (z_1,...,z_k)\in\calT_k^{e}\}$$
the permutation group of order $k$ on $\calT_k^{e}$.

\begin{defn}
$\calT_k^{p,e}:=\Sigma_n(\calT_k^{e})$.
\end{defn}

\begin{exe}
Show that the Euclidean points $\calE_n\subseteq \calT_n^{p,e}$, where $\calE_n\subseteq \mathbb{C}$ and $(z_1,...,z_n)\in\calE_n$ if and only if $z_j=(\I t_j,\vec{x}_j)$.
\end{exe}

\begin{rem}
In fact, $W_n(z_1,...,z_n)$ can be extended to an analytic function on $\calT_n^{p,e}$ (technical result).
\end{rem}

\begin{defn}
$$\widetilde{\calE}_n=\{(\underbrace{(t_1,\vec{x}_1)}_{y_1},...,\underbrace{(t_n,\vec{x}_n)}_{y_n})\mid (\I t_1,\vec{x}_1,...,\I t_n,\vec{x}_n)\in\calE_n\}.$$
\end{defn}

We call $(y_1,...,y_n)$ non-coincident if $y_k\not=y_\ell$ for all $k\not=\ell$. 

\begin{defn}[Schwinger function]
For a non.coincident Euclidean point $(y_1,...,y_n)$, we define $S_n(y_1,...,y_n)=W_n(\I t_1,\vec{x}_1,...,\I t_n,\vec{x}_n)$. We call $S_n$ \textsf{Schwinger functions}.
\end{defn}

\subsubsection{Properties of Schwinger functions}

For a free massive scalar field theory we can compute $W_2(x,y)$ explicitely. It is given by 
\[
W_2(x,y)=C_{W_2}\int_{\R^3}\frac{\ee^{-\I(\omega(\vec{\xi})(x_0-y_0)+\vec{\xi}(\vec{x}-\vec{y}))}}{\omega(\vec{\xi})}\dd\vec{\xi},
\]
where $\omega(\vec{\xi}):=\sqrt{m^2+\vec{\xi}^2}$ and $C_{W_2}$ some constant. Let $$W_2(t,\vec{x})=C_{W_2}\int_{\R^3}\frac{\ee^{\I t(\omega(\vec{\xi})+\vec{\xi}\cdot\vec{x})}}{\omega(\vec{\xi})}\dd\vec{\xi}.$$ Thus we get 
\begin{multline*}
S_2(y)=W_2(\I t,\vec{x})=C_{W_2}\int_{\R^3}\frac{\ee^{t\omega(\vec{\xi})\ee^{-\I \vec{x}\cdot\vec{x}}}}{\omega(\vec{\xi})}\dd\vec{\xi}\\=C_{W_2}\int_{\R^3}\ee^{-\I\vec{\xi}\cdot\vec{x}}\dd\vec{\xi}\int_0^\infty\frac{\ee^{-\I t\xi_0}}{\xi_0^2+m^2+\vec{\xi}^2}\dd\xi_0=C_{W_2}\int_{\R^4}\frac{\ee^{-\I y\xi}}{m^2+\xi^2}\dd\xi=G(y),
\end{multline*}
where $G$ is the Green's function. The Osterwalder-Schrader axioms can be reformulated with the Schwinger functions as follows:
\begin{enumerate}[(OS1)]
\item{$S(y_1,...,y_n)$ defines a distribution on $\calS_{\not=}(\R^{4n})$, where 
\[
\calS_{\not=}(\R^{4n})=\{f\in\calS(\R^{4n})\mid f(y_i-y_j)=0,\forall u\not=j\}.
\]
Moreover, $\overline{S_n(f)}=S_n(\overline{\theta f})$, where $\theta(t,\vec{x})=(-t,\vec{x})$, and if $h\in\calS((\R^4_+)^{n-1})$ we get 
\[
\vert S_n(h)(y_2-y_1,...,y_n-y_{n-1})\vert\leq \| h\|,
\]
where $\|\cdot\|$ is some norm on $\calS((\R^4_+)^{n-1})$.
}
\item{(Euclidean invariance)
}
\item{Let $f_n\in\calS((\R^4_+)^n)$. Then $\sum_{m,n}S_{n+m}(\overline{\theta f_n}\otimes f_m)\geq 0$.
}
\item{$S_n(y_{\sigma(1)},...,y_{\sigma(n)})=S_n(y_1,...,y_n)$ for all $\sigma\in \Sigma_n$.
}
\item{(Cluster property)
}
\end{enumerate}

\begin{thm}[Reconstruction theorem]
If we have $S_n(y_1,y_n)$ satisfying (OS1)-(OS5), then there is an unique Garding-Wightman theory.
\end{thm}

\end{document}